\documentclass[12pt]{article}
\usepackage[left=1in,right=1in,top=1in,bottom=1in]{geometry}
\usepackage{setspace}
\usepackage{pdflscape}
\usepackage{booktabs}
\usepackage{soul,xcolor}
\definecolor{seagreen}{RGB}{46, 139, 87}
\definecolor{coral}{RGB}{234,112,112}

\usepackage[reqno]{amsmath}
\usepackage{mathtools}
\usepackage{hyperref}
\definecolor{darkblue}{rgb}{0.0, 0.0, 0.55}
\definecolor{crimsonred}{rgb}{0.54, 0.08, 0.13}
\hypersetup{
    colorlinks=true,
    linkcolor=crimsonred,
    citecolor=darkblue,
    urlcolor=crimsonred
}
\usepackage{amsthm}
\usepackage{amssymb,enumerate}
\usepackage{tikz}
\usetikzlibrary{fit,positioning,shapes,decorations.pathmorphing,decorations.pathreplacing,arrows}

\newcommand\E{\mathbb{E}}

\usepackage{titletoc}
\usepackage{framed}
\usepackage{caption}
\newtheorem{theorem}{Theorem}
\newtheorem{assumption}{Assumption}
\newtheorem{definition}{Definition}
\newtheorem{proposition}{Proposition}
\newtheorem{corollary}{Corollary}
\usepackage[authoryear,sort]{natbib}
\usepackage{bibunits}

\title{\singlespacing Let Time Tell: Identification and Gaussian Process Estimation for Interrupted Time Series \thanks{The author thanks Jeffrey Lewis and Chad Hazlett for their invaluable advice and support throughout this project. The author is also grateful to Kosuke Imai, Ian Lundberg, Andrew Bertoli, Chris Felton, Ryan Baxter-King for helpful comments and discussions relevant to this project, and Jack Kappelman for both discussion and sharing National Instant Criminal Background Check System (NICS) data. This research benefited from feedback in presentations at the PolMeth 2024. This research was supported in part by a UCLA Dissertation Year Fellowship. The content is solely the responsibility of the author.}}

\author{Soonhong Cho\footnote{Department of Political Science, University of California, Los Angeles. Email: \href{mailto:soonhongcho@ucla.edu}{soonhongcho@ucla.edu}, URL: \href{https://soonhong-cho.github.io/}{https://soonhong-cho.github.io/} }}

\begin{document}

\maketitle
\thispagestyle{empty}
\vspace{-.5in}

\begin{abstract}
We study causal inference in interrupted time series designs where a treatment affects every unit simultaneously, so that the contemporaneous controls used by difference-in-differences and synthetic control are unavailable and the counterfactual must be extrapolated from a unit's own pre-treatment history. We establish identification within the potential outcomes framework and estimate the counterfactual by Gaussian process regression. Rather than committing to a single best-fitting trend, the estimator retains the functions consistent with the pre-treatment series and widens its intervals where extrapolation magnifies their divergence. Connecting it to reproducing kernel Hilbert space theory, we derive a bias decomposition that isolates the component extrapolation inflates and a worst-case bound on that component, justifying the Gaussian process estimator's posterior variance as extrapolation-aware uncertainty quantification. In closed form, the band equals the worst-case divergence the model class permits among functions consistent with the pre-treatment data. The method is illustrated with calibrated simulations and an analysis of handgun purchases after the Supreme Court's \textit{Heller} decision, a universal treatment whose practical effect concentrates in a single jurisdiction. An \texttt{R} package, \texttt{gpss}, implements the approach.
\end{abstract}

\newpage
\setcounter{page}{1}

\begin{bibunit}[apalike]

\section{Introduction}

A recurring challenge for causal inference is to estimate the effect of an event that reaches an entire population at once. Supreme Court decisions transform legal landscapes nationwide overnight, with \textit{Citizens United} reshaping campaign finance and \textit{Heller} altering the national interpretation of individual gun rights. Wars engulf entire societies and federal policies bind every jurisdiction simultaneously. The COVID-19 pandemic is the starkest case of all, a global shock that altered behaviour and institutions everywhere. Such events are \textit{universal treatments}, single interventions that reach every unit simultaneously, and they pose a distinct challenge for causal inference. With no units left untreated, comparing treated to control units within the same post-treatment period is infeasible, and the practical effect is often concentrated: the precedent in \textit{Heller} binds every state, yet the ruling reaches only the handgun ban in Washington, D.C., so its immediate consequences fall on a single jurisdiction.

Standard identification strategies rely on comparisons that are unavailable here. Difference-in-differences (DiD) requires untreated comparison units \citep{card1993minimum}, often with variation in treatment timing \citep{callaway2021difference, sun2021estimating}, and the synthetic control (SC) method requires a donor pool that remains untreated \citep{abadie2003economic, abadie2010synthetic}. Researchers must instead turn to temporal variation within treated units, asking not how treated units differ from controls but how outcomes deviate from what would have occurred absent treatment.

One natural strategy for this setting is the interrupted time series (ITS) design \citep{campbell1963experimental, box1975intervention}. Pre-treatment patterns encode structural relationships, such as economic, institutional, and behavioural regularities, that would persist absent the intervention; an ITS learns these patterns and projects them forward to form the counterfactual trajectory, whose contrast with the realised outcomes measures the effect. Identification rests on a stability condition, that the conditional expectation of the untreated outcome given covariates remains the same function before and after the treatment onset. The condition can fail under concurrent shocks or structural breaks, but it states a clear requirement, that no force other than the intervention alters how the trajectory would otherwise have evolved.

Existing ITS approaches carry two limitations. First, they impose restrictive functional forms. Segmented regression, the workhorse of ITS analysis, fits piecewise linear or polynomial models whose coefficients recover the level/slope change attributed to treatment \citep{bernal2017interrupted}, risking substantial bias when the true process exhibits nonlinearities beyond the support of the observed data. Second, and more importantly, they are not designed to quantify \textit{extrapolation uncertainty}. Many functions in a given class can fit the pre-treatment record equally well yet diverge at a post-treatment target; the worst-case divergence is a property of the class itself and grows with temporal distance. Standard methods commit to a single specification and condition the interval on that choice, so the interval does not widen as the counterfactual extends beyond the data.

We propose Gaussian process (GP) regression for the counterfactual estimation, a nonparametric method that carries uncertainty over the whole function rather than fixing a single fit \citep{rasmussen2006gaussian}. The kernel sets the available function space, and conditioning on the pre-treatment series concentrates it on the compatible functions. For the kernels of Section~\ref{sec:estimation}, the predictive variance interval widens with extrapolation distance, so the reported uncertainty reflects how far the counterfactual reaches beyond the training data \citep{cho2026inference}.

This article's central contribution is an ITS estimator with extrapolation-aware uncertainty quantification. The GP posterior mean estimates the counterfactual, and a reproducing kernel Hilbert space (RKHS) reading of the posterior variance justifies it as the measure of uncertainty about that estimate. A bias decomposition separates the estimation error into identification, approximation, learning, and noise components and names the assumption or condition that controls each. A worst-case bound then shows that the posterior variance controls the learning error over functions within a fixed complexity budget, and the bound is not loose: a least favourable counterfactual attains it, no estimator using the same data improves on it, and the band carries the extrapolation uncertainty of the model class itself, in closed form.

Two further contributions support the central one. We establish identification of the ITS estimand within the potential outcomes framework \citep{neyman1923application, rubin1974estimating}, resting on a conditional-mean stability condition weaker than conditional independence, and we propose placebo checks and supplementary diagnostics that probe the assumptions.

The rest of the paper proceeds as follows. Section~\ref{sec:identification} formalises identification, Section~\ref{sec:estimation} develops the estimator and placebo checks, Section~\ref{sec:theory} establishes the bias decomposition and worst-case bounds through RKHS theory. Section~\ref{sec:simulation} examines finite-sample properties of the proposed estimator in calibrated simulations, Section~\ref{sec:application} applies the method to handgun purchases after \textit{Heller} with further examples in Supplementary Material~\ref{app:additional_examples}, and Sections~\ref{sec:discussion}--\ref{sec:conclusion} situate the estimator among existing approaches and conclude.

\section{Identification in Interrupted Time Series Design}\label{sec:identification}

To gain some intuition that informs our approach, suppose we observe a sequence of untreated outcomes before an event of interest followed by treated outcomes. Without contemporaneous controls, a researcher predicts what would have happened absent treatment by learning from pre-treatment history. If pre-treatment outcomes follow a steady trend with a stable cycle, that pattern should continue absent treatment, and post-treatment deviations from the projection are attributed to the causal effect (Figure~\ref{fig:illustration}). This section formalises the strategy, showing that stability in the outcome--observable relationship substitutes for controls and stating the conditions under which pre-treatment patterns identify counterfactual outcomes.

\begin{figure}[tbh!]
    \centering
    \includegraphics[width=1\textwidth]{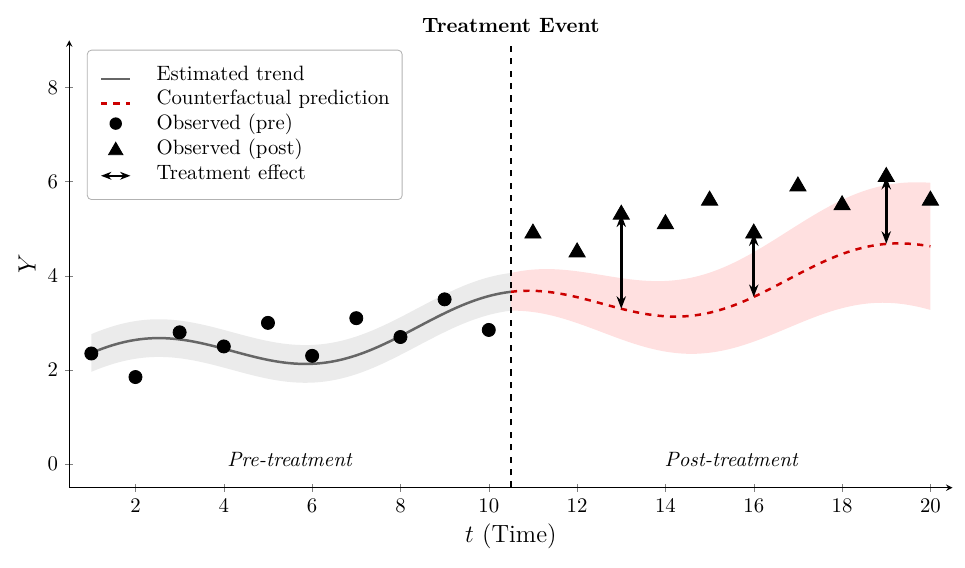}
    \caption{Illustration of Interrupted Time Series Design. \textit{Note:} A flexible trend (solid grey) fitted to pre-treatment outcomes (black dots) is extrapolated as the counterfactual (dashed red), with 95\% prediction intervals (shaded) widening with distance from the data. Treatment effects (black arrows) are the gaps between observed outcomes (black triangles) and the counterfactual; the vertical dashed line marks treatment onset.}    
    \noindent\textbf{Alt text:} Schematic time series plot with a fitted pre-treatment trend extended past the intervention as a dashed counterfactual, its shaded interval widening with distance, and arrows marking gaps to observed points.
    \label{fig:illustration}
\end{figure}

\paragraph{Setup.}

Consider a balanced panel of $N$ units observed over $T$ periods. For each unit $i$ and time $t$, let $Y_{it}$ denote the scalar observed outcome, $D_{it} \in \{0,1\}$ treatment status, and $\mathbf{X}_{it}$ a vector of observed covariates that includes time $t$ itself alongside any additional characteristics. Unobserved time-varying attributes that may affect outcomes are collected in $\mathbf{U}_{it}$. Under the potential outcomes framework \citep{neyman1923application, rubin1974estimating}, let $Y_{it}(1)$ and $Y_{it}(0)$ denote the potential outcomes under treatment and control, respectively. Treatment begins simultaneously for all units at time $t_0$ and persists thereafter: $D_{it} = \mathbf{1}(t \geq t_0)$.

We link observed to potential outcomes through the standard \textit{consistency} assumption:
\begin{equation}
    Y_{it} = D_{it}\,Y_{it}(1) + (1 - D_{it})\,Y_{it}(0). \label{eq:consistency}
\end{equation}
For consistency to deliver untreated outcomes before $t_0$, treatment must not act in advance of its onset (\textit{no anticipation}):
\begin{equation}
    Y_{it}(1) = Y_{it}(0) \quad \text{and} \quad 
    \mathbf{X}_{it}(1) = \mathbf{X}_{it}(0), 
    \quad t < t_0.
    \label{eq:no_anticipation}
\end{equation}
Equation~\eqref{eq:no_anticipation} pairs two conditions. No anticipation proper is the outcome equality $Y_{it}(1) = Y_{it}(0)$ for $t < t_0$, which with consistency gives $Y_{it} = Y_{it}(0)$ before $t_0$. Because treatment is simultaneous and universal, spillovers from treated to untreated units cannot arise by design; the causal quantities below are total effects inclusive of any interactions among treated units \citep{cerqua2023causal}. The counterfactual $Y_{it}(0)$ encodes the continuation of the pre-treatment regime, so expectations held at a stable level are absorbed into the conditional mean the method learns; what the equality rules out is anticipatory drift, the pattern the placebo checks of Section~\ref{subsec:placebo_check} are built to detect, and Supplementary Material~\ref{app:assumptions} assesses it in each application. At post-treatment periods the covariate equality $\mathbf{X}_{it^*}(1) = \mathbf{X}_{it^*}(0)$ is additionally required; it holds by construction when $\mathbf{X}_{it}$ consists of time and calendar indicators (e.g., quarter, month, etc.).

We partition at $t_0$: $Y_{i,\text{pre}} = (Y_{i1}, \ldots, Y_{i,t_0-1})^\top$ and $\mathbf{X}_{i,\text{pre}} = (\mathbf{X}_{i1}, \ldots, \mathbf{X}_{i,t_0-1})^\top$ are the pre-treatment outcomes and covariates, and $Y_{it^*}$, $\mathbf{X}_{it^*}$ denote the observed outcome and covariates at a post-treatment period $t^* \geq t_0$. In the estimation procedure of Section~\ref{sec:estimation}, $(\mathbf{X}_{i,\text{pre}}, Y_{i,\text{pre}})$ are training data and $\mathbf{X}_{it^*}$ the test inputs at which counterfactual predictions are formed.

\paragraph{Structural Model.}
For each unit $i$ at time $t$, the untreated potential outcome is governed by
\begin{align}
    Y_{it}(0) = h_i(\mathbf{X}_{it}, \mathbf{U}_{it}) + \varepsilon_{it}, 
    \quad E[\varepsilon_{it} \mid \mathbf{X}_{it}, \mathbf{U}_{it}] = 0,
    \label{eq:structural}
\end{align}
where $\mathbf{X}_{it}$ is the observed covariate vector, $\mathbf{U}_{it}$ collects unobservables, and $\varepsilon_{it}$ is idiosyncratic noise with variance $\sigma^2$. Any unobserved factor that moves the conditional mean of $Y_{it}(0)$ belongs to $\mathbf{U}_{it}$; $\varepsilon_{it}$ is the remainder, mean-zero given both arguments, so it contributes variance but does not bias the counterfactual. Measurement error in the outcome that is mean-zero given $(\mathbf{X}_{it}, \mathbf{U}_{it})$ joins $\varepsilon_{it}$; a systematic shift is a change in $\mathbf{U}_{it}$ (see Supplementary Material~\ref{app:assumptions}). The structural function $h_i: \mathcal{X} \times \mathcal{U} \to \mathbb{R}$ governs how observables and unobservables jointly determine the systematic component of outcomes. The function is unit-specific: two units at the same $(\mathbf{X}_{it}, \mathbf{U}_{it})$ may have different outcomes, reflecting heterogeneity in how units respond to observable and unobservable conditions. 

Any causal influence of pre-treatment history operates through $\mathbf{X}_{it}$ if observed or $\mathbf{U}_{it}$ if unobserved; past outcomes do not enter $h_i$ as arguments. Conditioning directly on lagged outcomes \citep{cerqua2023causal, brodersen2015inferring} requires recursive substitution of predicted values when forecasting several steps ahead, compounding error at each horizon, whereas the estimator of Section~\ref{sec:estimation} carries temporal dependence through a covariance kernel that assigns high covariance to nearby periods.

\paragraph{Target Function.}

The function $g_i : \mathcal{X} \to \mathbb{R}$ represents the best prediction of $Y_{it}(0)$ achievable from observables:
\begin{align}
    g_i(\mathbf{X}_{it}) \coloneqq E[Y_{it}(0) \mid \mathbf{X}_{it}] 
    = \int h_i(\mathbf{X}_{it}, \mathbf{u})\, 
    dP(\mathbf{U}_{it} = \mathbf{u} \mid \mathbf{X}_{it}),
    \label{eq:target}
\end{align}
where the integral averages $h_i$ over the conditional distribution of $\mathbf{U}_{it}$ given $\mathbf{X}_{it}$, not the marginal; the distinction matters because unobservables and observables can be correlated. Defined purely through conditional expectations, $g_i$ is the object a well-specified regression of $Y_{it}(0)$ on $\mathbf{X}_{it}$ would recover. The identification strategy maps $(Y_{i,\text{pre}}, \mathbf{X}_{i,\text{pre}}, \mathbf{X}_{it^*}) \mapsto g_i(\mathbf{X}_{it^*})$, recovering $g_i$ from the pre-treatment history and evaluating it at the post-treatment input; we write $g_i(\mathbf{X}_{it^*})$ throughout, suppressing its dependence on the pre-treatment data.

\paragraph{Key Identification Assumption.}

The structural function $h_i$ depends on unobservables $\mathbf{U}_{it}$, making direct estimation infeasible. Identification requires that \textit{within} each unit's data-generating process (DGP), $\mathbf{X}_{it}$ determine the conditional mean of $Y_{it}(0)$ without knowledge of $\mathbf{U}_{it}$. The unit-specific function $g_i$ is then recovered from unit $i$'s own pre-treatment data.

\begin{assumption}[Mean Sufficiency]
\label{assum:sufficiency}
For all units $i$ and all periods $t$: $$\mathbb{E}[Y_{it}(0) \mid \mathbf{X}_{it},\, \mathbf{U}_{it}] = \mathbb{E}[Y_{it}(0) \mid \mathbf{X}_{it}].$$
\end{assumption}
\noindent Adding and subtracting $g_i$ in the structural model~\eqref{eq:structural} writes the untreated outcome as an observable mean, an unobservable deviation, and noise,
\begin{align}
    Y_{it}(0) = \underbrace{g_i(\mathbf{X}_{it})}_{\text{observable mean}} + \underbrace{\bigl[h_i(\mathbf{X}_{it}, \mathbf{U}_{it}) - g_i(\mathbf{X}_{it})\bigr]}_{\text{unobservable deviation}} + \underbrace{\varepsilon_{it}}_{\text{noise}}.
    \label{eq:decomp_obs}
\end{align}
By construction of $g_i$, the deviation has mean zero given $\mathbf{X}_{it}$. Mean Sufficiency strengthens this to $h_i(\mathbf{X}_{it}, \mathbf{U}_{it}) = g_i(\mathbf{X}_{it})$ almost surely, so once $\mathbf{X}_{it}$ is known the particular value of $\mathbf{U}_{it}$ does not shift the expected outcome. The gap $h_i - g_i$ is the identification error that Section~\ref{sec:theory} isolates as Term~I. The assumption is not determinism, since outcomes still deviate from their conditional expectations through $\varepsilon_{it}$; it restricts only the \textit{functional} role of $\mathbf{U}_{it}$, which conditional on $\mathbf{X}_{it}$ must not carry additional information about the mean of $Y_{it}(0)$, while $\mathbf{U}_{it}$ and $\mathbf{X}_{it}$ may be arbitrarily dependent. The unit level is deliberate: each counterfactual is learned from that unit's own series, so every estimand rests on this condition, and a weaker requirement that biases merely average to zero would identify only $\bar{\tau}_{t^*}$.

This is a conditional-mean restriction, not a distributional one. It constrains only the first moment of $Y_{it}(0)$ given $(\mathbf{X}_{it}, \mathbf{U}_{it})$ and is weaker than the conditional-independence form of selection-on-observables \citep{hazlett2018trajectory}. Mean Sufficiency is the temporal counterpart of that cross-sectional condition: where selection-on-observables identifies a counterfactual by comparing units that share covariates, the ITS requires the within-unit relationship $g_i$ to be stable across time. Section~\ref{sec:discussion} situates the design relative to DiD and SC.

Under Assumption~\ref{assum:sufficiency}, the pre-treatment observations simplify to
\begin{align}
    Y_{it} = g_i(\mathbf{X}_{it}) + \varepsilon_{it}, 
    \quad E[\varepsilon_{it} \mid \mathbf{X}_{it}] = 0,
    \quad t < t_0,
    \label{eq:pretreatment_regression}
\end{align}
with the same idiosyncratic $\varepsilon_{it}$ as in~\eqref{eq:structural}. No part of $\mathbf{U}_{it}$ is folded into the error, so $g_i$ is the conditional expectation function (CEF) of a regression of $Y_{it}$ on $\mathbf{X}_{it}$. When $\mathbf{X}_{it}$ carries calendar time and period indicators, the covariate path is fixed by the design and the randomness in $Y_{it}(0)$ is due to the noise term $\varepsilon_{it}$ alone, the convention under which the bias in Section~\ref{sec:theory} is computed and the intervals are read.

\paragraph{Estimands.}

We define the causal estimands, beginning at the unit-period level. The unit- and period-specific treatment effect is $\delta_{it^*} = Y_{it^*}(1) - Y_{it^*}(0),$ the difference between the two potential outcomes for the same unit-period. The post-treatment period furnishes only $Y_{it^*}(1)$; no design recovers a unit's realised untreated outcome at a treated period without untenable assumption, so $\delta_{it^*}$ is not identified. We therefore target its conditional-mean counterpart,
\begin{align}
    \tau_{it^*} = Y_{it^*}(1) - g_i(\mathbf{X}_{it^*}),
    \label{eq:tau}
\end{align}
where $g_i(\mathbf{X}_{it^*}) = \mathbb{E}[Y_{it^*}(0) \mid \mathbf{X}_{it^*}]$. The estimand replaces the unrecoverable realisation with its conditional mean given observables, the quantity that standard methods in ITS actually estimate. Aggregate estimands average or accumulate the unit-period effects,
\begin{align}
    \bar{\tau}_{t^*} = \frac{1}{N}\sum_{i=1}^{N} \tau_{it^*}, \qquad
    \bar{\tau}_i = \frac{1}{T - t_0 + 1}\sum_{s=0}^{T-t_0} \tau_{i,t_0+s}, \qquad
    \tau_{i,\text{cum}} = \sum_{t^*=t_0}^{T} \tau_{it^*},
    \label{eq:aggregates_estimands}
\end{align}
the average effect at period $t^*$ across units, the average effect for unit $i$ over the post-treatment window, and its cumulative total. Identification of every estimand turns on recovering $g_i$ from pre-treatment data and on its stability across $t_0$, which Assumption~\ref{assum:sufficiency} secures.

\begin{theorem}[Identification of Treatment Effects]
\label{theorem:identification}
Under consistency \eqref{eq:consistency}, no anticipation \eqref{eq:no_anticipation}, and Assumption~\ref{assum:sufficiency}, the treatment effect at $t^* \geq t_0$ is identified as $$\tau_{it^*} = Y_{it^*} - g_i(\mathbf{X}_{it^*}),$$ where $g_i$ coincides with the CEF of the observed pre-treatment data, $g_i(\mathbf{x}) = \mathbb{E}[Y_{it} \mid \mathbf{X}_{it} = \mathbf{x}]$ for $t < t_0$.
\end{theorem}

\noindent A proof is given in Supplementary Material~\ref{app:id_proof}. Theorem~\ref{theorem:identification} shows that the counterfactual half of the estimand is controlled by two features of the design: consistency and no anticipation deliver untreated outcomes before $t_0$, and Mean Sufficiency makes the relationship those outcomes trace a fixed function of $\mathbf{X}_{it}$ that persists past $t_0$. The substitution of $Y_{it^*}$ for $Y_{it^*}(1)$ uses consistency alone, so the content of the theorem is the second clause, that the unobservable half of the estimand is pinned down by the observed pre-treatment distribution. That distribution recovers $g_i$ on the support of $\mathbf{X}_{i,\text{pre}}$; because $\mathbf{X}_{it^*}$ carries calendar time, evaluating $g_i$ there reaches beyond the support, and the reach is disciplined by the smoothness the kernel encodes in place of a parametric form. What that discipline leaves unresolved is a property of the model class itself.

\begin{definition}[Extrapolation uncertainty]\label{def:extrap_uncertainty}
For the ball $\mathcal{F}_B = \{f \in \mathcal{H}_k : \|f\|_{\mathcal{H}_k} \le B\}$ in the RKHS induced by the kernel (Section~\ref{subsec:gp_estimation}) and a post-treatment input $\mathbf{X}_{it^*}$, the extrapolation uncertainty of $\mathcal{F}_B$ at $\mathbf{X}_{it^*}$ is
\begin{equation}\label{eq:extrap_uncertainty}
    \omega_B(\mathbf{X}_{it^*}) = 
    \sup\bigl\{\, |f(\mathbf{X}_{it^*}) - f'(\mathbf{X}_{it^*})| : f, f' \in
    \mathcal{F}_B,\; f(\mathbf{X}_{it}) = f'(\mathbf{X}_{it}) \text{ for all }
    t < t_0 \,\bigr\},
\end{equation}
the largest disagreement at $\mathbf{X}_{it^*}$ between members of the class that agree at every pre-treatment input.
\end{definition}

\noindent The definition adapts \citet[Definition~1]{shen2025engression} to the fixed design of the ITS: agreement at the observed inputs replaces agreement on a covariate support, the disagreement is read at a given post-treatment input rather than over a neighbourhood of that support, and the class is a norm ball rather than a shape-restricted one. The ball is what keeps $\omega_B$ finite, since $\mathcal{H}_k$ contains functions vanishing at every pre-treatment input and any multiple of one may be added to $f$ without leaving the space. Targeting the full conditional distribution rather than the conditional mean does not lower it either, because the structural model~\eqref{eq:structural} places the noise after the nonlinearity, and for that class distributional and functional extrapolability coincide \citep[Theorem~1(ii)]{shen2025engression}. Section~\ref{sec:theory} computes $\omega_B(\mathbf{X}_{it^*})$ in closed form.

The primary threat to Assumption~\ref{assum:sufficiency} is inherent to the ITS design rather than to a particular estimator. A component of $\mathbf{U}_{it^*}$ may take values after $t_0$ that it never took before, so a factor that sat constant throughout the pre-treatment window, and was therefore absorbed into the learned relationship, begins to shift $\E[Y_{it^*}(0) \mid \mathbf{X}_{it^*}]$, as with the post-election anxiety about gun restrictions in our \textit{Heller} application (Section~\ref{sec:application}). The bias decomposition of Section~\ref{sec:theory} expresses this threat through Term~I. Such breaks become more likely as the horizon lengthens, so shorter post-treatment windows are more credible; domain knowledge about concurrent events is essential, and the placebo checks (Section~\ref{subsec:placebo_check}) diagnose the estimation procedure in the pre-treatment period where the true effect is zero.

\section{Proposed Methodology: Gaussian Process Estimation for ITS}\label{sec:estimation}

Building on the identification result of Theorem~\ref{theorem:identification}, this section introduces a general estimation strategy and proposes a GP estimator \citep{rasmussen2006gaussian}.

\subsection{A General Estimation Procedure}\label{subsec:estimation_procedure}

Theorem~\ref{theorem:identification} reduces causal identification to a function learning problem: estimate $g_i(\cdot)$ from pre-treatment data and evaluate it at each post-treatment input. An ideal estimator would recover the target at the test input,
\begin{equation}
    \mathbb{E}[\hat{g}_i(\mathbf{X}_{it^*}) \mid \mathbf{X}_{i,\text{pre}}, \mathbf{X}_{it^*}] = g_i(\mathbf{X}_{it^*}),
    \label{eq:zero_learning}
\end{equation}
with the expectation over the pre-treatment noise. We do not assume~\eqref{eq:zero_learning}. The GP posterior mean is a regularised estimator that shrinks toward the prior away from the data, so~\eqref{eq:zero_learning} holds only in the limit under extrapolation. We therefore bound the learning error directly through the GP posterior variance (Section~\ref{sec:theory}).

The estimation procedure runs in four steps. (a) For each unit $i$, estimate $g_i$ from the pre-treatment data $(\mathbf{X}_{i,\text{pre}}, Y_{i,\text{pre}})$ to obtain $\hat{g}_i$. (b) For each post-treatment period $t^* \geq t_0$, predict the counterfactual outcome $\hat{Y}_{it^*}(0) = \hat{g}_i(\mathbf{X}_{it^*})$. (c) Estimate the treatment effect $\hat{\tau}_{it^*} = Y_{it^*} - \hat{g}_i(\mathbf{X}_{it^*})$. (d) Aggregate to the summary estimands of Section~\ref{subsec:panel} as needed. The procedure imposes no \textit{a priori} restriction on the temporal pattern of treatment effects, such as the step, pulse, or gradual transfer functions of the standard ITS toolkit \citep{mcdowall2019interrupted, schaffer2021interrupted}; the shape of $\tau_{it^*}$ over $t^*$ emerges from the data. Any learner for $g_i$ slots into step (a) without disturbing the causal interpretation established in Theorem~\ref{theorem:identification}.

\subsection{Gaussian Process Estimation}\label{subsec:gp_estimation}

Standard ITS approaches assume parametric forms that risk misspecification bias and undercoverage under unmodelled nonlinearities. We estimate $g_i$ by GP regression, which suits the ITS setting on three counts. It carries temporal dependence through the kernel rather than a fixed functional form, its predictive uncertainty grows with extrapolation distance, and the kernel encodes domain knowledge while regularising against overfitting, a combination standard for time series prediction \citep[e.g.,][]{gibson2012gaussian}.

We model the untreated potential outcome for unit $i$ at time $t$ as
\begin{align}
    Y_{it}(0) = f_i(\mathbf{X}_{it}) + \varepsilon_{it}, \quad 
    f_i \sim \mathcal{GP}(0, k), \quad
    \varepsilon_{it} \sim \mathcal{N}(0, \sigma^2),
    \label{eq:gp_model}
\end{align}
where $f_i$ is a unit-specific latent function and $k$ is a positive-definite covariance function encoding assumptions about the smoothness and structure of the underlying process; the matrix $\mathbf{K}_{\text{pre}}$, with elements $K_{tt'} = k(\mathbf{X}_{it}, \mathbf{X}_{it'})$, collects its evaluations at the pre-treatment inputs. Normality of $\varepsilon_{it}$ enters only the posterior computation and the marginal likelihood used for hyperparameter tuning; neither Theorem~\ref{theorem:identification} nor the bias decomposition of Proposition~\ref{prop:bias} relies on it. The target function is $g_i(\mathbf{X}_{it}) = \mathbb{E}[Y_{it}(0) \mid \mathbf{X}_{it}]$; under Assumption~\ref{assum:sufficiency} it is recoverable from pre-treatment data, and the GP posterior mean $\hat{g}_i$ is our estimator of $g_i$.

\paragraph{Posterior Prediction and Inference.}
Given pre-treatment data $(\mathbf{X}_{i,\text{pre}}, Y_{i,\text{pre}})$ and post-treatment covariate values $\mathbf{X}_{i,\text{post}}$, the posterior predictive distribution over counterfactual outcomes is \citep{rasmussen2006gaussian}:
\begin{align}
    \mathbf{Y}_{i,\text{post}}(0) \mid 
    Y_{i,\text{pre}},\, \mathbf{X}_{i,\text{pre}},\, 
    \mathbf{X}_{i,\text{post}}
    &\sim \mathcal{N}(\boldsymbol{\mu}_{\text{post}},\, 
    \boldsymbol{\Sigma}_{\text{post}}), \nonumber \\
    \boldsymbol{\mu}_{\text{post}} 
    &= \mathbf{K}_{\text{post,pre}}
       (\mathbf{K}_{\text{pre}} + \sigma^2 I)^{-1} 
       Y_{i,\text{pre}}, \label{eq:posterior_mean} \\
    \boldsymbol{\Sigma}_{\text{post}} 
    &= \mathbf{K}_{\text{post}} + \sigma^2 I 
       - \mathbf{K}_{\text{post,pre}}
         (\mathbf{K}_{\text{pre}} + \sigma^2 I)^{-1}
         \mathbf{K}_{\text{pre,post}}, \label{eq:posterior_var}
\end{align}
where $\mathbf{K}_{\text{post,pre}}$ has $(t^*, t)$-entry $k(\mathbf{X}_{it^*}, \mathbf{X}_{it})$ and $\mathbf{K}_{\text{post}}$ has $(t^*, t^{**})$-entry $k(\mathbf{X}_{it^*}, \mathbf{X}_{it^{**}})$. Writing $\mathbf{k}_{t^*}$ for the row of $\mathbf{K}_{\text{post,pre}}$ at $t^*$, the corresponding entry of $\boldsymbol{\mu}_{\text{post}}$ is the GP posterior mean $\hat{g}_i(\mathbf{X}_{it^*}) = \mathbf{k}_{t^*}^\top (\mathbf{K}_{\text{pre}} + \sigma^2 I)^{-1} Y_{i,\text{pre}}$, and the diagonal entry of $\boldsymbol{\Sigma}_{\text{post}}$ is the posterior predictive variance. Section~\ref{sec:theory} establishes that the GP posterior variance controls the learning error pointwise.

\paragraph{Kernel Specification.}

The kernel $k(\mathbf{X}_t, \mathbf{X}_{t'}) = \text{cov}(f(\mathbf{X}_t),\, f(\mathbf{X}_{t'}))$ sets the covariance between function values at any two inputs, so that inputs closer in the covariate space receive higher covariance. Each kernel induces a RKHS $\mathcal{H}_k$ of functions representable as weighted combinations of kernel evaluations, and the norm $\|f\|_{\mathcal{H}_k}$ measures the complexity of $f$ relative to the kernel, with functions built from directions the kernel treats as smooth carrying a small norm and rougher functions a large one. Section~\ref{sec:theory} puts this norm to work as the budget that disciplines extrapolation. As a starting point we take the Gaussian kernel,
\begin{equation}
    k_{\text{Gaussian}}(\mathbf{X}, \mathbf{X}') 
    = \exp\!\left(-\frac{\|\mathbf{X} - \mathbf{X}'\|^2}{b}\right),
    \label{eq:gaussian_kernel}
\end{equation}
which is universal, so $g_i$ is taken as smooth but otherwise unrestricted in shape \citep{micchelli2006universal}. In our setting it encodes the assumption that temporally proximate observations with similar covariate values have similar function values.

\paragraph{Combining Kernels.}

The Gaussian kernel is stationary, so predictions far from the pre-treatment data revert to the prior mean with uncertainty bounded by the marginal prior variance, a ceiling that suits local extrapolation but not longer horizons. Sums of valid covariance functions are valid covariance functions \citep{rasmussen2006gaussian}, so components can be assembled to suit an application \citep{duvenaud2013structure}; our working form is
\begin{equation}
    k_{\text{combined}}(\mathbf{X}, \mathbf{X}') 
    = k_{\text{Gaussian}}(\mathbf{X}, \mathbf{X}') 
    + k_{\text{periodic}}(\mathbf{X}, \mathbf{X}') 
    + k_{\text{linear}}(\mathbf{X}, \mathbf{X}').
    \label{eq:combined_kernel}
\end{equation}
The periodic kernel $k_{\text{periodic}}(\mathbf{X}, \mathbf{X}') = \exp\bigl(-2\sin^2(\pi|\mathbf{X} - \mathbf{X}'|/p)\,/\,(b/2)\bigr)$ captures seasonal patterns with period $p$, assigning maximum covariance to time points separated by integer multiples of $p$ \citep{mackay1998introduction}. The linear kernel $k_{\text{linear}}(\mathbf{X}, \mathbf{X}') = \mathbf{X}^\top \mathbf{X}'$ models trends that extend beyond the observed data, so predictive uncertainty grows with distance from the pre-treatment support; we use the homogeneous form, which keeps all components at the same degree and stays conservative when the post-treatment window is much shorter than the pre-treatment period. Scaling the outcome to unit variance fixes the signal variance at one and the relative weight of the three components, a reparameterisation that preserves the function space \citep{cho2026inference}. The Gaussian term dominates at short temporal distances and the linear term at long ones, hence the trend component matters most where extrapolation uncertainty is greatest. Where a particular composition is hard to defend, researchers can report how conclusions move across specifications.

\paragraph{Hyperparameter Selection.}
We center and scale each continuous covariate and the outcome to zero mean and unit variance within every unit, and one-hot encode categorical covariates such as the month and day-of-week indicators of our applications, so that the kernel reads all coordinates on a common scale; Supplementary Material~\ref{app:kernel_details} gives all the details. Including period indicators in $\mathbf{X}_{it}$ lets the kernel learn period-specific baselines without additional parametric mean functions to be estimated. The length-scale $b$ controls how rapidly covariance decays with temporal distance, with larger $b$ concentrating the prior on smooth, low-frequency functions and smaller $b$ admitting more rapid fluctuations. We select it by the variance-maximisation rule of \citet{kpop}, $\hat{b} = \underset{b}{\operatorname{argmax}}\; \mathbb{V}(K_{ij})$ for $i \neq j$, where $K_{ij}$ are the off-diagonal entries of $\mathbf{K}_{\text{pre}}$. The objective locates the length-scale at which the pairwise covariances are most spread out, avoiding values that drive every entry toward near zero or one. Marginal likelihood (all hyperparameters combined) and cross-validation optimize against outcome data and risk overfitting under extrapolation, whereas this rule reads only the covariate structure, keeping the design and analysis stages separate \citep{chattopadhyay2023implied}. Supplementary Material~\ref{app:kernel_details} re-estimates all effects across a broad range of $b$ and reports that the findings are robust. The period $p$ is set from the data's natural cycle, for instance $p = 12$ for monthly data, or from the periodogram of the detrended series when unknown.

\paragraph{Noise Variance and Regularisation.}
With the length-scale hyperparameter fixed at $\hat{b}$, we estimate $\sigma^2$ by maximising the marginal likelihood, $\hat{\sigma}^2 = \underset{\sigma^2}{\operatorname{argmax}}\; \log p(Y_{i,\text{pre}} \mid \mathbf{X}_{i,\text{pre}}, b, \sigma^2)$.
This parameter is both the irreducible noise and a regulariser, since the posterior mean solves
\begin{equation}\label{eq:krr}
    \hat{g}_i = \underset{f \in \mathcal{H}_k}{\operatorname{argmin}}\;
    \sum_{t < t_0} \bigl(Y_{it} - f(\mathbf{X}_{it})\bigr)^2
    + \sigma^2 \|f\|_{\mathcal{H}_k}^2,
\end{equation}
so $\sigma^2$ penalises the RKHS norm and pulls the estimate toward the directions the kernel represents cheaply (i.e., smooth). Fitting $\sigma^2$ by marginal likelihood does not reintroduce the overfitting that the covariate-only rule for $b$ avoids, since $\sigma^2$ is a single scalar fitted with $b$ already fixed.

\subsection{Placebo Checks}\label{subsec:placebo_check}

Our identification strategy requires that $g_i$, learned from pre-treatment data, remain valid for counterfactual prediction after $t_0$. Pre-treatment periods, where no treatment was applied, serve as held-out targets for checking this plausibility \citep{eggers2023placebo}. The exercise parallels the pre-trends check in difference-in-differences, run temporally and within-unit.

The procedure pretends that treatment began before its actual onset \citep[cf.][]{liu2024practical}. Take the last $M$ pre-treatment periods $\{t_0 - M, \ldots, t_0 - 1\}$ as a placebo post-treatment window, for instance matching the length of the actual one, and define relative time $m = -(M-1), \ldots, -1, 0$, with $m = 0$ the period immediately before onset. For each $m$ we fit the GP on all periods before $t_0 + m$ and form the placebo estimate
\begin{equation}
    \hat{\tau}_{i,m}^{\text{placebo}} = 
    Y_{i,\, t_0 + m} - 
    \hat{g}_i^{(t_0 + m)}(\mathbf{X}_{i,\, t_0 + m}), 
    \quad m = -(M-1), \ldots, -1, 0,
    \label{eq:placebo}
\end{equation}
where $\hat{g}_i^{(t_0 + m)}$ is the GP posterior mean trained on $\{1, \ldots, t_0 + m - 1\}$. No treatment is applied at a placebo target, so under an adequate kernel the estimate should fall near zero with intervals that cover it. Equation~\eqref{eq:placebo} reproduces the one-step-ahead extrapolation the estimator performs at the first post-treatment period, the shortest and least demanding horizon, so a clean result is a necessary bar rather than a guarantee for longer horizons.

The placebo serves two roles. First, it is a calibration check on the learning step, since systematic placebo error implicates the kernel rather than treatment. Second, it is a one-sided check on Assumption~\ref{assum:sufficiency}, and the asymmetry is logical. No pre-treatment statistic can confirm that $\mathbf{X}_{it}$ is a sufficient conditioning set, so the placebo can only falsify. A confounder inside the training window is absorbed into $\hat{g}_i$ and leaves placebo error near zero even when the same disturbance, persisting past $t_0$, violates Mean Sufficiency. A clean placebo therefore reports the absence of detected instability, not confirmation of the assumptions. Where feasible we might pair temporal placebos with placebo outcomes \citep[see][]{felton2023}; Supplementary Material~\ref{app:empirical_examples} reports the comparison for each application.

\subsection{Aggregation for Panel Data} \label{subsec:panel}
With $N$ units over common periods, we fit a separate GP to each unit's pre-treatment series and predict unit-specific counterfactuals $\hat{g}_i(\mathbf{X}_{it^*})$. Fitting unit by unit accommodates arbitrary heterogeneity in temporal trajectories, at the cost of not borrowing strength across units. The unit-specific effect estimate is $\hat{\tau}_{it^*} = Y_{it^*} - \hat{g}_i(\mathbf{X}_{it^*})$, with prediction intervals from each unit's GP posterior predictive variance. Summary estimands aggregate the unit- and period-specific effects $\hat{\tau}_{it^*}$ over units at a fixed period, or over the post-treatment periods of a fixed unit as a mean or a cumulative total:
\begin{equation}
    \hat{\bar{\tau}}_{t^*} = \frac{1}{N} 
    \sum_{i=1}^{N} \hat{\tau}_{it^*}, \qquad
    \hat{\bar{\tau}}_{i} = \frac{1}{T - t_0 + 1}
    \sum_{s=0}^{T-t_0} \hat{\tau}_{i,t_0+s}, \qquad
    \hat{\tau}_{i,\text{cum}} = \sum_{t^*=t_0}^{T} 
    \hat{\tau}_{it^*}.
    \label{eq:aggregates}
\end{equation}
Because each unit's GP is fitted independently, the posterior predictive distributions are independent across units conditional on the data, and prediction intervals for the aggregate estimands follow from the Gaussian posterior. The independence is a modelling choice rather than a feature of the data, holding when common structure is pre-treatment and absorbed into each $\hat{g}_i$; common shocks or spatial spillovers enter each unit's own temporal pattern rather than a cross-unit covariance \citep[cf.][]{ben2023estimating}. At the unit level this is conservative, regularising each series on its own and paying for it in wider intervals; the conservativeness does not carry over to aggregates under positive cross-unit dependence.

\section{Theoretical Properties of Gaussian Process Estimator}\label{sec:theory}

Several sources of error may cause the GP posterior mean $\hat{g}(\mathbf{X}_{t^*})$ to deviate from the true $Y_{t^*}(0)$. This section decomposes the estimation error into four interpretable components, derives a bias result isolating the role of each, and connects the GP posterior variance to a worst-case bound on the learning error through RKHS theory.

\subsection{Setup and the Four Functions}

We suppress the unit subscript $i$ throughout; all objects are unit-specific and the analysis applies unit by unit. The structural model~\eqref{eq:structural} defines $h_i(\mathbf{X}_{t}, \mathbf{U}_{t})$, which depends on both observables and unobservables, and the CEF $g(\mathbf{X}_{t^*}) = \mathbb{E}[Y_{t^*}(0) \mid \mathbf{X}_{t^*}]$ integrates out $\mathbf{U}_{t}$ so that only observables remain. When estimation is carried out in an RKHS $\mathcal{H}_k$, two more arise. Fix a \textit{complexity budget} $B > 0$, the RKHS norm we are willing to grant the counterfactual. The \textit{best approximation within budget}
\begin{equation}\label{eq:gstar}
    g^* = \underset{f \in \mathcal{F}_B}{\operatorname{argmin}}
    \|f - g\|_{L^2(P_{\mathbf{X}})}
\end{equation}
is the closest function to $g$ in $\mathcal{F}_B$, and $g^* = g$ when
$g \in \mathcal{F}_B$. The $L^2(P_X)$ criterion makes the minimizer unique and aligns with the spectral machinery of Section~4.3; in the leading case $g$ sits inside the ball and the criterion does no work. The norm bound makes $g^*$ well defined, since a universal kernel drives the infimum to zero along functions of diverging norm; Supplementary Material~\ref{app:rkhs_theory} measures this budget on the series of our application. The \textit{GP estimator} $\hat{g}(\mathbf{X}_{t^*}) = \mathbf{k}_{t^*}^\top(\mathbf{K}_{\text{pre}} + \sigma^2 \mathbf{I})^{-1} Y_{\text{pre}}$, where $\mathbf{k}_{t^*} = (k(\mathbf{X}_{t^*}, \mathbf{X}_1), \ldots, k(\mathbf{X}_{t^*}, \mathbf{X}_{t_0-1}))^\top$, is the GP posterior mean. By the representer theorem \citep{kimeldorf1970correspondence}, this lies in the finite-dimensional empirical subspace $\mathcal{S}_{\text{pre}} = \operatorname{span}\{k(\cdot,\mathbf{X}_t) : t=1, \ldots, t_0 - 1\} \subset \mathcal{H}_k$. Table~\ref{tab:four_functions} summarises these four functions.

\begin{table}[htbp]
    \centering
    \caption{Functions in the bias decomposition for counterfactual prediction.}
    \label{tab:four_functions}
    \begin{tabular}{lll}
        \toprule
        Function & Definition & Governed by \\
        \midrule
        $h(\mathbf{X}_t, \mathbf{U}_t)$ & Structural function & Observables + unobservables \\
        $g(\mathbf{X}_t)$ & CEF: $\mathbb{E}[Y_t(0) \mid \mathbf{X}_t]$ & Observables only \\
        $g^*(\mathbf{X}_t)$ & Best $\mathcal{H}_k$ approximation of $g$ & Kernel choice \\
        $\hat{g}(\mathbf{X}_t)$ & GP estimator in $\mathcal{S}_{\text{pre}}$ & Kernel + training data \\
        \bottomrule
    \end{tabular}
\end{table}

\subsection{Bias Decomposition}

The accuracy of $\hat{\tau}_{t^*}$ depends on how well $\hat{g}(\mathbf{X}_{t^*})$ approximates the untreated potential outcome $Y_{t^*}(0)$. Substituting $Y_{t^*}(0) = h(\mathbf{X}_{t^*}, \mathbf{U}_{t^*}) + \varepsilon_{t^*}$ and inserting $g$ and $g^*$ decomposes the counterfactual prediction error as
\begin{align}\label{eq:four_term}
    \hat{\tau}_{t^*} - \delta_{t^*}  &= Y_{t^*}(0) - \hat{g}(\mathbf{X}_{t^*}) \\
    &=
    \underbrace{[h(\mathbf{X}_{t^*}, \mathbf{U}_{t^*})
    - g(\mathbf{X}_{t^*})]}_{\text{I: identification}} +
    \underbrace{[g(\mathbf{X}_{t^*})
    - g^*(\mathbf{X}_{t^*})]}_{\text{II: approximation}} +
    \underbrace{[g^*(\mathbf{X}_{t^*})
    - \hat{g}(\mathbf{X}_{t^*})]}_{\text{III: learning}} +
    \underbrace{\varepsilon_{t^*}}_{\text{IV: noise}}.
\end{align}
The four right-hand terms trace prediction error through successive layers, from the unobservable structural function $h$, through the identifiable target $g$, to the best RKHS approximation $g^*$, and finally to the finite-sample estimator $\hat{g}$. The decomposition targets $Y_{t^*}(0) - \hat{g}(\mathbf{X}_{t^*})$, which equals $\hat{\tau}_{t^*} - \delta_{t^*}$; the error for the primary estimand, $\hat{\tau}_{t^*} - \tau_{t^*} = g(\mathbf{X}_{t^*}) - \hat{g}(\mathbf{X}_{t^*})$, involves only Terms~II and~III, since $Y_{t^*}(1)$ cancels.

\textbf{Term~I: Identification error.} The discrepancy between $h(\mathbf{X}_{t^*}, \mathbf{U}_{t^*})$ and $g(\mathbf{X}_{t^*})$ measures how much not observing $\mathbf{U}_{t^*}$ biases the prediction. By definition of $g$ as a conditional expectation, $\mathbb{E}[\text{Term~I} \mid \mathbf{X}_{t^*}] = 0$; under Assumption~\ref{assum:sufficiency}, $h = g$ almost surely and Term~I vanishes.

\textbf{Term~II: Approximation error.} The gap between the true CEF $g$ and the best approximation within the budget. It is zero when $\|g\|_{\mathcal{H}_k} \leq B$ and persists otherwise, however much training data accumulates. A universal kernel is dense in the continuous functions on a compact domain \citep{micchelli2006universal}, so it approximates $g$ to any tolerance, but the norm of the approximants grows as the tolerance shrinks (Supplementary Material~\ref{app:spectral_rep}). The combined kernel keeps the budget small for such targets, carrying a trend in its linear component at a horizon-independent cost rather than at a cost that grows with the length of the series.

\textbf{Term~III: Learning error.} The gap between the best approximation $g^*$ and the estimate $\hat{g}$ recoverable from the finite training sample. Even when Term~II vanishes, so that $g^* = g$, the representer theorem confines $\hat{g}$ to the empirical subspace $\mathcal{S}_{\text{pre}}$, so the part of $g^*$ outside that span is unreachable from the observed inputs, and inside the span the ridge term $\sigma^2$ penalises $\|\hat g\|_{\mathcal{H}_k}$ and deliberately biases the fit toward smoother functions. Proposition~\ref{prop:worst_case} splits this error into a deterministic part and mean-zero noise propagation and controls each through a summand of the GP posterior variance.

\textbf{Term~IV: Noise.} The irreducible error $\varepsilon_{t^*}$, with $\mathbb{E}[\varepsilon_{t^*} \mid \mathbf{X}_{t^*}, \mathbf{U}_{t^*}] = 0$ by the structural model~\eqref{eq:structural}, contributing $\sigma^2$ to the variance but nothing to the bias. 

Taking the conditional expectation of the four-term decomposition pins down which terms survive in the bias.
\begin{proposition}[Conditional Bias]\label{prop:bias}
Under the structural model~\eqref{eq:structural}, Assumption~\ref{assum:sufficiency}, and $\mathbb{E}[\varepsilon_{t^*} \mid \mathbf{X}_{\mathrm{pre}}, \mathbf{X}_{t^*}] = 0$, the conditional bias of $\hat{\tau}_{t^*}$ for $\delta_{t^*}$ satisfies
\begin{equation}\label{eq:bias}
    \mathbb{E}[\hat{\tau}_{t^*} - \delta_{t^*} \mid
    \mathbf{X}_{\text{pre}}, \mathbf{X}_{t^*}] =
    \underbrace{g(\mathbf{X}_{t^*}) - g^*(\mathbf{X}_{t^*})}_{\text{approximation bias}} +
    \underbrace{\mathbb{E}[g^*(\mathbf{X}_{t^*}) - \hat{g}(\mathbf{X}_{t^*})
    \mid \mathbf{X}_{\text{pre}}, \mathbf{X}_{t^*}]}_{\text{learning bias}}.
\end{equation}
\end{proposition}
\noindent The conditional bias is controlled by two gaps: how far $g$ sits from $\mathcal{F}_B$ and how much of $g^*$ the training inputs cannot reach. The approximation bias is deterministic given $\mathbf{X}_{t^*}$. The learning bias vanishes when $\hat{g}$ is unbiased for $g^*$, the ideal benchmark~\eqref{eq:zero_learning} with $g^* = g$, and is in general nonzero; Proposition~\ref{prop:worst_case} bounds its worst case. Terms~I and~IV drop out for different reasons. Term~I has mean zero given $\mathbf{X}_{t^*}$ by the definition of $g$, but serial dependence can let $\mathbf{X}_{\text{pre}}$ carry information about $\mathbf{U}_{t^*}$; Assumption~\ref{assum:sufficiency} closes this gap by setting $h = g$ almost surely. Term~IV drops out under the weaker requirement that $\varepsilon_{t^*}$ be mean-independent of $\mathbf{X}_{\text{pre}}$ given $\mathbf{X}_{t^*}$; the noise contributes no bias but persists in the pointwise error, adding $\sigma^2$ to the posterior predictive variance.

\subsection{RKHS Interpretation of GP Posterior Variance}

The GP posterior mean coincides with the kernel ridge regression estimator in $\mathcal{H}_k$ \citep{kimeldorf1970correspondence}, so as a point predictor it has a frequentist counterpart. Kernel ridge regression returns a point estimate; the GP additionally returns a posterior variance, whose two summands control the deterministic and stochastic parts of the learning error (Term~III) \citep[Ch.~6]{rasmussen2006gaussian}; \citet{kanagawa2025gaussian} give a recent treatment of the GP--RKHS connection.

\begin{proposition}[GP posterior variance controls the learning error]
\label{prop:worst_case}
Suppose $g$ lies in the unit ball $\mathcal{F}_1$, so that $g^* = g$, and
suppose outcomes are observed with noise, $Y_t = g^*(\mathbf{X}_t) + \varepsilon_t$, where $\varepsilon_t$ are mean-zero with variance $\sigma^2 > 0$. The GP posterior mean is the weighting estimator $\hat g(\mathbf{X}_{t^*}) = \sum_t w_t Y_t$ with weights $\mathbf{w} = (\mathbf{K}_{\mathrm{pre}} + \sigma^2 I)^{-1}\mathbf{k}_{t^*}$, and its posterior variance is
\begin{equation*}
    \mathbb{V}_{t^*} = k(\mathbf{X}_{t^*},\mathbf{X}_{t^*}) - \mathbf{k}_{t^*}^\top(\mathbf{K}_{\mathrm{pre}} + \sigma^2 I)^{-1}\mathbf{k}_{t^*}.
\end{equation*}
Let $k_{t^*}^\perp = k(\cdot,\mathbf{X}_{t^*}) - \sum_t w_t\, k(\cdot,\mathbf{X}_t)$ be the residual of the test feature after its ridge reconstruction from the training features $\mathcal{S}_{\mathrm{pre}} = \mathrm{span}\{k(\cdot,\mathbf{X}_t)\}$, coinciding with the orthogonal-projection residual onto $\mathcal{S}_{\mathrm{pre}}$ as $\sigma^2 \to 0$. Then the learning error decomposes as
\begin{equation*}
    g^*(\mathbf{X}_{t^*}) - \hat g(\mathbf{X}_{t^*})
    = \underbrace{\Big(g^*(\mathbf{X}_{t^*}) - \textstyle\sum_t w_t\,g^*(\mathbf{X}_t)\Big)}_{\text{deterministic error}}
    \;-\; \underbrace{\textstyle\sum_t w_t\,\varepsilon_t}_{\text{noise propagation}},
\end{equation*}
and its two sources are controlled by the two summands of the GP posterior variance:
\begin{equation*}
    \Big|g^*(\mathbf{X}_{t^*}) - \sum_t w_t g^*(\mathbf{X}_t)\Big| \le \|k_{t^*}^\perp\|_{\mathcal{H}_k}, \quad
    \mathrm{Var}\big[\sum_t w_t \varepsilon_t\big] = \sigma^2\|\mathbf{w}\|^2,
    \qquad
    \underbrace{\|k_{t^*}^\perp\|_{\mathcal{H}_k}^2 + \sigma^2\|\mathbf{w}\|^2 = \mathbb{V}_{t^*}}_{\text{GP posterior variance}}.
\end{equation*}
Therefore $\mathbb{V}_{t^*}$ bounds the worst case of the deterministic error through its first summand and carries the exact variance of the propagated noise in its second.
\end{proposition}
\noindent A proof is given in Supplementary Material~\ref{app:rkhs_theory}. Restricting to $\mathcal{F}_1$ costs no generality, since over $\mathcal{F}_B$ the worst-case bound scales to $B\,\|k_{t^*}^\perp\|_{\mathcal{H}_k}$ and the noise term is unchanged; the learning bias of Proposition~\ref{prop:bias} is then bounded by $\|g^*\|_{\mathcal{H}_k}\,\|k_{t^*}^\perp\|_{\mathcal{H}_k}$, a product of what the target costs and what the design cannot see. The geometric factor $\|k_{t^*}^\perp\|_{\mathcal{H}_k}$ is the part of the test feature the training inputs cannot reconstruct, and it grows as $\mathbf{X}_{t^*}$ leaves the pre-treatment support. Under a stationary kernel it never exceeds $\sqrt{k(\mathbf{X}_{t^*}, \mathbf{X}_{t^*})}$, a constant, so the band flattens at long horizons and any further growth in the error must come through the budget. The Gaussian RKHS contains no non-trivial linear function, and tracks a trend only through elements whose norm grows at least linearly with the window's length (Supplementary Material~\ref{app:spectral_rep}), so the band flattens where the budget inflates. The combined kernel mitigates this: by the Aronszajn identity \citep{aronszajn1950theory}, any decomposition $g = g_{\text{lin}} + g_{\text{per}} + g_{\text{Gaussian}}$ yields
\begin{equation}\label{eq:combined_bound_main}
    \Big| g^*(\mathbf{X}_{t^*}) - \textstyle\sum_t w_t\, g^*(\mathbf{X}_t) \Big|
    \le \sqrt{\|g_{\text{lin}}\|_{\mathcal{H}_{\text{lin}}}^2 + \|g_{\text{per}}\|_{\mathcal{H}_{\text{per}}}^2 + \|g_{\text{Gaussian}}\|_{\mathcal{H}_{\text{Gaussian}}}^2}\;\, \|k_{t^*}^\perp\|_{\mathcal{H}_k},
\end{equation}
charging each structure where it is cheapest, with $\|g_{\text{lin}}\|_{\mathcal{H}_{\text{lin}}} = \|\beta\|$ depending only on the slope and not the horizon. A counterfactual needing $B > 1$ leaves the reported interval short by that factor on the deterministic part. Two design choices hold $B$ near one: the unit-variance standardisation of Section~\ref{subsec:gp_estimation} fixes the scale, and the Aronszajn split charges trend and seasonality where their cost is horizon-independent. On the \textit{Heller} series a unit-variance trend needs $B = 0.96$ under the combined kernel and $B = 3.52$ under a Gaussian-only kernel. This finite-sample bound, not an in-support $L^2$ convergence rate, is the operative guarantee off support, because the fill distance at $\mathbf{X}_{t^*}$ is fixed by the study design and does not shrink with the pre-treatment length (Supplementary Material~\ref{app:convergence_rates}).

Let $T_k$ be the kernel integral operator---a smoothing map that replaces a function's value with a kernel-weighted average of its values---with eigenvalues $\mu_1 \ge \mu_2 \ge \cdots > 0$ indexing the \textit{spectral directions} of $\mathcal{H}_k$. Applying $T_k^{r}$ damps the small-eigenvalue directions, so the source condition of order $r \ge 0$, $g^* = T_k^{r} v$ with $\|v\|_{\mathcal{H}_k} \le 1$, asks that $g^*$ be built from directions the kernel represents cheaply, more strictly for larger $r$ (Supplementary Material~\ref{app:source_condition_formal}). The deterministic learning error then obeys
\begin{equation}\label{eq:source_main}
    \Big| g^*(\mathbf{X}_{t^*}) - \textstyle\sum_t w_t\, g^*(\mathbf{X}_t) \Big|
    \le \big\| T_k^{r}\, k_{t^*}^\perp \big\|_{\mathcal{H}_k}
    \le \mu_1^{r}\, \|k_{t^*}^\perp\|_{\mathcal{H}_k},
\end{equation}
with $r = 0$ recovering Proposition~\ref{prop:worst_case}. The middle quantity is the exact worst case over targets of order $r$, and the right-hand side trades sharpness for an explicit constant. Alignment therefore tightens the extrapolation bound itself, not only the in-support rate for which the source condition is usually invoked (Supplementary Material~\ref{app:source_condition_formal}).

These bounds carry a frequentist meaning that does not treat $g$ as a draw from the prior. The realised unit, its covariates, and the target $g$ are held fixed, and the only randomness is the idiosyncratic noise within the observed series. The worst case is an adversarial choice of the counterfactual, attained by the trajectory that spends its entire budget on the direction the training periods cannot see, $g^* \propto k_{t^*}^\perp$.

The GP posterior mean is also minimax-optimal at the extrapolation points: no estimator built from the same pre-treatment data has a smaller worst-case error over the unit ball, and the reported band does not merely bound the minimax risk but equals it (Supplementary Material~\ref{app:rkhs_theory}). With noiseless training the geometry is the whole story (Corollary~B.1), $|g^*(\mathbf{X}_{t^*}) - \hat g(\mathbf{X}_{t^*})| \le \|k_{t^*}^\perp\|_{\mathcal{H}_k} = \sqrt{\mathbb{V}_{t^*}}$; with $\sigma^2 > 0$ the second summand adds the variance $\sigma^2\|\mathbf{w}\|^2$ of the propagated noise. The same geometry also caps how far any two admissible counterfactuals can disagree at $t^*$, and the pair $\pm B\,k_{t^*}^\perp/\|k_{t^*}^\perp\|_{\mathcal{H}_k}$ attains the cap, so it measures the extrapolation uncertainty of the class itself in the sense of Definition~\ref{def:extrap_uncertainty}, $\omega_B(\mathbf{X}_{t^*}) = 2B\,\|k_{t^*}^\perp\|_{\mathcal{H}_k}$ (Supplementary Material~\ref{app:extrap_uncertainty}). The band's half-width at $t^*$ is $\omega_B(\mathbf{X}_{t^*})/2$, so drawn across horizons it traces the extrapolation profile of the model class.

This interval's breadth is the point of the method. The GP keeps the rival continuations the kernel admits and lets their spread set the post-treatment uncertainty. Read backwards, a researcher who claims an effect should argue that the counterfactual trajectory belongs to a class narrow enough for the interval to exclude zero, and the class that sustains the claim is visible in the reported band. Enriching the set of admissible continuations can only widen the interval, so the conclusion rests on the excluded trajectories as much as on the included ones.

\section{Calibrated Simulation Studies}
\label{sec:simulation}

\paragraph{Setup.}
To construct a realistic simulation design, we base the DGPs on the monthly series of our empirical illustration (Section~\ref{sec:application}), a seasonal outcome observed over a short pre-treatment window above a smooth trend. Our interest is deliberately narrow. Point forecasts are easy to compare across methods, but the main benefit of our proposed method is to deliver an interval whose width tracks what the pre-treatment record pins down. We ask whether each method's nominal $95\%$ interval covers the counterfactual at its stated rate, and how coverage changes with the pre-treatment length.

We generate monthly outcomes $y_t = g(t) + \varepsilon_t$ for $t = 1, \ldots, n_{\mathrm{pre}} + 12$, with $\varepsilon_t \sim \mathcal{N}(0,\sigma^2)$, treatment beginning at $t_0 = n_{\mathrm{pre}} + 1$, a fixed twelve-month post-treatment horizon, and a constant effect $\tau$ added over that horizon. A small random Gaussian bump displaces each trend from a clean closed form.

We consider three regimes for $g$ that share the monthly seasonality and short record of the application but place the target at increasing distance from the estimator (see Supplementary Material~\ref{app:sim_details} for more details). The first, \textit{Smooth}, is a draw from a kernel-based DGP, so the approximation error (Term~II of Section~\ref{sec:theory}) is negligible; the estimator still tunes its own hyperparameters, so the target is not exactly in its span. The second, \textit{Nonlinear trend}, places a milder seasonal term on a concave, decelerating trend that mimics a bounded growth path. The third, \textit{Saturating trend}, sets a strong annual cycle above a levelling trend where seasonality dominates and the trend flattens at both ends.

We fit the GP as in the application of Section~\ref{sec:application} and compare it with three benchmarks: (i) a segmented regression carrying a linear pre-trend and an annual harmonic, (ii) \texttt{C-ARIMA} \citep{menchetti2023combining}, and (iii) \texttt{CausalImpact}, a Bayesian structural time-series model \citep{brodersen2015inferring}. For each regime we let $n_{\mathrm{pre}}$ run from $12$ to $120$ months in steps of twelve, draw $500$ replications at each length, and calculate the absolute bias, the root-mean-squared error (RMSE) and the coverage rate across the post-treatment window.

\begin{figure}[thb!]
  \centering
  \includegraphics[width=1\textwidth]{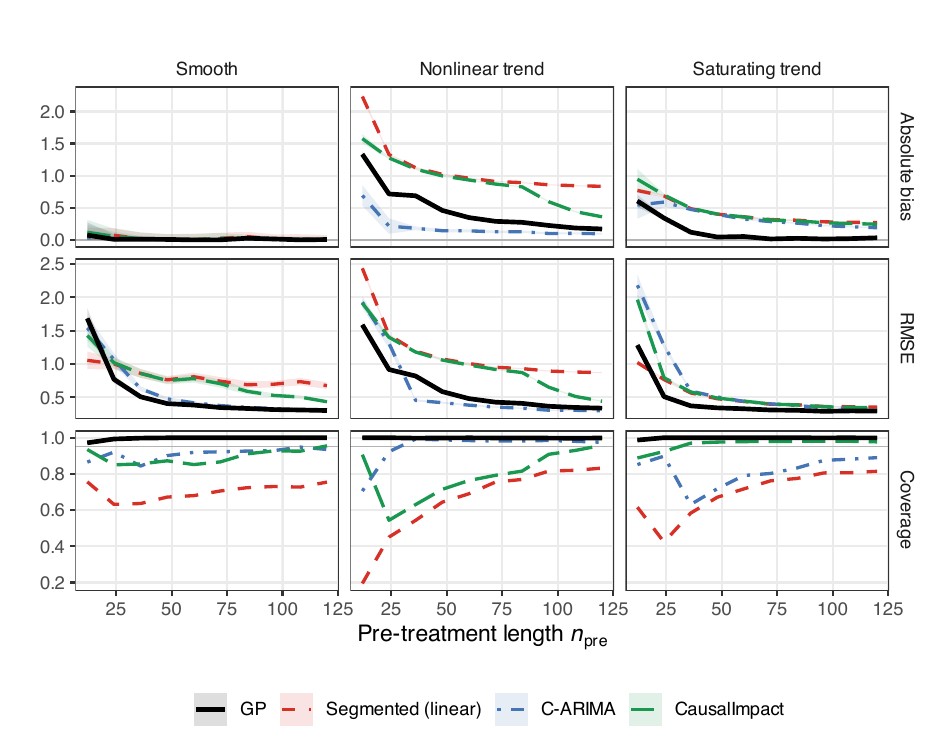}
  \caption{Absolute bias, RMSE and $95\%$ coverage against the pre-treatment length $n_{\mathrm{pre}}$, over a fixed twelve-month horizon and $500$ replications, for three scenarios (columns). Bias is that of the average post-treatment effect; RMSE and coverage are computed at each post-treatment period and then averaged. The GP (black) is compared with a segmented regression, \texttt{C-ARIMA} and \texttt{CausalImpact}. Shaded bands are Monte-Carlo $95\%$ intervals. The GP stays at or above its nominal rate in every scenario.}
  \noindent\textbf{Alt text:} Grid of line plots comparing four estimation methods across three simulation regimes, with rows for absolute bias, root mean squared error, and coverage plotted against the pre-treatment length.
  \label{fig:app_sim_samplesize}
\end{figure}

\paragraph{Results.}
Figure~\ref{fig:app_sim_samplesize} reports the three metrics, and the coverage panels carry the main message the theory leads us to expect. The GP stays at or above its nominal rate even at the shortest pre-treatment, a widening that the worst-case reading of Section~\ref{sec:theory} justifies. The benchmarks differ in how their intervals respond to extrapolation. A segmented regression's interval reflects the residual variance around its fitted value and does not widen as the forecast leaves that form, so its coverage sits well below nominal. \texttt{C-ARIMA} is close to calibrated once it is given the seasonal period, though it under-covers when the pre-treatment window is short. The coverage of \texttt{CausalImpact} moves sharply with the pre-treatment length as its state components are re-estimated at each sample size, dropping below nominal at intermediate length before recovering. Point estimation improves with the pre-treatment record for every method, and the methods separate less on bias and RMSE than on coverage, with the GP competitive throughout.

\section{Empirical Application} \label{sec:application}

Did the Supreme Court's decision in \textit{District of Columbia v.\ Heller} (2008) affect legal handgun purchases? On June 26, 2008, the Court ruled that the Second Amendment protects an individual right to keep and bear arms for self-defence. The ruling is a universal treatment, binding every jurisdiction at once, but its practical reach was narrow: it struck down only D.C.'s comprehensive handgun ban, in effect since the late 1970s, so D.C.\ went from a near-total ban to legal purchasing overnight while residents elsewhere could already purchase handguns. Effects outside D.C.\ would materialise only through later litigation, as in \textit{McDonald v.\ Chicago} (2010), or regulatory adaptation. Standard panel methods do not apply, since there are no untreated controls and D.C.'s pre-treatment trajectory is unlike that of jurisdictions with active handgun markets. The question is not treated versus control but \textit{where and when} a universal treatment has practical bite, and fitting GP-ITS independently to every jurisdiction should recover a large immediate effect in D.C.\ and null estimates elsewhere if its counterfactuals are well calibrated.

The data are a monthly panel with seasonal structure and a short post-treatment window, placing weight on the periodic and linear components of the Gaussian$+$Periodic$+$Linear kernel. Supplementary Material~\ref{app:additional_examples} applies the same estimator to two further settings on daily series with different temporal structure, All-Women Police Stations in India \citep{jassal2020gender} and the Russian invasion of Ukraine \citep{damann2024women}.

We measure handgun purchases using monthly background check rates per 100{,}000 population from the FBI's National Instant Criminal Background Check System (NICS), a reliable proxy for new purchases \citep{steidley2018toward}. The raw monthly series appears in Supplementary Material~\ref{app:empirical_examples}. The panel covers all U.S.\ states and D.C.\ except Hawaii, whose record is incomplete, giving $N = 50$ jurisdictions. The post-treatment window covers only July--October 2008, ending before the November presidential election, since Obama's victory might have introduced political anxiety about future gun restrictions, a distinct treatment that would confound attribution to \textit{Heller}. Anticipation is an obvious threat to event-based designs, which it can contaminate even when the specific outcome remains uncertain \citep{bertoli2026analyzing}, but D.C.'s near-total ban structurally precluded preemptive legal handgun purchases, which keeps no anticipation plausible here. Supplementary Material~\ref{app:assumptions} discusses each assumption in detail.

\paragraph{Model.} We employ the composite Gaussian$+$Periodic$+$Linear kernel to capture the temporal structure of the series, with the Gaussian component carrying local temporal dependence, the periodic component annual seasonality at $p = 12$, and the linear component the long-run trend. Month indicators enter as covariates so that the kernel learns month-specific baseline levels.

\paragraph{Illustrative Cases: D.C.\ and Washington State.} We begin with two jurisdictions that bracket the range of expected effects: D.C., whose handgun ban was directly struck down, and Washington state, a representative case with no immediate policy change (Figure~\ref{fig:plot_heller_dc_wa}).
\begin{figure}[tbh!]
    \centering
    \includegraphics[width=1\textwidth]{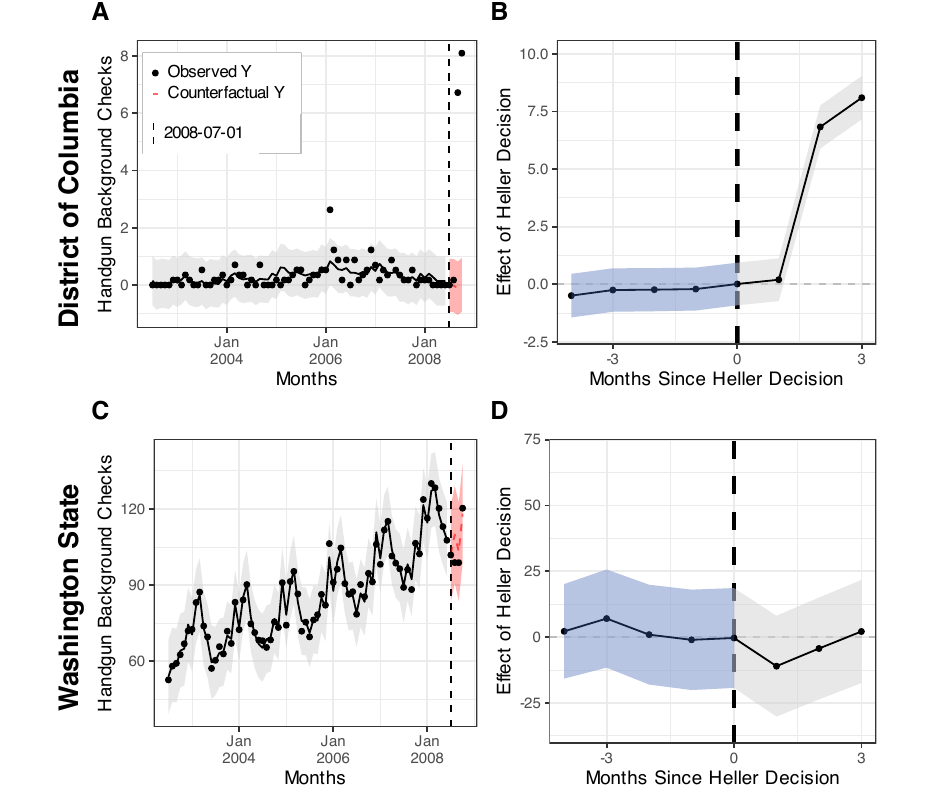}
    \caption{Model Fit and Dynamic Treatment Plot for D.C.\ and Washington State. \textit{Note: }The top row shows results for D.C., the bottom row for Washington state. Left panels show observed outcomes (black dots), the fitted GP (grey bands), and the predicted counterfactual (red bands); right panels the month-specific effects for the four post-treatment months (grey) and four placebo months (blue). Vertical dashed lines mark the \textit{Heller} decision.}
    \noindent\textbf{Alt text:} Model fit and treatment effect plots for two jurisdictions in the Heller application, showing fitted values and monthly effect estimates with placebo periods.
    \label{fig:plot_heller_dc_wa}
\end{figure}
For D.C., the treatment effect is immediate and substantial, remaining significant throughout the four-month post-treatment period, while the placebo checks (blue) cluster near zero. Washington state, by contrast, exhibits effects indistinguishable from zero, with the GP capturing its pre-treatment seasonality and increasing trend, propagating the associated uncertainty into the post-treatment period.

\paragraph{Localised \textit{Heller} Effect and Nationwide Null.} Figure~\ref{fig:heller_atts} presents unit-specific effects across all 50 jurisdictions. The left panel displays cumulative treatment effects for the four-month post-treatment period, standardised within each unit so that the panel ranks jurisdictions relative to their own pre-treatment variability. D.C.\ (red) is a distinct outlier with a large positive cumulative effect, and the remaining units cluster around zero, the pattern the substantive account predicts. Because the near-total ban holds D.C.'s pre-treatment series near zero, its suppressed variance inflates the standardised scale, so we report the magnitude on the count and per-capita scales instead. The cumulative four-month effect for D.C.\ is $15.1$ additional checks per 100{,}000 population, with a $95\%$ interval of $[13.0,\, 17.3]$, roughly $90$ additional background checks.

\begin{figure}[tbh!]
    \centering
    \includegraphics[width=1\textwidth]{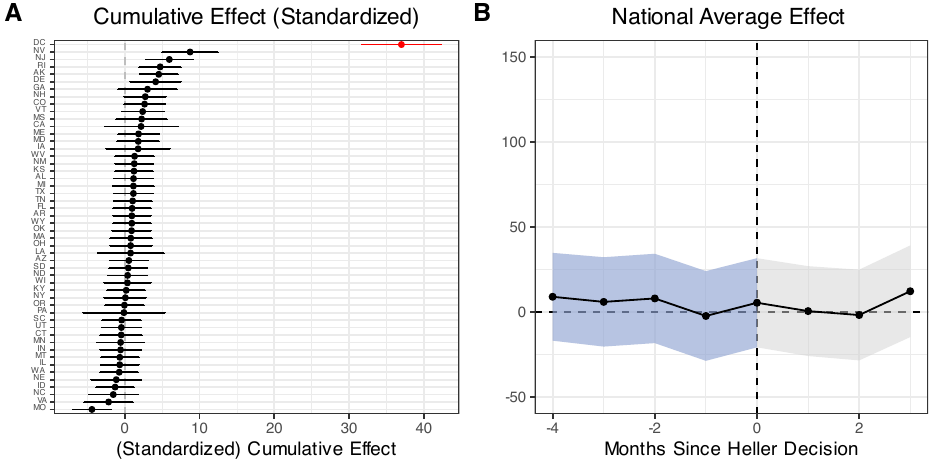}
    \caption{Heterogeneous and Nationwide Effects of the \textit{Heller} Decision on Handgun Purchases. \textit{Note:} The left panel gives unit-specific standardised cumulative effects for the four months following the decision, with $95\%$ prediction intervals. The right panel gives the estimated national average treatment effect $\hat{\bar{\tau}}_{t^*}$ by month.}
    \noindent\textbf{Alt text:} Summary of heterogeneous treatment effects across fifty jurisdictions in two panels, showing standardised cumulative effects and the national average effect over time.
    \label{fig:heller_atts}
\end{figure}
The right panel traces the national average effect $\hat{\bar{\tau}}_{t^*}$ over time, and the cumulative version confirms the same flat path, matching the absence of a practical policy change outside D.C. That an independent GP per jurisdiction isolates the D.C.\ effect while leaving the rest null is what well-calibrated unit-specific counterfactuals would produce. The same pattern speaks to the anticipation concern. Had the D.C.\ estimates reflected anticipatory purchasing ahead of the November election rather than the ruling, similar anticipation should appear elsewhere, since gun owners in other jurisdictions had equal reason to read the electoral signal. No nationwide effect appears during July--October 2008, which supports attribution to the ruling rather than to election-related anxiety.

As Section~\ref{subsec:placebo_check} notes, a placebo check can falsify a method but not confirm it. Under the matched one-step-ahead placebo protocol, segmented regression excludes zero at $58$ to $67\%$ of pre-treatment periods, flagging effects where none occurred, whereas the GP flags none. The contrast reproduces what the simulations of Section~\ref{sec:simulation} report: an interval built around a fixed parametric trend does not widen as the forecast leaves that trend. The GP survives a falsification that segmented regression fails. In-sample residuals show no degradation as the treatment date approaches, and a sensitivity analysis across a wide range of $b$ confirms that both the D.C.\ effect and the null national trend hold (Supplementary Material~\ref{app:kernel_details}).

As a complementary validation, \textit{Heller} addressed only handgun restrictions, so long gun checks serve as a placebo outcome. The standardised cumulative effect on long guns in D.C.\ is $0.25$, with a $95\%$ interval of $[-5.2,\, 5.7]$, consistent with a null where the ruling did not apply. Supplementary Material~\ref{app:empirical_examples} reports the cumulative effects, placebo comparisons, residual diagnostics, and long gun analysis.

\section{Relation to Existing Approaches} \label{sec:discussion}
Counterfactual outcomes draw on two sources of variation, cross-sectional contrasts between treated and untreated units and temporal contrasts within treated units before and after $t_0$. Most designs, including DiD, SC, and their extensions, use both, requiring parallel trends or an untreated donor pool. Universal treatment removes the cross-sectional source, leaving methods that extrapolate temporally; these differ in the structure they impose on the counterfactual trajectory and in how they quantify extrapolation uncertainty. We compare the GP first to time series methods built for this setting and then to designs that rest on cross-sectional comparison.

\paragraph{Time Series Approaches.} Classical ITS methods, from the intervention analysis of \citet{box1975intervention} to segmented regression, fix both the counterfactual trend and the treatment response pattern (e.g., step, pulse, or gradual transfer functions), commitments that aid estimation but constrain the analysis before any post-treatment data are seen. Recent work relaxes parts of this structure, with \citet{miratrix2022using} simulating counterfactual trajectories around a linear prediction model and \citet{felton2023} imputing counterfactuals from ARIMA fits within the potential outcomes framework. The GP learns temporal dependence through the kernel, and the period-specific estimand $\tau_{it^*}$ traces how effects unfold.

The \texttt{CausalImpact} method \citep{brodersen2015inferring} employs Bayesian structural time series with latent state-space representations. As \citet{hazlett2018trajectory} show, \texttt{CausalImpact} assumes linearity in prior outcomes (LPO), that the counterfactual is linear in the pre-treatment outcomes. The GP does not escape this, since its posterior mean is a linear \textit{smoother}, $\hat{g}(\mathbf{X}_{t^*}) = \sum_t w_t Y_{it}$. What sets it apart is that the weights are fixed by the kernel over the inputs, so the estimator recovers any $g_i \in \mathcal{H}_k$ rather than a fixed state-space form, and the posterior variance widens off-support where the state-space intervals do not (Section~\ref{sec:theory}). The kernel relaxes the trend form, not the linearity in outcomes. \texttt{C-ARIMA} \citep{menchetti2023combining} brings ARIMA machinery into the potential outcomes framework and assumes stationarity of the differenced series; the GP places stationarity on the covariance rather than on the process, a weaker requirement.

\paragraph{Related Identification Strategies.} Under simultaneous treatment, DiD and its modern variants lack the contemporaneous controls that parallel trends compares, and synthetic control lacks donors that remain untreated. Where parallel trends governs cross-unit comparisons, Assumption~\ref{assum:sufficiency} governs within-unit temporal stability. SC weights control units to match the pre-treatment trajectory under stable factor loadings, so that the mapping from unobserved factors to outcomes survives treatment; our approach runs the same logic temporally, asking the unit-specific CEF $g_i$ learned from a unit's own history to remain stable past $t_0$, with this temporal stability substituting for the missing donors.

Regression discontinuity in time (RDiT) targets only the immediate effect at the temporal threshold, buying strong internal validity at the cutoff at the price of missing effects that unfold gradually \citep{hausman2018regression}, and it faces difficulties when observations cluster at discrete time points \citep{de2016misunderstandings, bertoli2026analyzing}. Our framework estimates $\tau_{it^*}$ at each post-treatment period and traces how effects evolve or decay.

\section{Concluding Remarks} \label{sec:conclusion}

Universal treatment removes the contemporaneous controls that difference-in-differences and synthetic control rely on, leaving an ITS that must extrapolate the counterfactual from the unit's own history. We show that the GP bounds the worst-case learning error for any counterfactual the kernel represents within a fixed complexity budget, so the reported interval widens with the extrapolation it carries. The approach rests on three components: an identification result stating what the ITS design can and cannot recover, an estimator whose predictive uncertainty grows with extrapolation distance, and an RKHS reading that gives the reported interval its frequentist content. Across the applications the estimator recovers established findings while tracing post-treatment dynamics that a single level shift cannot reveal, and an open-source \texttt{R} implementation accompanies the paper.

Several limitations deserve emphasis. Performance depends on the kernel, which encodes how the counterfactual could evolve beyond the data; deeper extrapolation leans on the non-stationary components, where the worst-case bound scales with $\|g^*\|_{\mathcal{H}_k}$, a quantity the data cannot estimate. The source-condition bounds tighten when the target aligns with the kernel's leading directions, an order likewise not estimated. Fitting each series on its own forgoes the efficiency that cross-sectional structure could supply, which hierarchical or multi-task formulations \citep{ben2023estimating} could recover for short panels at the cost of assumptions about cross-unit similarity. Computation is a further constraint, since inverting the kernel matrix scales poorly beyond a few thousand observations; sparse approximations \citep[e.g.,][]{titsias2009variational} may extend the method to long daily series.

Two extensions merit exploration. The distance-based covariance accommodates irregularly spaced observations without modification, a common feature of administrative and judicial records; our applications use evenly spaced series, and a formal study of performance under missing observations would be valuable. More broadly, the logic developed here applies wherever a credible counterfactual must be extrapolated from a unit's own history and the conclusions drawn from it must remain extrapolation-aware.

\paragraph{Data availability}
The data and code for replication are available at
\href{https://github.com/soonhong-cho/gpits}{https://github.com/soonhong-cho/gpits}. The replication package will be deposited in the Harvard Dataverse upon acceptance.

\begin{singlespacing}
\putbib[bib] 
\end{singlespacing}
\end{bibunit}

\newpage
\appendix
\begin{bibunit}[apalike]

\clearpage
\thispagestyle{empty} 

\vspace{1.5em}
\begin{center}
{\Large\bf Supplementary Material for:\\[0.5em]
{\large\bf ``Let Time Tell: Identification and Gaussian Process Estimation for Interrupted Time Series''}}
\end{center}

\begin{singlespacing}
\startcontents[app]
\printcontents[app]{l}{1}{}
\end{singlespacing}

\clearpage

\setcounter{page}{1}
\setcounter{section}{0}
\setcounter{footnote}{0}

\counterwithin{figure}{section}
\counterwithin{table}{section}
\counterwithin{equation}{section}

\counterwithin{theorem}{section}
\counterwithin{proposition}{section}
\counterwithin{definition}{section}
\counterwithin{corollary}{section}
\counterwithin{assumption}{section}

\renewcommand{\thesection}{\Alph{section}}
\renewcommand{\theequation}{\Alph{section}.\arabic{equation}}
\renewcommand{\thefigure}{\Alph{section}.\arabic{figure}}
\renewcommand{\thetable}{\Alph{section}.\arabic{table}}

\renewcommand{\thetheorem}{\Alph{section}.\arabic{theorem}}
\renewcommand{\theproposition}{\Alph{section}.\arabic{proposition}}
\renewcommand{\thedefinition}{\Alph{section}.\arabic{definition}}
\renewcommand{\thecorollary}{\Alph{section}.\arabic{corollary}}
\renewcommand{\theassumption}{\Alph{section}.\arabic{assumption}}
\sloppy

\newpage

\section{Proof of Theorem~\ref{theorem:identification}}
\label{app:id_proof}

\begin{proof}[Proof of Theorem~\ref{theorem:identification}]
Fix a unit $i$ and a post-treatment period $t^* \geq t_0$. Since $D_{it^*} = 1$, consistency~\eqref{eq:consistency} gives $Y_{it^*} = Y_{it^*}(1)$, so the individual effect $\delta_{it^*} = Y_{it^*}(1) - Y_{it^*}(0)$ has its treated outcome observed and its untreated outcome missing. Replacing the unobserved $Y_{it^*}(0)$ by its conditional mean gives the target estimand~\eqref{eq:tau}, and
\begin{equation}
    \tau_{it^*} = Y_{it^*}(1) - \mathbb{E}[Y_{it^*}(0) \mid \mathbf{X}_{it^*}]
    = Y_{it^*} - \mathbb{E}[Y_{it^*}(0) \mid \mathbf{X}_{it^*}]
    = Y_{it^*} - g_i(\mathbf{X}_{it^*}),
    \label{eq:id_reduction}
\end{equation}
so it remains to identify $g_i(\mathbf{X}_{it^*}) = \mathbb{E}[Y_{it^*}(0) \mid \mathbf{X}_{it^*}]$ from the observed data.

The counterfactual conditional mean coincides with $g_i$ by definition. By the structural model~\eqref{eq:structural} and the linearity of conditional expectation,
\begin{align}
    \mathbb{E}[Y_{it^*}(0) \mid \mathbf{X}_{it^*}]
    &= \mathbb{E}\bigl[h_i(\mathbf{X}_{it^*}, \mathbf{U}_{it^*}) + \varepsilon_{it^*} \mid \mathbf{X}_{it^*}\bigr] \notag \\
    &= \mathbb{E}\bigl[h_i(\mathbf{X}_{it^*}, \mathbf{U}_{it^*}) \mid \mathbf{X}_{it^*}\bigr]
       + \mathbb{E}\bigl[\varepsilon_{it^*} \mid \mathbf{X}_{it^*}\bigr] \notag \\
    &= \int_{\mathcal{U}} h_i(\mathbf{X}_{it^*}, \mathbf{u})\, dP(\mathbf{U}_{it^*} = \mathbf{u} \mid \mathbf{X}_{it^*})
       \;=\; g_i(\mathbf{X}_{it^*}),
    \label{eq:id_cmidentity}
\end{align}
where $\mathbb{E}[\varepsilon_{it^*} \mid \mathbf{X}_{it^*}] = \mathbb{E}\bigl[\mathbb{E}[\varepsilon_{it^*} \mid \mathbf{X}_{it^*}, \mathbf{U}_{it^*}] \mid \mathbf{X}_{it^*}\bigr] = 0$ by iterated expectations applied to~\eqref{eq:structural}, and the last equality is the definition~\eqref{eq:target}. Identity~\eqref{eq:id_cmidentity} holds for any model of the form~\eqref{eq:structural}, with or without Assumption~\ref{assum:sufficiency}; it does not identify $g_i(\mathbf{X}_{it^*})$, since its right-hand side depends on the unobserved distribution $P(\mathbf{U}_{it^*} \mid \mathbf{X}_{it^*})$ and at $t^*$ only $Y_{it^*} = Y_{it^*}(1)$ is observed.

Assumption~\ref{assum:sufficiency} supplies the link. For each $t$ it states $\mathbb{E}[Y_{it}(0) \mid \mathbf{X}_{it}, \mathbf{U}_{it}] = \mathbb{E}[Y_{it}(0) \mid \mathbf{X}_{it}]$; the left-hand side equals $h_i(\mathbf{X}_{it}, \mathbf{U}_{it})$ and the right-hand side equals $g_i(\mathbf{X}_{it})$, so
\begin{equation}
    h_i(\mathbf{X}_{it}, \mathbf{U}_{it}) = g_i(\mathbf{X}_{it})
    \quad \text{a.s.}, \quad \forall\, t.
    \label{eq:id_hgas}
\end{equation}
Because $h_i$ enters~\eqref{eq:structural} with a form not depending on $t$ (Section~\ref{sec:identification}), the $g_i$ in~\eqref{eq:id_hgas} is common across periods, so the systematic part of $Y_{it}(0)$ depends on $\mathbf{X}_{it}$ alone and is identical before and after $t_0$.

It remains to recover $g_i$ from pre-treatment data. For $t < t_0$, consistency~\eqref{eq:consistency} and no anticipation~\eqref{eq:no_anticipation} give $Y_{it} = Y_{it}(0)$, so by~\eqref{eq:structural} and~\eqref{eq:id_hgas},
\begin{equation}
    Y_{it} = h_i(\mathbf{X}_{it}, \mathbf{U}_{it}) + \varepsilon_{it}
    = g_i(\mathbf{X}_{it}) + \varepsilon_{it},
    \quad \mathbb{E}[\varepsilon_{it} \mid \mathbf{X}_{it}] = 0,
    \quad t < t_0.
    \label{eq:id_pre}
\end{equation}
This is a nonparametric regression of the observed $Y_{it}$ on $\mathbf{X}_{it}$ with regression function $g_i$, so $g_i$ is identified on the support of $\mathbf{X}_{i,\text{pre}}$. Since~\eqref{eq:id_hgas} holds at $t^*$ with the same $g_i$, evaluating this function at $\mathbf{X}_{it^*}$ returns $g_i(\mathbf{X}_{it^*}) = \mathbb{E}[Y_{it^*}(0) \mid \mathbf{X}_{it^*}]$ by~\eqref{eq:id_cmidentity}, and substituting into~\eqref{eq:id_reduction} identifies $\tau_{it^*}$.
\end{proof}
\noindent The recovery of $g_i$ in~\eqref{eq:id_pre} is on the support of the pre-treatment inputs, whereas $\mathbf{X}_{it^*}$ lies beyond it in the time coordinate; evaluating $g_i$ at $t^*$ therefore extrapolates. The residual indeterminacy this leaves is the extrapolation uncertainty $\omega_B(\mathbf{X}_{it^*}) = 2B\,\|k_{t^*}^\perp\|_{\mathcal{H}_k}$ of Definition~\ref{def:extrap_uncertainty} (Corollary~\ref{cor:extrap_uncertainty}), and the reported band is what carries it. No parametric form is imposed on $g_i$; segmented regression, ARIMA, and Bayesian structural time series share this identification argument, differing only in the restrictions they place on $g_i$.

\section{Connections to RKHS Theory}
\label{app:rkhs_theory}

This appendix proves the results in Section~\ref{sec:theory} of the main text. For the link between GPs and RKHS theory, see \citet[Ch.~6]{rasmussen2006gaussian} and \citet{kanagawa2025gaussian}. We suppress the unit subscript $i$ throughout for notational simplicity. Proposition~\ref{prop:bias} is immediate, since taking the conditional expectation of the four-term decomposition~\eqref{eq:four_term} and applying the drop-out argument of Section~\ref{sec:theory} leaves only the approximation and learning terms. The substantive result is Proposition~\ref{prop:worst_case}, proved below.

\subsection{RKHS Preliminaries}

Let $k: \mathcal{X} \times \mathcal{X} \to \mathbb{R}$ be a positive definite kernel. Its reproducing kernel Hilbert space $\mathcal{H}_k$ is the completion of $\operatorname{span}\{k(\cdot, \mathbf{x}) : \mathbf{x} \in \mathcal{X}\}$ under the inner product determined by
\begin{equation}
    \langle k(\cdot, \mathbf{x}),\, k(\cdot, \mathbf{x}') \rangle_{\mathcal{H}_k} = k(\mathbf{x}, \mathbf{x}'),
    \label{eq:rkhs_inner}
\end{equation}
for all $\mathbf{x}, \mathbf{x}' \in \mathcal{X}$. This inner product yields the reproducing property: for any $f \in \mathcal{H}_k$,
\begin{equation}
    f(\mathbf{x}) = \langle f,\, k(\cdot, \mathbf{x}) \rangle_{\mathcal{H}_k},
    \label{eq:reproducing}
\end{equation}
so point evaluation is the inner product of $f$ with $k(\cdot, \mathbf{x})$. Applying Cauchy--Schwarz to~\eqref{eq:reproducing} gives
\begin{equation}
    |f(\mathbf{x})| \leq \|f\|_{\mathcal{H}_k}\, \sqrt{k(\mathbf{x}, \mathbf{x})},
    \label{eq:cauchy_schwarz}
\end{equation}
which establishes that point evaluation is a bounded linear functional on $\mathcal{H}_k$ and underlies the worst-case bound of Proposition~\ref{prop:worst_case}.

The norm $\|f\|_{\mathcal{H}_k} = \sqrt{\langle f, f \rangle_{\mathcal{H}_k}}$ measures the complexity of $f$ relative to the kernel, with functions aligned with its covariance structure carrying a small norm and functions requiring features the kernel does not naturally accommodate carrying a large one. It is distinct from the $L^2(P_{\mathbf{X}})$ norm $\|f\|_{L^2}^2 = \int f(\mathbf{x})^2 \, dP_{\mathbf{X}}(\mathbf{x})$, which measures average size under the marginal distribution of $\mathbf{X}$. The regularity that $g \in \mathcal{H}_k$ requires depends on the kernel. For the Gaussian kernel $\mathcal{H}_k$ consists of smooth functions, and for the linear kernel $k(\mathbf{x}, \mathbf{x}') = \mathbf{x}^\top \mathbf{x}'$ it consists of linear functions with $\|f\|_{\mathcal{H}_k}$ equal to the Euclidean norm of the coefficient vector.

The best approximation within budget $g^*$ of~\eqref{eq:gstar} minimises the $L^2(P_{\mathbf{X}})$ distance to $g$ subject to $\|f\|_{\mathcal{H}_k} \le B$, and the factor entering the worst-case bound is $\|g^*\|_{\mathcal{H}_k}$, which the budget caps at $B$. The minimiser exists and is unique, since for a bounded continuous kernel on a compact domain the inclusion $\mathcal{H}_k \hookrightarrow L^2(P_{\mathbf{X}})$ is compact, so the constraint set is convex and closed in $L^2(P_{\mathbf{X}})$. Dropping the constraint removes both properties whenever $\mathcal{H}_k$ is dense in $L^2(P_{\mathbf{X}})$, since the infimum is then zero and unattained.

The pre-treatment inputs span a finite-dimensional subspace of $\mathcal{H}_k$,
\begin{equation}
    \mathcal{S}_{\text{pre}} = \operatorname{span}\{k(\cdot, \mathbf{X}_1), \ldots, k(\cdot, \mathbf{X}_{t_0-1})\}.
\end{equation}
By the representer theorem \citep{kimeldorf1970correspondence}, the minimiser of the regularised empirical risk over $\mathcal{H}_k$ lies in $\mathcal{S}_{\text{pre}}$, and for squared loss with regularisation $\sigma^2$ the GP posterior mean takes the form
\begin{equation}
    \hat{g} = \sum_{t=1}^{t_0-1} w_t k(\cdot, \mathbf{X}_t), \quad \boldsymbol{w} = (\mathbf{K}_{\text{pre}} + \sigma^2 \mathbf{I})^{-1} Y_{\text{pre}},
    \label{eq:representer_applied}
\end{equation}
recovering the kernel ridge regression estimator \citep[Ch.~6]{rasmussen2006gaussian}. The regularisation parameter $\sigma^2$ in~\eqref{eq:representer_applied} penalises the RKHS norm, pulling the estimate toward functions the kernel favours; in dual form it shrinks the coefficients $\boldsymbol{w}$, but the effect in function space is smoothing rather than the sparsity an $\ell_1$ penalty in a fixed basis would deliver \citep{ghosal2017fundamentals}. The learning error arises because $\hat{g}$ is confined to $\mathcal{S}_{\text{pre}}$ while $g^*$ may have a component orthogonal to it, which the observed covariate values cannot recover.

\subsection{Proof of Proposition~\ref{prop:worst_case}}

We prove Proposition~\ref{prop:worst_case} stated in Section~\ref{sec:theory} of the main text. The proof follows results in \citet{kanagawa2025gaussian}, adapted to our setting.

\begin{proposition}[Proposition 2 (GP posterior variance controls the learning error)] \label{prop:worst_case_restate}
Suppose $g$ lies in the unit ball $\mathcal{F}_1$, so that $g^* = g$, and
suppose outcomes are observed with noise, $Y_t = g^*(\mathbf{X}_t) + \varepsilon_t$, where $\varepsilon_t$ are mean-zero with variance $\sigma^2 > 0$. The GP posterior mean is the weighting estimator $\hat g(\mathbf{X}_{t^*}) = \sum_t w_t Y_t$ with weights $\mathbf{w} = (\mathbf{K}_{\mathrm{pre}} + \sigma^2 I)^{-1}\mathbf{k}_{t^*}$, and its posterior variance is
\begin{equation*}
    \mathbb{V}_{t^*} = k(\mathbf{X}_{t^*},\mathbf{X}_{t^*}) - \mathbf{k}_{t^*}^\top(\mathbf{K}_{\mathrm{pre}} + \sigma^2 I)^{-1}\mathbf{k}_{t^*}.
\end{equation*}
Let $k_{t^*}^\perp = k(\cdot,\mathbf{X}_{t^*}) - \sum_t w_t\, k(\cdot,\mathbf{X}_t)$ be the residual of the test feature after its ridge reconstruction from the training features $\mathcal{S}_{\mathrm{pre}} = \mathrm{span}\{k(\cdot,\mathbf{X}_t)\}$, coinciding with the orthogonal-projection residual onto $\mathcal{S}_{\mathrm{pre}}$ as $\sigma^2 \to 0$. Then the learning error decomposes as
\begin{equation*}
    g^*(\mathbf{X}_{t^*}) - \hat g(\mathbf{X}_{t^*})
    = \underbrace{\Big(g^*(\mathbf{X}_{t^*}) - \textstyle\sum_t w_t\,g^*(\mathbf{X}_t)\Big)}_{\text{deterministic error}}
    \;-\; \underbrace{\textstyle\sum_t w_t\,\varepsilon_t}_{\text{noise propagation}},
\end{equation*}
and its two sources are controlled by the two summands of the GP posterior variance:
\begin{equation*}
    \Big|g^*(\mathbf{X}_{t^*}) - \sum_t w_t g^*(\mathbf{X}_t)\Big| \le \|k_{t^*}^\perp\|_{\mathcal{H}_k}, \quad
    \mathrm{Var}\big[\sum_t w_t \varepsilon_t\big] = \sigma^2\|\mathbf{w}\|^2,
    \qquad
    \underbrace{\|k_{t^*}^\perp\|_{\mathcal{H}_k}^2 + \sigma^2\|\mathbf{w}\|^2 = \mathbb{V}_{t^*}}_{\text{GP posterior variance}}.
\end{equation*}
Therefore $\mathbb{V}_{t^*}$ bounds the worst case of the deterministic error through its first summand and carries the exact variance of the propagated noise in its second.
\end{proposition}

\begin{proof}[Proof of Proposition~\ref{prop:worst_case_restate}]
Let $\mathbf{w} = (\mathbf{K}_{\text{pre}} + \sigma^2 \mathbf{I})^{-1} \mathbf{k}_{t^*}$ denote the posterior mean weights, where $\mathbf{k}_{t^*} = (k(\mathbf{X}_1, \mathbf{X}_{t^*}), \ldots, k(\mathbf{X}_{t_0-1}, \mathbf{X}_{t^*}))^\top$. By~\eqref{eq:representer_applied}, $\hat{g}(\mathbf{X}_{t^*}) = \sum_t w_t Y_t$. Recall the residual feature
\begin{equation} 
    k_{t^*}^\perp = k(\cdot, \mathbf{X}_{t^*}) - \sum_{t=1}^{t_0-1} w_t k(\cdot, \mathbf{X}_t), 
\end{equation}
the residual of the test feature after the ridge reconstruction. As $\sigma^2 \to 0$ the weights become the orthogonal projection coefficients $\mathbf{K}_{\mathrm{pre}}^{-1}\mathbf{k}_{t^*}$ and $\|k_{t^*}^\perp\|_{\mathcal{H}_k}$ reduces to the geometric distance between $\mathbf{X}_{t^*}$ and the training data; for $\sigma^2 > 0$ it carries an additional dependence on $\sigma^2$ through
$\mathbf{w}$; Corollary~\ref{cor:noiseless} states the limiting case.

By the reproducing property,
\begin{equation}
    g^*(\mathbf{X}_{t^*}) - \sum_t w_t g^*(\mathbf{X}_t) = \langle g^*,\, k_{t^*}^\perp \rangle_{\mathcal{H}_k}.
\end{equation}
Applying Cauchy--Schwarz:
\begin{equation}
    \bigl|g^*(\mathbf{X}_{t^*}) - \sum_t w_t g^*(\mathbf{X}_t)\bigr| \leq \|g^*\|_{\mathcal{H}_k} \cdot \|k_{t^*}^\perp\|_{\mathcal{H}_k}.
    \label{eq:cauchy_schwarz_bound}
\end{equation}
It remains to bound $\|k_{t^*}^\perp\|_{\mathcal{H}_k}$. Expanding via the reproducing property and bilinearity:
\begin{equation}
    \|k_{t^*}^\perp\|_{\mathcal{H}_k}^2 = k(\mathbf{X}_{t^*}, \mathbf{X}_{t^*}) - 2\mathbf{k}_{t^*}^\top \mathbf{w} + \mathbf{w}^\top \mathbf{K}_{\text{pre}} \mathbf{w}.
\end{equation}
Substituting $\mathbf{K}_{\text{pre}} = (\mathbf{K}_{\text{pre}} + \sigma^2 \mathbf{I}) - \sigma^2 \mathbf{I}$ and using the definition of $\mathbf{w}$:
\begin{equation}
    \mathbf{w}^\top \mathbf{K}_{\text{pre}} \mathbf{w} = \mathbf{k}_{t^*}^\top (\mathbf{K}_{\text{pre}} + \sigma^2 \mathbf{I})^{-1} \mathbf{k}_{t^*} - \sigma^2 \|\mathbf{w}\|^2.
\end{equation}
Since $\mathbf{k}_{t^*}^\top \mathbf{w} = \mathbf{k}_{t^*}^\top (\mathbf{K}_{\text{pre}} + \sigma^2 \mathbf{I})^{-1} \mathbf{k}_{t^*}$, combining terms yields
\begin{equation}
    \|k_{t^*}^\perp\|_{\mathcal{H}_k}^2 = \underbrace{k(\mathbf{X}_{t^*}, \mathbf{X}_{t^*}) - \mathbf{k}_{t^*}^\top (\mathbf{K}_{\text{pre}} + \sigma^2 \mathbf{I})^{-1} \mathbf{k}_{t^*}}_{\displaystyle \mathbb{V}[f(\mathbf{X}_{t^*}) \mid \mathbf{X}_{\text{pre}}, Y_{\text{pre}}]} - \underbrace{\sigma^2 \|\mathbf{w}\|^2}_{\geq\, 0}.
    \label{eq:geometric_decomp}
\end{equation}
Combining~\eqref{eq:cauchy_schwarz_bound} with~\eqref{eq:geometric_decomp} and $\|g^*\|_{\mathcal{H}_k} \le 1$ gives the deterministic bound $|g^*(\mathbf{X}_{t^*}) - \sum_t w_t g^*(\mathbf{X}_t)| \le \|k_{t^*}^\perp\|_{\mathcal{H}_k}$, while the noise propagation $\sum_t w_t\varepsilon_t$ is mean-zero with variance $\sigma^2\|\mathbf{w}\|^2$. Rearranging~\eqref{eq:geometric_decomp} gives $\|k_{t^*}^\perp\|_{\mathcal{H}_k}^2 + \sigma^2\|\mathbf{w}\|^2 = \mathbb{V}_{t^*}$, so the two summands of the posterior variance are exactly the two scales just bounded, and in particular $\|k_{t^*}^\perp\|_{\mathcal{H}_k} \le \sqrt{\mathbb{V}_{t^*}}$. This completes the proof.
\end{proof}
\noindent Under kernel misspecification the GP posterior variance can understate the true prediction error \citep{beckers2018mean}, but the worst-case bound remains valid for any $g^* \in \mathcal{H}_k$. Figure~\ref{fig:rkhs_decomposition} illustrates the decomposition.

The noiseless case isolates the geometric leg as a result of its own.
\begin{corollary}[Noiseless training]\label{cor:noiseless}
Suppose the pre-treatment outcomes are observed without noise, $Y_t = g^*(\mathbf{X}_t)$ for $t < t_0$, with $\mathbf{K}_{\mathrm{pre}}$ invertible. In the limit $\sigma^2 \to 0$ of Proposition~\ref{prop:worst_case}, the weights become the interpolation weights $\mathbf{w} = \mathbf{K}_{\mathrm{pre}}^{-1}\mathbf{k}_{t^*}$, the residual $k_{t^*}^\perp$ is the orthogonal projection residual of $k(\cdot, \mathbf{X}_{t^*})$ onto $\mathcal{S}_{\mathrm{pre}}$, and the learning error is deterministic with
\begin{equation}\label{eq:noiseless_bound}
    \bigl| g^*(\mathbf{X}_{t^*}) - \hat g(\mathbf{X}_{t^*}) \bigr|
    \;\le\; \|k_{t^*}^\perp\|_{\mathcal{H}_k}
    \;=\; \sqrt{\mathbb{V}_{t^*}},
\end{equation}
with equality attained at $g^* \propto k_{t^*}^\perp$.
\end{corollary}
\noindent The corollary follows by setting $\sigma^2 = 0$ in~\eqref{eq:geometric_decomp}, whereupon the noise propagation term vanishes and the posterior variance reduces to $\|k_{t^*}^\perp\|_{\mathcal{H}_k}^2$ alone. The pre-treatment data pins the counterfactual down wherever it lies in $\mathcal{S}_{\mathrm{pre}}$, and what the record cannot see is at once the entire error and the entire posterior standard deviation.

\begin{figure}[hbt!]
    \centering
    \includegraphics[width=1\textwidth]{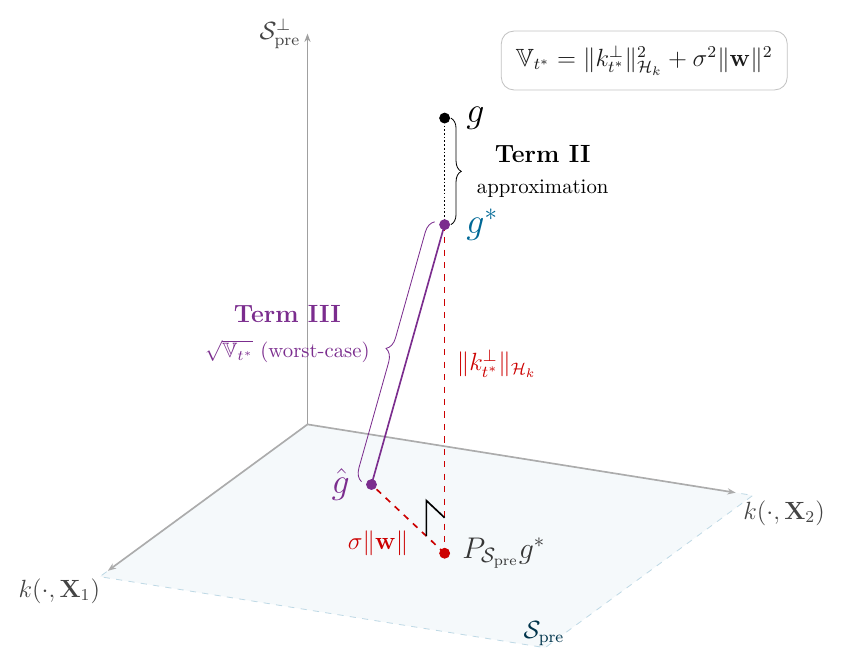}
    \caption{Schematic of the error decomposition at $t^*$ in $\mathcal{H}_k$. \textit{Note:} The plane is $\mathcal{S}_{\text{pre}} = \operatorname{span}\{k(\cdot,\mathbf{X}_1),\ldots,k(\cdot,\mathbf{X}_{t_0-1})\}$. The true CEF $g$ lies outside $\mathcal{H}_k$, its best RKHS approximation $g^*$ is estimated by $\hat{g}\in\mathcal{S}_{\text{pre}}$, the gap $g - g^*$ is the approximation error (Term~II), and $g^* - \hat{g}$ is the learning error (Term~III), which splits into a component orthogonal to $\mathcal{S}_{\text{pre}}$ and a noise component within it. Over $\|g^*\|_{\mathcal{H}_k}\le1$ the geometric leg is longest for the target aligned with $k_{t^*}^\perp$, drawn here with length $\|k_{t^*}^\perp\|_{\mathcal{H}_k}$; a generic target gives the smaller $\langle g^*,k_{t^*}^\perp\rangle$. With the noise leg $\sigma\|\mathbf{w}\|$ the two legs combine into the hypotenuse, $\mathbb{V}_{t^*}$, when $\sigma^2=0$ and approximately otherwise.}
    \label{fig:rkhs_decomposition}
\end{figure}
Proposition~\ref{prop:worst_case} is sharp over the unit ball rather than merely an upper bound, since $\|k_{t^*}^\perp\|_{\mathcal{H}_k} = \sup_{\|g^*\|_{\mathcal{H}_k}\le 1}|\langle g^*, k_{t^*}^\perp\rangle_{\mathcal{H}_k}|$, attained at $g^* \propto k_{t}^\perp$. Figure~\ref{fig:app_sim_attainment} reads this in function space, drawing a family of least favourable targets that meet the band, each at one horizon.

\begin{figure}[hbt!]
    \centering
    \includegraphics[width=0.9\textwidth]{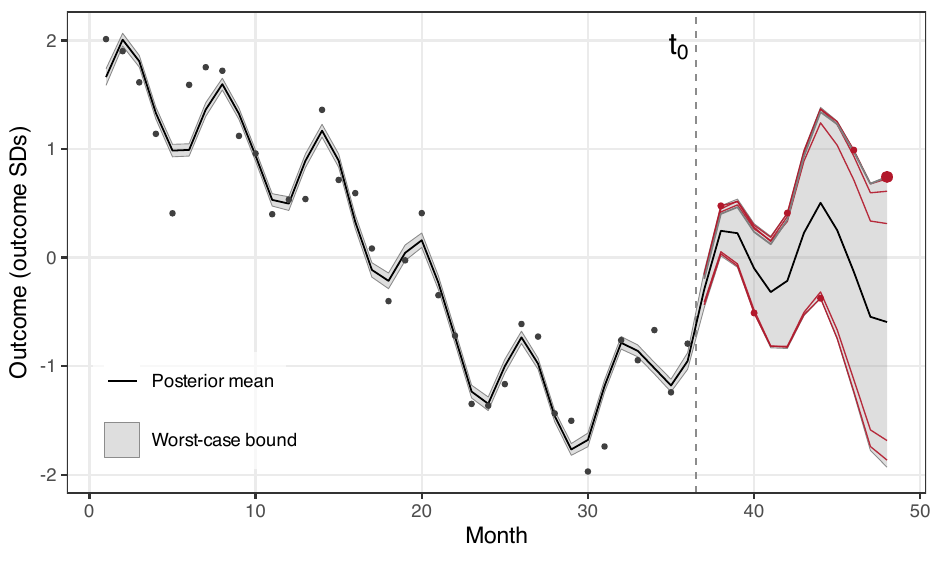}
    \caption{Attainment in Proposition~\ref{prop:worst_case}, in function space. The GP posterior mean (black) interpolates the noisy pre-treatment series (points) and extrapolates past $t_0$ (dashed). The shaded band is the geometric term $\|k_{t}^\perp\|_{\mathcal{H}_k}$ of the posterior variance, the largest error any counterfactual with $\|g^*\|_{\mathcal{H}_k}\le 1$ agreeing with the pre-treatment series can incur at $t^*$, so it is the attained worst case rather than a loose envelope, exactly under noiseless training (Corollary~\ref{cor:noiseless}) and approximately under regularisation. The bold red curve is the least favourable admissible counterfactual at the furthest horizon, meeting the band edge at $t^*$ (point); the thin curves meet it at their own horizons ($t^*=2, 4, 6, 8, 10$), so the bound is attained pointwise across the window.}
    \label{fig:app_sim_attainment}
\end{figure}

\paragraph{Minimax Optimality at the Extrapolation Point.}
Proposition~\ref{prop:worst_case} and Corollary~\ref{cor:noiseless} bound what the GP posterior mean loses against the least favourable target. A stronger question is whether any other estimator could lose less. To make this precise, consider noiseless training and let $D(g) = (g(\mathbf{X}_1), \ldots, g(\mathbf{X}_{t_0-1}))$ collect the pre-treatment evaluations of $g$. An estimator of $g(\mathbf{X}_{t^*})$ is any map $\varphi \colon \mathbb{R}^{t_0-1} \to \mathbb{R}$ that reads the pre-treatment data and returns a single number, an unrestricted class that includes the GP posterior mean, kernel ridge regression, and segmented regression. The minimax risk at the extrapolation point is
\begin{equation}\label{eq:minimax_risk}
    \mathfrak{R}(\mathbf{X}_{t^*})
    \;=\;
    \inf_{\varphi}\,
    \sup_{\|g\|_{\mathcal{H}_k} \le 1}\,
    \bigl| g(\mathbf{X}_{t^*}) - \varphi(D(g)) \bigr|.
\end{equation}
The infimum searches over all estimators; the supremum searches over all regression functions in the unit ball. The question is what the cleverest possible use of the pre-treatment data can guarantee against the hardest possible target.

\begin{proposition}[Minimax optimality of the noiseless band]
\label{prop:minimax}
Under the conditions of Corollary~\ref{cor:noiseless}, with $\|k_{t^*}^\perp\|_{\mathcal{H}_k} > 0$,
\begin{equation*}
    \mathfrak{R}(\mathbf{X}_{t^*})
    \;=\; \|k_{t^*}^\perp\|_{\mathcal{H}_k}
    \;=\; \sqrt{\mathbb{V}_{t^*}},
\end{equation*}
and the GP posterior mean $\hat{g}(\mathbf{X}_{t^*}) = \mathbf{k}_{t^*}^\top \mathbf{K}_{\mathrm{pre}}^{-1} Y_{\mathrm{pre}}$ attains the infimum. Over the ball of budget $B$ the risk scales to $B\,\|k_{t^*}^\perp\|_{\mathcal{H}_k}$, attained by the same estimator.
\end{proposition}

\begin{proof}
\textit{Lower bound.}
Set $g_\pm = \pm\, k_{t^*}^\perp / \|k_{t^*}^\perp\|_{\mathcal{H}_k}$. Both have unit norm, so both lie in the unit ball. Because $k_{t^*}^\perp$ is orthogonal to $\mathcal{S}_{\mathrm{pre}}$, both functions vanish at every pre-treatment input: $D(g_+) = D(g_-) = \mathbf{0}$. Any estimator/mapping $\varphi$ therefore returns the same value $d = \varphi(\mathbf{0})$ for both. Yet at the extrapolation point the two functions separate. By the reproducing property and the orthogonal decomposition $k(\cdot, \mathbf{X}_{t^*}) = k_{t^*}^\perp + P_{\mathcal{S}_{\mathrm{pre}}}\, k(\cdot, \mathbf{X}_{t^*})$,
\begin{equation*}
    g_\pm(\mathbf{X}_{t^*})
    \;=\; \frac{\langle k_{t^*}^\perp,\, k(\cdot,
    \mathbf{X}_{t^*})\rangle_{\mathcal{H}_k}}
    {\pm\|k_{t^*}^\perp\|_{\mathcal{H}_k}}
    \;=\; \pm\,\|k_{t^*}^\perp\|_{\mathcal{H}_k}.
\end{equation*}
The two values sit at $\pm\|k_{t^*}^\perp\|_{\mathcal{H}_k}$, so for every $d \in \mathbb{R}$,
\begin{equation*}
    \max\bigl\{\,|g_+(\mathbf{X}_{t^*}) - d|,\;
    |g_-(\mathbf{X}_{t^*}) - d|\,\bigr\}
    \;\ge\; \|k_{t^*}^\perp\|_{\mathcal{H}_k},
\end{equation*}
and the floor $\mathfrak{R}(\mathbf{X}_{t^*}) \ge \|k_{t^*}^\perp\|_{\mathcal{H}_k}$ follows.

\textit{Upper bound.}
Corollary~\ref{cor:noiseless} states that the GP posterior mean satisfies $|g(\mathbf{X}_{t^*}) - \hat{g}(\mathbf{X}_{t^*})| \le \|k_{t^*}^\perp\|_{\mathcal{H}_k}$ for every $g$ in the unit ball. Its worst case therefore meets the floor, and the two bounds coincide. The budget-$B$ statement follows by replacing $g$ with $g/B$. The argument adapts the classical theory of optimal recovery \citep{golomb1959optimal, micchelli1977survey} to the extrapolation target of the ITS design.
\end{proof}

The pair $g_\pm$ is a thought experiment, not a claim about what the true regression function looks like. In the lower-bound argument the pre-treatment data happen to be zero, because $g_\pm$ are constructed to live entirely in the direction the training data cannot see. The upper bound, by contrast, holds for every function in the ball and every dataset, zero or not. Together the two bounds say that the GP posterior mean extracts everything the pre-treatment data contains about the extrapolation target, and the reported band is the residual indeterminacy that no procedure can remove.

\paragraph{Data-conditional sharpening.}\label{app:sharpening}
The unconditional minimax risk $\|k_{t^*}^\perp\|_{\mathcal{H}_k}$ is the worst case across all datasets. With a specific observed series $Y_{\mathrm{pre}}$, the argument tightens. Let $\hat{g}$ be the minimum-norm interpolant of $Y_{\mathrm{pre}}$, satisfying $\|\hat{g}\|_{\mathcal{H}_k} \le 1$. Any admissible function $g$ with $D(g) = Y_{\mathrm{pre}}$ and $\|g\|_{\mathcal{H}_k} \le 1$ decomposes as $g = \hat{g} + u$ with $u \perp \mathcal{S}_{\mathrm{pre}}$. By the Pythagorean identity, $\|u\|_{\mathcal{H}_k}^2 \le 1 - \|\hat{g}\|_{\mathcal{H}_k}^2$, so $u$ has less budget to spend in the unseen direction. The reproducing property gives $g(\mathbf{X}_{t^*}) - \hat{g}(\mathbf{X}_{t^*}) = \langle u, k_{t^*}^\perp \rangle_{\mathcal{H}_k}$, and applying Cauchy--Schwarz yields the identified set
\begin{equation*}
    \bigl\{\, g(\mathbf{X}_{t^*}) : \|g\|_{\mathcal{H}_k} \le 1,\; D(g) = Y_{\mathrm{pre}} \,\bigr\}
    \;=\;
    \bigl[\, \hat{g}(\mathbf{X}_{t^*}) \pm \sqrt{1 - \|\hat{g}\|_{\mathcal{H}_k}^2}\; \|k_{t^*}^\perp\|_{\mathcal{H}_k} \,\bigr],
\end{equation*}
with the endpoints attained at $u \propto \pm k_{t^*}^\perp$. Fitting the pre-treatment data consumes part of the RKHS budget, leaving less for the adversary to deploy at $t^*$. The reported band $\hat{g}(\mathbf{X}_{t^*}) \pm \|k_{t^*}^\perp\|_{\mathcal{H}_k}$ contains this identified set for every dataset, coincides with it in the least favourable one ($\|\hat{g}\|_{\mathcal{H}_k} = 0$), and the GP posterior mean sits at the midpoint throughout.

\paragraph{Noise does not restore learnability.}
With noisy training, $Y_t = g(\mathbf{X}_t) + \varepsilon_t$, the pair $g_\pm$ induces identical distributions of the observed data, since both conditional means vanish at every pre-treatment input. For any estimator $\varphi$, coupling the two DGPs and applying the triangle inequality gives $\sup_{\|g\|_{\mathcal{H}_k} \le 1} \mathbb{E}\,|g(\mathbf{X}_{t^*}) - \varphi(Y_{\mathrm{pre}})| \ge \|k_{t^*}^\perp\|_{\mathcal{H}_k}$, where $k_{t^*}^\perp$ is the ridge projection residual of Proposition~\ref{prop:worst_case}. Lengthening the pre-treatment data moves this floor only through the geometry of $\mathcal{S}_{\mathrm{pre}}$, not through the sample size: additional observations refine the in-support fit but leave the component orthogonal to $\mathcal{S}_{\mathrm{pre}}$ unidentified.

\subsection{RKHS of Kernel Sums, Source Condition, and Convergence Rates}\label{app:source_condition}

This section develops the machinery underlying Section~\ref{sec:theory}, namely the spectral representation that makes norm inflation precise (Section~\ref{app:spectral_rep}), norm control via kernel sums (Section~\ref{app:kernel_sums}), the source condition for kernel--target alignment (Section~\ref{app:source_condition_formal}), and convergence rates with their limits under extrapolation (Section~\ref{app:convergence_rates}).

\subsubsection{Spectral Representation of the RKHS Norm}\label{app:spectral_rep}

Let $P_{\mathbf{X}}$ be the marginal distribution of $\mathbf{X}$ and $T_k \colon L^2(P_{\mathbf{X}}) \to L^2(P_{\mathbf{X}})$ the kernel integral operator $T_k(f)(\mathbf{x}) = \int k(\mathbf{x}, \mathbf{x}') f(\mathbf{x}') \, dP_{\mathbf{X}}(\mathbf{x}')$. On a compact domain with a continuous kernel, $T_k$ is compact, positive, and self-adjoint, with Mercer decomposition
\begin{equation}
    k(\mathbf{x}, \mathbf{x}') = \sum_{j \geq 1} \mu_j \phi_j(\mathbf{x}) \phi_j(\mathbf{x}'),
\end{equation}
where $\mu_1 \geq \mu_2 \geq \cdots > 0$ are eigenvalues and $\{\phi_j\}$ are eigenfunctions orthonormal in $L^2(P_{\mathbf{X}})$ \citep{steinwart2008support}. The RKHS norm then admits the representation
\begin{equation}\label{eq:rkhs_norm_mercer}
    \|f\|_{\mathcal{H}_k}^2 = \sum_{j \geq 1} \frac{f_j^2}{\mu_j},
\end{equation}
with $f_j = \langle f, \phi_j \rangle_{L^2(P_{\mathbf{X}})}$ \citep{kanagawa2025gaussian}. Energy on eigenfunctions with large eigenvalues costs little, while energy on high-frequency eigenfunctions, those with small $\mu_j$, pays a steep penalty.

A stationary kernel carries a trend only by spreading it across many small-eigenvalue directions, and~\eqref{eq:rkhs_norm_mercer} charges for each.  The reproducing property fixes the total charge without computing the spectrum. Take $g(t) = \beta t$ over a window of length $T$ and let $h$ approximate it to within $\varepsilon$ in the supremum norm. Applying~\eqref{eq:reproducing} and Cauchy--Schwarz to the difference of two feature maps,
\begin{equation}\label{eq:gauss_norm_lower}
    \|h\|_{\mathcal{H}_{\text{Gaussian}}}
    \;\ge\; \frac{|h(T) - h(0)|}{\sqrt{2}}
    \;\ge\; \frac{|\beta|\,T - 2\varepsilon}{\sqrt{2}},
\end{equation}
since $\|k(\cdot,T) - k(\cdot,0)\|_{\mathcal{H}_k}^2 = 2 - 2k(0,T) \le 2$ for a unit-diagonal kernel, while $h$ must climb at least $|\beta| T - 2\varepsilon$ across the window. The cost of a trend is the distance it climbs, and no bandwidth avoids it: the bound reads the kernel only through its diagonal.

Two consequences follow.  At a fixed tolerance the cost grows with the length of the series.  As the tolerance shrinks the cost grows without bound, since approximants of bounded norm would place $g$ inside $\mathcal{H}_{\text{Gaussian}}$, and $g$ is not there on any window: members over $[0,T]$ extend to the whole line \citep{aronszajn1950theory} and are analytic there \citep{steinwart2006explicit}, so a member agreeing with $\beta t$ on the window agrees with it everywhere, yet~\eqref{eq:cauchy_schwarz} bounds every member and $\beta t$ is unbounded. Universality therefore buys approximation rather than membership, and charges for it by the length of the window. The Aronszajn decomposition~\eqref{eq:rkhs_sum_norm} removes the charge by carrying the trend in $\mathcal{H}_{\text{lin}}$, where $\|\beta t\|_{\mathcal{H}_{\text{lin}}}^2 = \beta^2$ whatever the window.

\begin{figure}[thb!]
    \centering
    \includegraphics[width=1\textwidth]{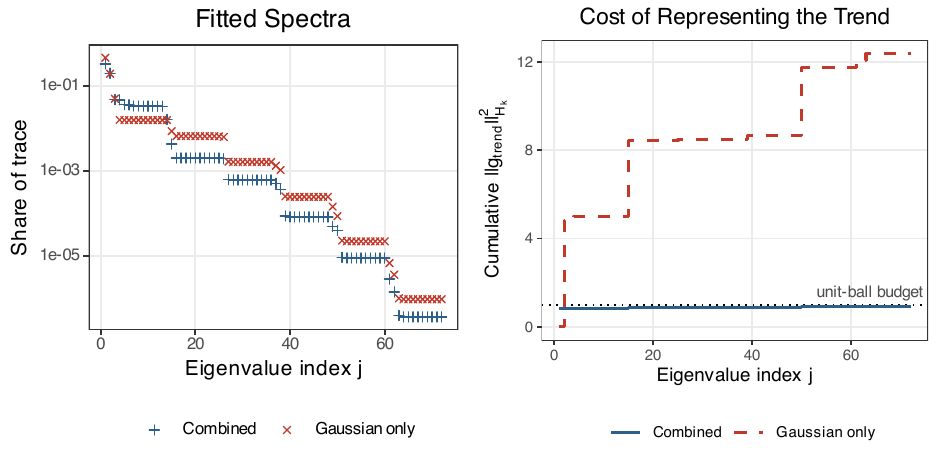}
    \caption{What the kernel sum buys, on the national \textit{Heller} pre-treatment series ($n = 72$ months).  Left panel plots the eigenvalues of each fitted kernel matrix as shares of its trace on a logarithmic scale, the prior variance the kernel assigns to each orthogonal direction on the design. Right panel accumulates over those same directions the squared RKHS norm a unit-variance trend costs, with the dotted line marking the unit-ball budget of Proposition~\ref{prop:worst_case}.  Under the combined kernel the trend is the leading direction, carrying thirty-four percent of the trace, and costs $0.924$ with ninety percent of that paid in the leading direction alone.  Under the Gaussian-only kernel the leading direction is orthogonal to the trend, which is assembled instead from many small directions at a cost of $12.4$.  In norm the budgets are $B = 0.96$ and $B = 3.52$.}
    \label{fig:app_scree}
\end{figure}
For a single kernel the cost of a trend can be bounded analytically, as in~\eqref{eq:gauss_norm_lower}; for a sum the norm is an infimum over decompositions~\eqref{eq:rkhs_sum_norm} and no such bound is available, so we read the fitted spectrum instead.  Figure~\ref{fig:app_scree} does this on the national \textit{Heller} pre-treatment series.  On an equally spaced monthly design the Gaussian kernel matrix is unchanged by reversing time, so each eigenvector is either symmetric or antisymmetric about the midpoint; the leading one is symmetric and the centred trend is antisymmetric, leaving the two orthogonal at any bandwidth. The trend is then assembled from many small directions, each charged by~\eqref{eq:rkhs_norm_mercer}, and the budget reaches $B = 3.52$, in line with the lower bound of~\eqref{eq:gauss_norm_lower}. The linear component supplies the missing direction, and the budget falls to $B = 0.96$, inside the unit ball. Universality is not what separates the two specifications, since a Gaussian kernel does represent the trend on this window; the price of doing so is.

\subsubsection{RKHS of Kernel Sums}\label{app:kernel_sums}

\begin{theorem}[\citealt{aronszajn1950theory}]
\label{thm:aronszajn}
For positive definite kernels $k_1, \ldots, k_m$ on $\mathcal{X}$, the RKHS of $k = k_1 + \cdots + k_m$ is
\begin{equation}
    \mathcal{H}_k = \bigl\{ f_1 + \cdots + f_m : f_j \in \mathcal{H}_{k_j} \bigr\},
\end{equation}
with norm
\begin{equation}\label{eq:rkhs_sum_norm}
    \|f\|_{\mathcal{H}_k}^2 = \inf\Bigl\{ \textstyle\sum_{j=1}^m \|f_j\|_{\mathcal{H}_{k_j}}^2 \;:\; f = \textstyle\sum_{j=1}^m f_j,\; f_j \in \mathcal{H}_{k_j} \Bigr\}.
\end{equation}
The infimum is attained. When the component RKHSs intersect only at $\{0\}$, the decomposition is unique and $\mathcal{H}_k$ is an orthogonal direct sum.
\end{theorem}
\noindent The infimum assigns each function the cost of its cheapest decomposition. For our kernel $k = k_{\text{Gaussian}} + k_{\text{per}} + k_{\text{lin}}$ the component RKHSs overlap on a compact domain, so the decomposition is not unique and the infimum is no larger than any particular feasible decomposition.

Suppose $g \in \mathcal{H}_k$ with decomposition $g = g_{\text{lin}} + g_{\text{per}} + g_{\text{Gaussian}}$. Theorem~\ref{thm:aronszajn} yields the bound
\begin{equation}\label{eq:norm_upper_bound}
    \|g\|_{\mathcal{H}_k}^2 \leq \|g_{\text{lin}}\|_{\mathcal{H}_{\text{lin}}}^2 + \|g_{\text{per}}\|_{\mathcal{H}_{\text{per}}}^2 + \|g_{\text{Gaussian}}\|_{\mathcal{H}_{\text{Gaussian}}}^2,
\end{equation}
which follows by dropping the infimum in~\eqref{eq:rkhs_sum_norm} over any feasible decomposition.

The linear kernel norm $\|g_{\text{lin}}\|_{\mathcal{H}_{\text{lin}}}^2 = \|\beta\|^2$ depends only on the slope and not on the horizon, the periodic norm depends on the seasonal amplitude and smoothness, and the Gaussian norm is moderate because $g_{\text{Gaussian}}$ is a stationary residual the Gaussian kernel represents efficiently. By contrast, a Gaussian-only specification pays a norm that grows with the window's length for the same target (Section~\ref{app:spectral_rep}, Figure~\ref{fig:app_scree}). Substituting into the worst-case bound of Proposition~\ref{prop:worst_case} yields the combined bound~\eqref{eq:combined_bound_main}, in which the first factor measures each component's complexity in its natural RKHS and the second grows with extrapolation distance.

\subsubsection{Source Condition}\label{app:source_condition_formal}

The source condition \citep{caponnetto2007optimal} characterises kernel--target alignment more finely than membership by quantifying how fast the eigencoefficients of $g^*$ decay relative to the kernel's eigenvalues.

\begin{definition}[Source condition]\label{def:source}
Let $T_k$ be the kernel integral operator with Mercer eigensystem $\{(\mu_j, \phi_j)\}_{j \geq 1}$, and define $T_k^r f = \sum_{j \geq 1} \mu_j^r f_j \phi_j$ for $r \geq 0$, with $f_j = \langle f, \phi_j \rangle_{L^2(P_{\mathbf{X}})}$. The best RKHS approximation $g^*$ satisfies a \emph{source condition of order $r$} if there exists $v \in \mathcal{H}_k$ with
\begin{equation}\label{eq:source_condition}
    g^* = T_k^r v.
\end{equation}
\end{definition}
\noindent We take $r = 0$ to mean plain membership $g^* \in \mathcal{H}_k$, the assumption of Proposition~\ref{prop:worst_case}. The condition constrains the learning target $g^*$ and is separate from Term~II.  Because $v \in \mathcal{H}_k = T_k^{1/2}(L^2(P_{\mathbf{X}}))$, our order $r$ corresponds to $r + 1/2$ in the $L^2$ parameterisation common in the learning-theory literature, so the saturation point $r = 1/2$ here is the familiar value of one \citep{caponnetto2007optimal}. In Sobolev scales with $\mathcal{H}_k \simeq W_2^{s}$ ($s > d/2$) and $g^* \in W_2^{s_0}$, $r = (s_0 - s)/(2s)$, so $r = 0$ corresponds to $s_0 = s$ and $r = 1/2$ to $s_0 = 2s$. This calibration requires polynomial eigenvalue decay; the combined kernel's components have exponential or finite-rank spectra, but Figure~\ref{fig:app_source_draws} shows the effect of $r$ on the fitted kernel directly.

\begin{figure}[hbt!]
    \centering
    \includegraphics[width=1\textwidth]{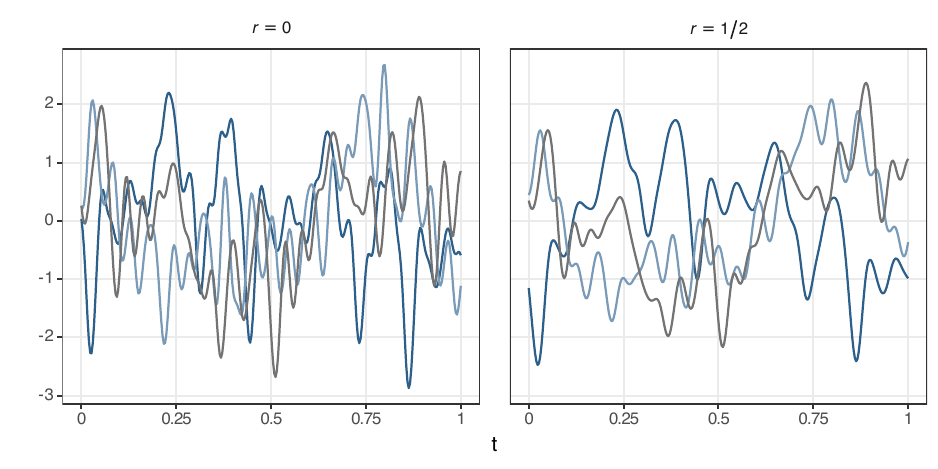}
    \caption{The source condition means alignment with the kernel's leading directions. Each panel shows three draws $g = T_k^{r} v$ with $v$ uniform on the unit sphere of $\mathcal{H}_k$, computed from the empirical eigensystem of the fitted combined kernel and rescaled to unit sample variance for display; the same three $v$ appear in both panels. The left panel takes $r = 0$, plain membership, and the right panel takes $r = 1/2$, the saturation order of Theorem~\ref{thm:l2_rate}.  Raising $r$ damps the small-eigenvalue directions, so the draws in panel~(b) carry less local movement.}
    \label{fig:app_source_draws}
\end{figure}
The Sobolev scale calibrates $r$ against a familiar quantity, but the combined kernel does not sit on that scale, so Figure~\ref{fig:app_source_draws} reads $r$ off the fitted eigensystem instead, as Figure~\ref{fig:app_scree} reads the cost of a trend. Its leading direction is the standardised trend. At $r = 0$ a target in $\mathcal{H}_k$ may still carry local movement that the leading directions do not describe, and pays for it in norm; at $r = 1/2$ that movement is damped and what remains is close to what those directions already describe. Alignment is therefore a property of the kernel and the design together, not of $g^*$ alone.

\begin{proposition}[Norm control under the source condition]\label{prop:source_norm}
If $g^* = T_k^r v$ for some $v \in \mathcal{H}_k$ and $r \geq 0$, then
\begin{equation}\label{eq:source_norm}
    \|g^*\|_{\mathcal{H}_k}^2 \leq \mu_1^{2r} \|v\|_{\mathcal{H}_k}^2.
\end{equation}
\end{proposition}
\begin{proof}
Since $g^* = T_k^r v$, its $L^2$ eigencoefficients are $(g^*)_j = \mu_j^r v_j$ with $v_j = \langle v, \phi_j \rangle_{L^2(P_{\mathbf{X}})}$. By~\eqref{eq:rkhs_norm_mercer},
\begin{equation*}
    \|g^*\|_{\mathcal{H}_k}^2 = \sum_{j \geq 1} \frac{(\mu_j^r v_j)^2}{\mu_j} = \sum_{j \geq 1} \mu_j^{2r-1} (v_j)^2 \leq \mu_1^{2r} \sum_{j \geq 1} \frac{(v_j)^2}{\mu_j} = \mu_1^{2r} \|v\|_{\mathcal{H}_k}^2,
\end{equation*}
where the inequality uses $\mu_j \leq \mu_1$.
\end{proof}
\noindent The ratio $\|g^*\|_{\mathcal{H}_k}^2 / \|v\|_{\mathcal{H}_k}^2 = \sum_j \mu_j^{2r}\,(v_j^2/\mu_j) \big/ \|v\|_{\mathcal{H}_k}^2$ is a weighted average of $\mu_j^{2r}$ over the spectral directions in which $v$ carries energy, with $\mu_1^{2r}$ as the worst case.  The source condition concentrates the target on directions the kernel represents cheaply, controlling the norm factor $\|g^*\|_{\mathcal{H}_k}$ in Proposition~\ref{prop:worst_case}; Proposition~\ref{prop:source_worstcase} below controls the geometric factor $\|k_{t^*}^\perp\|_{\mathcal{H}_k}$ separately. Under the standard normalisation $\mu_1 \le 1$ the factor $\mu_1^{2r}$ is a genuine reduction; for an unnormalised composite kernel it quantifies spectral alignment rather than unconditional norm reduction.

The source condition refines the worst-case learning bias of Proposition~\ref{prop:worst_case} directly rather than only through the norm factor, since over the source class of order $r$ the worst-case bias is governed by a smoothed residual.

\begin{proposition}[Alignment-refined worst-case learning bias] \label{prop:source_worstcase}
Suppose $g^*$ satisfies the source condition of order $r \ge 0$ (Definition~\ref{def:source}), $g^* = T_k^{r} v$ with $\|v\|_{\mathcal{H}_k} \le 1$. Then the deterministic learning bias of Proposition~\ref{prop:worst_case} satisfies
\begin{equation}\label{eq:source_worstcase}
    \Bigl| g^*(\mathbf{X}_{t^*}) - \textstyle\sum_t w_t\, g^*(\mathbf{X}_t) \Bigr|
    \;=\; \bigl| \langle g^*,\, k_{t^*}^\perp \rangle_{\mathcal{H}_k} \bigr|
    \;\le\; \bigl\| T_k^{r}\, k_{t^*}^\perp \bigr\|_{\mathcal{H}_k},
\end{equation}
and the right-hand side is the exact worst case as $g^*$ ranges over $\{T_k^{r} v : \|v\|_{\mathcal{H}_k}\le 1\}$. Writing $a_j = \langle k_{t^*}^\perp,\sqrt{\mu_j}\,\phi_j\rangle_{\mathcal{H}_k}$ and $\nu_j = \mu_j/\mu_1 \in (0,1]$,
\begin{equation}\label{eq:source_spectral}
    \bigl\| T_k^{r} k_{t^*}^\perp \bigr\|_{\mathcal{H}_k}^2
    \;=\; \mu_1^{2r}\sum_{j\ge 1} \nu_j^{\,2r}\, a_j^2
    \;\le\; \mu_1^{2r}\,\bigl\| k_{t^*}^\perp \bigr\|_{\mathcal{H}_k}^2 ,
\end{equation}
with equality iff $k_{t^*}^\perp$ is supported on the leading eigenspace. At $r=0$ the bound equals the constant $\|k_{t^*}^\perp\|_{\mathcal{H}_k}$ of Proposition~\ref{prop:worst_case}.
\end{proposition}

\begin{proof}
The first equality is the representer identity used for
Proposition~\ref{prop:worst_case}. In the orthonormal basis $\{\sqrt{\mu_j}\,\phi_j\}_{j\ge1}$ of $\mathcal{H}_k$ the operator $T_k^{r}$ is diagonal with entries $\mu_j^{r}$, hence self-adjoint on $\mathcal{H}_k$, so $\langle T_k^{r} v, k_{t^*}^\perp\rangle_{\mathcal{H}_k} = \langle v, T_k^{r} k_{t^*}^\perp\rangle_{\mathcal{H}_k}$. Maximising over $\|v\|_{\mathcal{H}_k}\le 1$ gives $\|T_k^{r} k_{t^*}^\perp\|_{\mathcal{H}_k}$, attained at $v \propto T_k^{r} k_{t^*}^\perp$. Expanding $k_{t^*}^\perp = \sum_j a_j \sqrt{\mu_j}\,\phi_j$ and applying $T_k^{r}$ coordinatewise yields~\eqref{eq:source_spectral}; the inequality uses $\nu_j\le 1$.
\end{proof}
\noindent Membership in the combined $\mathcal{H}_k$ holds the norm factor $\|g^*\|_{\mathcal{H}_k}$ finite and horizon-free (Section~\ref{app:kernel_sums}). A source condition of order $r > 0$ tightens the geometric factor: by~\eqref{eq:source_spectral}, $\|T_k^{r} k_{t^*}^\perp\|_{\mathcal{H}_k}^2 = \mu_1^{2r}\sum_j \nu_j^{2r}\,a_j^2$ with $\nu_j = \mu_j/\mu_1 \in (0,1]$, so $\sum_j \nu_j^{2r} a_j^2 \le \sum_j a_j^2 = \|k_{t^*}^\perp\|_{\mathcal{H}_k}^2$ is a strict reduction whenever $k_{t^*}^\perp$ has energy outside the leading eigenspace, regardless of whether $\mu_1 \le 1$.

\subsubsection{Convergence Rates}\label{app:convergence_rates}

\paragraph{$L^2$ rate for a single kernel.}

The source condition and the spectral decay together determine how fast the posterior mean converges on the pre-treatment support.  The rate below is the in-support benchmark; extrapolation points, where the fill distance is fixed by the study design, fall outside its scope.
\begin{theorem}[$L^2$ convergence of the GP posterior mean; adapted from \citealt{caponnetto2007optimal}]
\label{thm:l2_rate}
Let $k$ be a continuous positive definite kernel on a compact domain $\mathcal{X} \subset \mathbb{R}^d$ with Mercer eigenvalues $\mu_j \asymp j^{-\beta}$ for $\beta > 1$ and bounded diagonal $\sup_{\mathbf{x} \in \mathcal{X}} k(\mathbf{x}, \mathbf{x}) \le \kappa^2 < \infty$. Suppose $g^*$ satisfies the source condition of order $0 \le r \le 1/2$ (Definition~\ref{def:source}) with $\|v\|_{\mathcal{H}_k} \leq R$, and the noise $\varepsilon_t$ is sub-Gaussian. Let $\hat{g}$ be the GP posterior mean~\eqref{eq:representer_applied} with regularisation set to the rate-optimal scale $\lambda_n \asymp n^{-\beta/((2r+1)\beta + 1)}$. Then
\begin{equation}\label{eq:convergence_rate}
    \|\hat{g} - g^*\|_{L^2(P_{\mathbf{X}})}^2 = O_p\!\left(n^{-\frac{(2r+1)\beta}{(2r+1)\beta + 1}}\right), \quad n = t_0 - 1,
\end{equation}
where the $O_p$ constant depends on $R$, $\sigma^2$, and the kernel hyperparameters.
\end{theorem}
\noindent Theorem~\ref{thm:l2_rate} shows the in-support $L^2$ error is controlled by the pre-treatment length $n$ and the alignment order $r$, each entering the rate exponent. The theorem restates the upper bound of \citet[Theorem~1]{caponnetto2007optimal} in the parameterisation of Definition~\ref{def:source}. The restriction $r \le 1/2$ is the saturation point of ridge-type regularisation, beyond which the penalty bias dominates and the $L^2$ rate ceases to improve \citep{engl1996regularization}. The deterministic worst-case bound of Proposition~\ref{prop:source_worstcase} does not saturate: as a finite-sample bound rather than a rate, it reflects additional smoothness at every order $r \ge 0$.

The rate improves with both $n$ and $r$ and attains the minimax rate at the rate-optimal penalty \citep{caponnetto2007optimal}.  Two features of our estimation strategy place this optimum out of reach: the marginal-likelihood estimate $\hat\sigma^2$ gives an effective penalty $\lambda = \hat\sigma^2/n$ that decays as $O(n^{-1})$, faster than the rate-optimal scale, and the covariate-only bandwidth rule of \citet{kpop} does not adapt to the target's smoothness as the Gaussian-type regime~\eqref{eq:gauss_rate} requires, admitting only logarithmic contraction \citep{ghosal2017fundamentals}.  Both choices serve the design--analysis separation of Section~\ref{subsec:gp_estimation}, and the cost is small at the estimand that matters: the fill distance at a post-treatment point is fixed by the study design, so no in-support rate applies there, and the operative guarantee is the finite-sample bound of Proposition~\ref{prop:worst_case}.

When $k = k_{\text{Gaussian}}$, the eigenvalues fall outside the polynomial regime of Theorem~\ref{thm:l2_rate}, decaying exponentially in $j^{1/d}$, $\mu_j \leq C \exp(-c \, j^{1/d})$ \citep{belkin2018approximation}. The kernel then adapts to a finitely smooth target through its bandwidth rather than through the source order. If $g^* \in W_2^\beta(\mathcal{X})$ for $\beta > d/2$ and the bandwidth shrinks with $n$ at the appropriate rate, the posterior mean attains, up to an arbitrarily small polynomial loss,
\begin{equation}\label{eq:gauss_rate}
    \|\hat{g} - g^*\|_{L^2}^2 = O_p\!\left(n^{-\frac{2\beta}{2\beta + d}}\right),
\end{equation}
minimax-optimal over $W_2^\beta$ \citep{stone1980optimal}. In $d = 1$ with $\beta = 2$ this is $O_p(n^{-4/5})$.

\paragraph{Pointwise behaviour and the limits of extrapolation.} The $L^2$ rate~\eqref{eq:convergence_rate} governs average accuracy across the pre-treatment support; counterfactual prediction instead requires pointwise accuracy at $\mathbf{X}_{t^*}$, which may lie outside it. For interpolation within the support, \citet{kanagawa2025gaussian} show
\begin{equation}\label{eq:posterior_var_rate}
    \mathbb{V}[f(\mathbf{x}) \mid \mathbf{X}_{\text{pre}}, Y_{\text{pre}}] = O\bigl(h_n(\mathbf{x})^{2s - d}\bigr),
\end{equation}
where $s$ is the Sobolev-equivalent smoothness of the kernel and $h_n(\mathbf{x})$ is the local fill distance at $\mathbf{x}$, the radius of the largest gap in the training inputs within a fixed neighbourhood of $\mathbf{x}$. For equally spaced temporal data in $d = 1$, $h_n = O(n^{-1})$, so $\mathbb{V} = O(n^{-(2s-1)})$, contracting faster than any polynomial rate for the Gaussian kernel; with Proposition~\ref{prop:worst_case} the deterministic error at any pre-treatment point is $O(n^{-(s-1/2)})$. For extrapolation, $h_n(\mathbf{X}_{t^*})$ does not shrink with $n$; it is the gap between the pre-treatment frontier and $t^*$, fixed by the study design, so~\eqref{eq:posterior_var_rate} does not apply, and the non-stationary linear component instead makes $\mathbb{V}[f(\mathbf{X}_{t^*}) \mid \cdot]$ grow with temporal distance. For the combined kernel, $T_k$ does not share eigenfunctions with its components and none has a polynomially decaying spectrum (the Gaussian and MacKay periodic kernels are $C^\infty$, the linear kernel finite-rank), so the relevant in-support regime is the Gaussian-type one of~\eqref{eq:gauss_rate} rather than the polynomial-spectrum rate of Theorem~\ref{thm:l2_rate}.

\begin{figure}[htb]
    \centering
    \includegraphics[width=\textwidth]{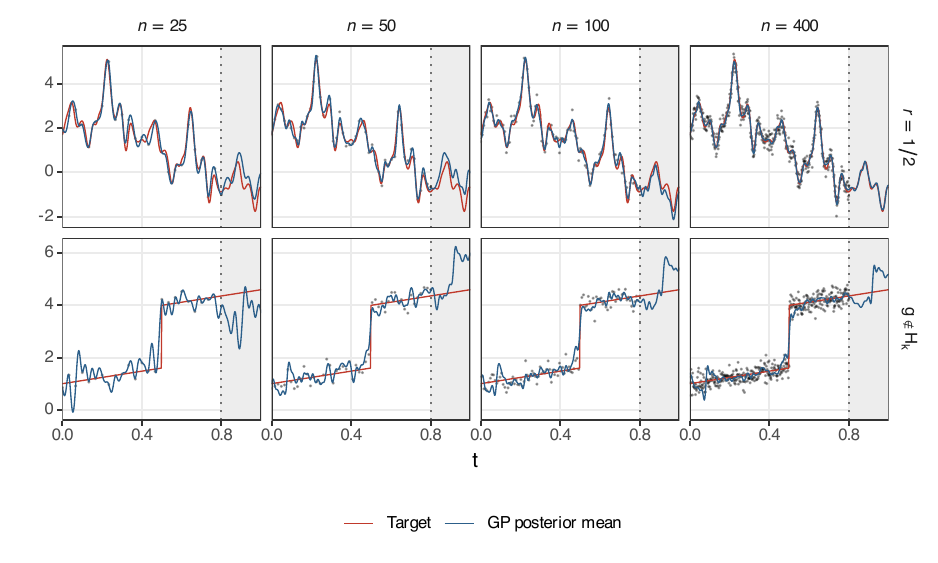}\\[3pt]
    \includegraphics[width=0.6\textwidth]{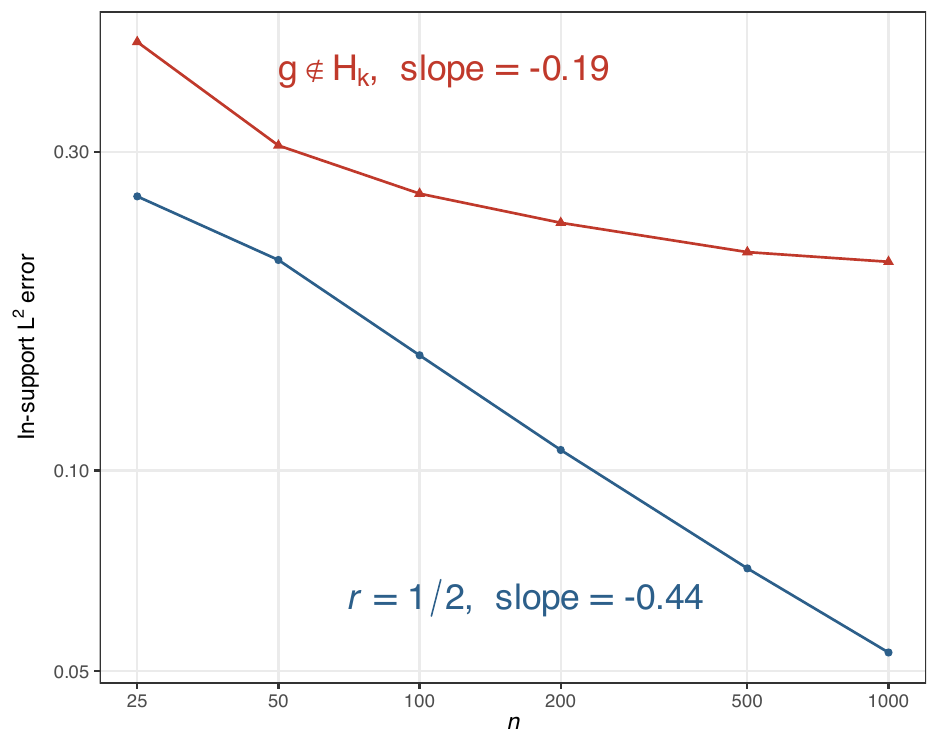}
    \caption{The source condition means faster rates.  In the upper grid, rows fix the target and columns fill in the pre-treatment window, $n = 25, 50, 100, 400$ on $[0, 0.8]$, with the shaded region marking extrapolation; tuning is pipeline-consistent, the bandwidth selected by variance maximisation from the design and $\sigma^2$ fixed at the true noise variance, so the penalty $\sigma^2/n$ falls as the window fills.  In the upper row the target satisfies the source condition at the saturation order $r = 1/2$, drawn from the empirical eigensystem as in Figure~\ref{fig:app_source_draws}; in the lower row it carries a level shift no finite budget contains, since every member of the combined RKHS is continuous.  The posterior mean closes on the target more quickly in the upper row.  The bottom panel reports the in-support $L^2$ error over thirty replications on log-log axes with fitted slopes, the aligned slope near the saturation rate of minus one half and the level-shift slope less than half as steep.  In the shaded region neither row contracts with $n$.}
    \label{fig:app_convergence}
\end{figure}
The posterior mean converges when the target lies in $\mathcal{H}_k$, and faster when it is better aligned; Figure~\ref{fig:app_convergence} shows both, and shows what is left when neither holds. In support, the aligned target is learned quickly and the (discontinuous) level-shift target slowly, the fitted slopes near $-1/2$ and less than half as steep. The slow one never arrives. It sits outside $\mathcal{H}_k$, so its remaining progress is a closer approximation bought at a larger budget, which the shrinking penalty $\sigma^2/n$ permits; hold the budget at $B$ and the gap stops closing and reappears as Term~II. Outside the support neither row improves as the window fills in, and what the band reports there is the extrapolation uncertainty of the model class itself, whose closed form Section~\ref{app:extrap_uncertainty} derives.

\subsection{Extrapolation Uncertainty of the Model Class}
\label{app:extrap_uncertainty}
This section proves the closed form of the extrapolation uncertainty $\omega_B$ of Definition~\ref{def:extrap_uncertainty} and places the model class on the scale of \citet{shen2025engression}.

\begin{corollary}[Closed form of the extrapolation uncertainty]
\label{cor:extrap_uncertainty}
Under the conditions of Corollary~\ref{cor:noiseless}, with $k_{t^*}^\perp$ the
orthogonal projection residual,
\begin{equation}\label{eq:omega_closed_form}
    \omega_B(\mathbf{X}_{t^*}) = 2B\,\|k_{t^*}^\perp\|_{\mathcal{H}_k} = 2B\sqrt{\mathbb{V}_{t^*}}.
\end{equation}
\end{corollary}

\begin{proof}
Write $u = f - f'$ for an admissible pair. Then $\|u\|_{\mathcal{H}_k} \le 2B$ by the triangle inequality, and $u$ vanishes at every pre-treatment input, so $u \perp \mathcal{S}_{\mathrm{pre}}$ by the reproducing property. Conversely, every $u \perp \mathcal{S}_{\mathrm{pre}}$ with $\|u\|_{\mathcal{H}_k} \le 2B$ arises from the admissible pair $(u/2, -u/2)$. Hence
\begin{equation*}
    \omega_B(\mathbf{X}_{t^*})
    = \sup\bigl\{ |\langle u, k_{t^*}^\perp\rangle_{\mathcal{H}_k}| : \|u\|_{\mathcal{H}_k} \le 2B,\; u \perp \mathcal{S}_{\mathrm{pre}} \bigr\}
    = 2B\,\|k_{t^*}^\perp\|_{\mathcal{H}_k},
\end{equation*}
where the first equality uses $u(\mathbf{X}_{t^*}) = \langle u, k(\cdot,\mathbf{X}_{t^*})\rangle_{\mathcal{H}_k} = \langle u, k_{t^*}^\perp\rangle_{\mathcal{H}_k}$ for $u \perp \mathcal{S}_{\mathrm{pre}}$, and the second is Cauchy--Schwarz with equality at $u \propto k_{t^*}^\perp$, the pair $B g_\pm$ of Proposition~\ref{prop:minimax}.
\end{proof}

\noindent The corollary ties the objects of this appendix together. The extrapolation uncertainty is twice the minimax risk of Proposition~\ref{prop:minimax}. Conditioning on an observed series shrinks it: by the same decomposition, if $\rho$ is the norm of the minimum-norm interpolant of $Y_{\mathrm{pre}}$, the data-conditional set of Section~\ref{app:sharpening} has diameter $2\sqrt{B^2 - \rho^2}\, \|k_{t^*}^\perp\|_{\mathcal{H}_k}$ at $\mathbf{X}_{t^*}$, reaching $\omega_B(\mathbf{X}_{t^*})$ only at $\rho = 0$, and Figure~\ref{fig:app_sim_attainment} draws the attaining pairs. The coupling argument of the paragraph \textit{Noise does not restore learnability} extends the necessity direction to noisy training, and is the fixed-design instance of the necessity clause of \citet[Theorem~1(ii)]{shen2025engression} for post-additive noise models. Since the target is the CEF, the operative notion is their mean extrapolability (their Definition~2), which under post-additive noise with a fixed error law coincides with the functional notion used here.

On their scale the class sits between the extremes. Linear functions carry zero uncertainty (their Example~1), $L$-Lipschitz classes grow as $2L\delta$ in the distance $\delta$ from the support (Example~2), and unrestricted monotone classes are already at infinity (Example~3). $\mathcal{F}_B$ has the finite closed form~\eqref{eq:omega_closed_form}: under a stationary kernel $\|k_{t^*}^\perp\|_{\mathcal{H}_k}$ never exceeds $\sqrt{k(\mathbf{X}_{t^*},\mathbf{X}_{t^*})}$, so the profile is bounded and flattens below the Lipschitz rate, and the linear component lifts that ceiling, the horizon profile inheriting the component structure of Section~\ref{app:kernel_sums}. Their extrapolability gain, the difference between the uncertainties of two specifications, transfers by evaluating each at the budget the same trajectory requires, both terms computable from~\eqref{eq:omega_closed_form}. On the \textit{Heller} series of Section~\ref{sec:application} the trend costs $B = 0.96$ under the combined kernel against $B = 3.52$ under a Gaussian-only kernel, a budget gap that the norm growth of Section~\ref{app:spectral_rep} widens with every added pre-treatment year.

\clearpage

\section{Simulation Details}\label{app:sim_details}

This section records the data-generating processes for the calibrated simulation of Section~\ref{sec:simulation} and connects each regime to the error decomposition of Section~\ref{sec:theory}.

\subsection{Common Structure}

We generate monthly outcomes $y_t = g(t) + \varepsilon_t$ for $t = 1, \ldots, n_{\mathrm{pre}} + 12$, with $\varepsilon_t \sim \mathcal{N}(0,\sigma^2)$, treatment beginning at $t_0 = n_{\mathrm{pre}} + 1$, a fixed twelve-month post-treatment horizon, and a constant effect $\tau$ added over that horizon. We write $z_t$ for the time index standardised on the pre-treatment window, and the annual period is twelve. All three regimes share monthly seasonality and a short pre-treatment data to mirror the application of Section~\ref{sec:application}, and all include a small random Gaussian bump
\[
  r(t) = a\exp\!\bigl\{-(t-\mu)^2/(2w^2)\bigr\}, \qquad a \sim U(-0.3,0.3),\; \mu \sim U(8, n_{\mathrm{pre}}-4),\; w \sim U(4,8),
\]
that displaces each trend from a clean closed form. The bump introduces transient local deviations of the kind an applied series would carry, so the estimator cannot exploit exact knowledge of the parametric shape. Throughout, $\beta \sim U(0.8,1.3)$, $A \sim U(0.7,1)$, $\phi \sim U(0,2\pi)$, $\sigma^2 = 0.05$ and $\tau = 0.6$. The noise level gives a signal-to-noise ratio typical of monthly administrative data after unit-variance standardisation, and $\tau = 0.6$ is moderate relative to the pre-treatment standard deviation, large enough to detect but not so large that every method finds it.

\subsection{Three Regimes}

The regimes place the target $g$ at increasing distance from the estimator, from a kernel-based one to shapes it can only approximate, so the simulation traces how coverage degrades as the four error terms of Section~\ref{sec:theory} come into play.

\paragraph{Smooth.} A kernel-based random function:
$$g(t) = 1.5\,\bar f(t), \quad f(t) = \sum_{j=1}^{10} \omega_j\, k(z_t, z_j), \quad \omega_j \sim \mathcal{N}(0,1),\ \ z_j \sim U(z_1, z_{n_{\mathrm{pre}}+12}),$$ with $\bar f$ the pre-treatment standardisation of $f$ and $k$ the Gaussian $+$ periodic $+$ linear kernel of Section~\ref{subsec:gp_estimation} at a fixed length-scale. Because $g$ is a finite linear combination of kernel evaluations, it lies in $\mathcal{H}_k$ with a finite norm, so the approximation error (Term~II) vanishes. The target also satisfies the source condition (Appendix~\ref{app:source_condition_formal}) at a positive order, since the random weights $\omega_j$ do not concentrate on the high-frequency eigenfunctions of the kernel integral operator, and draws from a GP prior satisfy a source condition of order $r = 1/2$ in expectation. The regime therefore isolates Term~III and Term~IV: the only errors are the learning gap the finite pre-treatment data leaves and the propagated noise. The estimator still selects its own length-scale hyperparameter and includes the calendar month as a covariate, so the target is not exactly in its span, and coverage depends on whether the worst-case bound of Proposition~\ref{prop:worst_case} tracks the realised error at realistic sample sizes.

\paragraph{Nonlinear trend.} A concave, decelerating trend with milder seasonality:
\[
  g(t) = \bar m(t) + 0.7\,A\,\sin(2\pi t/12 + \phi) + r(t), \qquad m(t) = 3\beta\bigl\{1 - \exp\!\bigl(-(t-2)_{+}/0.22\,n_{\mathrm{pre}}\bigr)\bigr\},
\]
with $\bar m$ the pre-treatment standardisation of $m$, mimicking a bounded growth path whose increments shrink as the series matures. Many social and economic series exhibit this pattern, for instance adoption curves for new technologies or the approach to a regulatory ceiling. The exponential saturation is smooth and well approximated by the combined kernel on the observed window, but it does not belong to $\mathcal{H}_{\mathrm{lin}}$, so the Aronszajn decomposition~\eqref{eq:rkhs_sum_norm} charges the curvature to the Gaussian component at a cost that grows with the window length. Term~II is therefore nonzero at any fixed budget, and the DGP tests whether the interval remains calibrated when the target sits near but not inside the unit ball. The reduced seasonal amplitude ($0.7\,A$ versus $1.3\,\beta$ in the third regime) keeps the periodic component from dominating, so the extrapolation challenge falls mainly on the trend.

\paragraph{Saturating trend.} A strong annual cycle above a levelling trend:
\[
  g(t) = 1.3\,\beta\,\sin(2\pi t/12 + \phi) + 0.9\,\beta\,\tanh(0.9\, z_t) + r(t),
\]
where seasonality dominates and the trend flattens at both ends. This is the most demanding regime for extrapolation. The $\tanh$ sigmoid is again outside $\mathcal{H}_{\mathrm{lin}}$, and the large seasonal amplitude forces the periodic component to carry substantial norm, so both legs of the Aronszajn bound~\eqref{eq:combined_bound_main} are loaded. The regime mirrors series in which a strong cyclical pattern sits above a level that has largely stabilised, such as seasonal background-check volumes in states with mature handgun markets. With trend and seasonality both large, the budget $B$ required to contain $g$ in the RKHS ball is the highest of the three regimes, and the test is whether the GP interval widens enough to remain calibrated even when the counterfactual is expensive in RKHS norm.

\subsection{Design Choices}

The three DGPs are designed to be ordered by the distance from the estimator's function class, and Figure~\ref{fig:app_sim_samplesize} suggests that the theoretical ordering shows up empirically. In the \textit{Smooth} regime, where Terms~I and~II are absent, coverage sits near or above nominal at every pre-treatment length, the outcome the worst-case bound of Proposition~\ref{prop:worst_case} predicts. In the \textit{Nonlinear trend} and \textit{Saturating trend} regimes, Term~II is positive and the budget is larger, so the GP runs conservative, widening its interval beyond what the realised error requires. The conservatism is a feature rather than a failure: the interval reports the worst case over the RKHS ball rather than the realised target, and a target that sits inside the ball but near its boundary should produce coverage above, not at, the nominal rate.

The estimator in each regime follows the same pipeline as the application, with the length-scale selected by the variance-maximisation rule, and $\sigma^2$ then estimated by marginal likelihood. No regime gives the estimator access to the true hyperparameters, so the simulation tests the full estimation pipeline rather than an oracle version.

The GP uses the Gaussian$+$Periodic$+$Linear kernel on the monthly index with the calendar month entered as a categorical covariate and $95\%$ intervals. The three comparators are (i)~a segmented regression carrying a linear pre-trend and an annual harmonic, (ii)~\texttt{C-ARIMA}, a counterfactual ARIMA model \citep{menchetti2023combining}, and (iii)~\texttt{CausalImpact}, a Bayesian structural time-series model \citep{brodersen2015inferring}. The harmonic rather than eleven month indicators keeps the segmented fit estimable at the shortest record, and \texttt{C-ARIMA} and \texttt{CausalImpact} receive the annual period the GP is given. 

\section{Additional Empirical Examples}\label{app:additional_examples}

Beyond the \textit{Heller} application of the main text, we apply the method to two further settings. The All-Women Police Stations example adds a daily series with strong weekly seasonality and a spillover-inclusive estimand, and the Ukraine example pushes the design to a severe case, a military invasion whose aftermath is hard to separate from the initial shock, where the placebo check flags a pre-trend that the analysis then carries into bounded conclusions.

\subsection{All-Women Police Stations and Gender-based Crime}\label{app:awps}

We reanalyse \cite{jassal2020gender}, who examines whether AWPS improve women's access to justice in India. On August 28, 2015, 20 AWPS opened simultaneously across Haryana state, a universal treatment affecting all 285 police stations. Using daily gendered crime registration rates per 100{,}000 population from January 2015 to August 2017 (8 months pre-treatment, 2 years post-treatment; Figure~\ref{fig:rawdata_awps}), the original segmented regression analysis found a 27--49\% reduction at standard stations without a system-wide increase, consistent with bureaucratic deflection as standard stations redirect women to AWPS.

\begin{figure}[hbt!]
    \centering
    \includegraphics[width=0.9\textwidth]{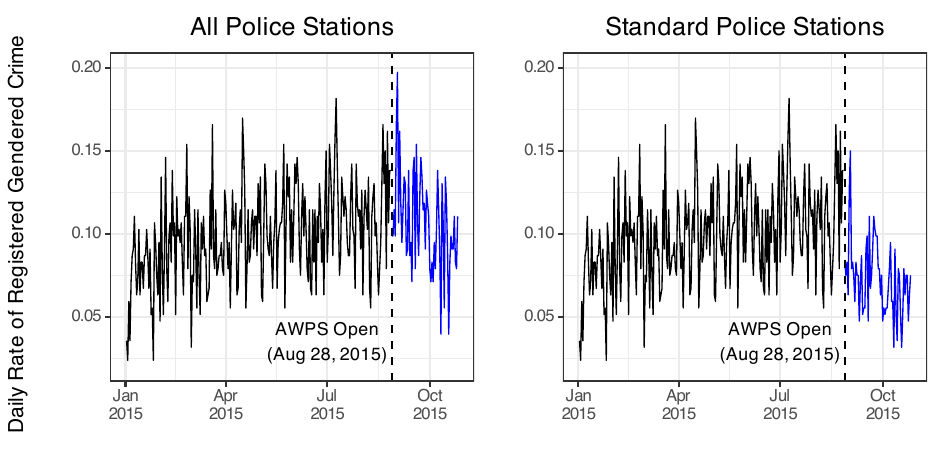}
    \caption{Raw Outcome Data for the AWPS Application. \textit{Note: } Daily gendered crime registration rates per 100{,}000 population. The vertical dashed line marks the simultaneous opening of the All-Women Police Stations (August 28, 2015).}
    \label{fig:rawdata_awps}
\end{figure}

We employ the Gaussian$+$Periodic$+$Linear kernel with 7-day periodicity and day-of-week covariates. Following the original study, we examine effects at both standard stations and system-wide, restricting analysis to two months post-treatment to limit the scope for concurrent events that would threaten Assumption~\ref{assum:sufficiency}. The abrupt opening date supports no anticipation for stations; Appendix~\ref{app:assumptions} takes up the residual concern of advance media coverage.

\begin{figure}[hbt!]
    \centering
    \includegraphics[width=0.9\textwidth]{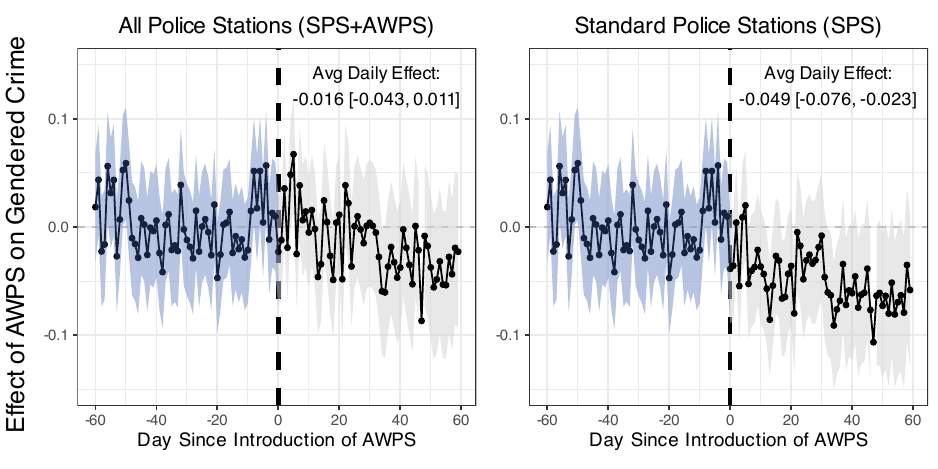}
    \caption{Impact of All-Women Police Stations on Gendered Crime Registration per 100{,}000 \citep{jassal2020gender}. \textit{Note:} Left panel: daily treatment effects on gendered crime registrations across all police stations. Right panel: effects at Standard Police Stations (SPS) only. Pre-AWPS placebo estimates shown in blue. The introduction of AWPS did not significantly affect system-wide registration rates (left) but led to a substantial reduction at SPS (right).}
    \label{fig:jassal}
\end{figure}

Figure~\ref{fig:jassal} shows contrasting patterns. System-wide registration rates show no significant change (left panel), while standard stations exhibit substantial reductions (right panel). The average two-month effect at standard stations is $-0.049$ registrations per 100{,}000 population, with a $95\%$ interval of $[-0.076,\, -0.023]$, a decrease of at least $23\%$ from baseline. The null system-wide total is more consistent with redirection to AWPS than with net deterrence, which would lower the system-wide count. The reduction at standard stations intensifies over time, a dynamic that several gradual processes could produce, including increasing redirection as the new stations become established. Segmented regression, summarising the post-treatment period as a single level shift, cannot reveal this pattern.

Under the matched one-day-ahead protocol the GP covers close to the nominal rate across the 60 pre-treatment points, while the segmented alternatives fall far short (Appendix~\ref{app:empirical_examples}).

\subsection{The Russian Invasion of Ukraine}\label{app:ukraine}

We reanalyse \citet{damann2024women}, who study how the Russian invasion of Ukraine on 24 February 2022 affected gendered patterns in politicians' public engagement. The original study fits linear ITS models with Newey-West standard errors to the Facebook posting of 469 Ukrainian politicians (81 women, 388 men) from November 2021 to June 2022 (Figure~\ref{fig:ukraine_raw}). We restrict to a one-month post-treatment window to limit the scope for components of $\mathbf{U}$ taking values they never took before $t_0$, and fit a separate GP per politician on pre-invasion data.

\begin{figure}[ht!]
    \centering
    \includegraphics[width=0.9\textwidth]{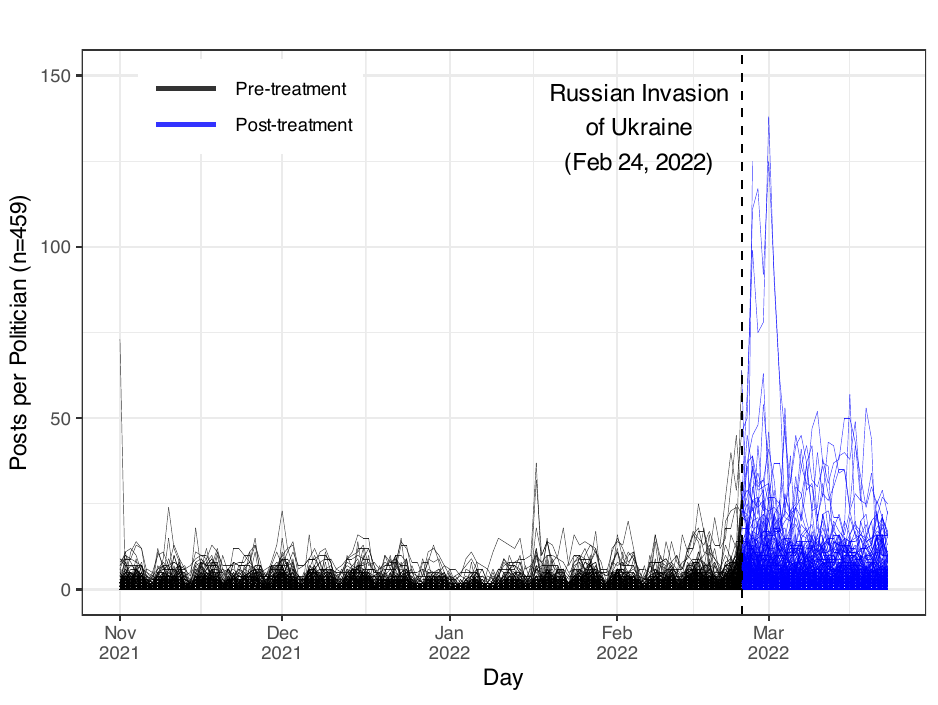}
    \caption{Daily Facebook Posts by Ukrainian Politicians. \textit{Note: } Daily Facebook posts of politicians from November 2021 to March 2022. Individual politician trajectories shown, pre-treatment (black) and post-treatment (blue). The vertical dashed line marks the Russian invasion on February 24, 2022. One politician is dropped for improved visualisation.}
    \label{fig:ukraine_raw}
\end{figure}

\begin{figure}[ht!]
    \centering
    \includegraphics[width=0.9\textwidth]{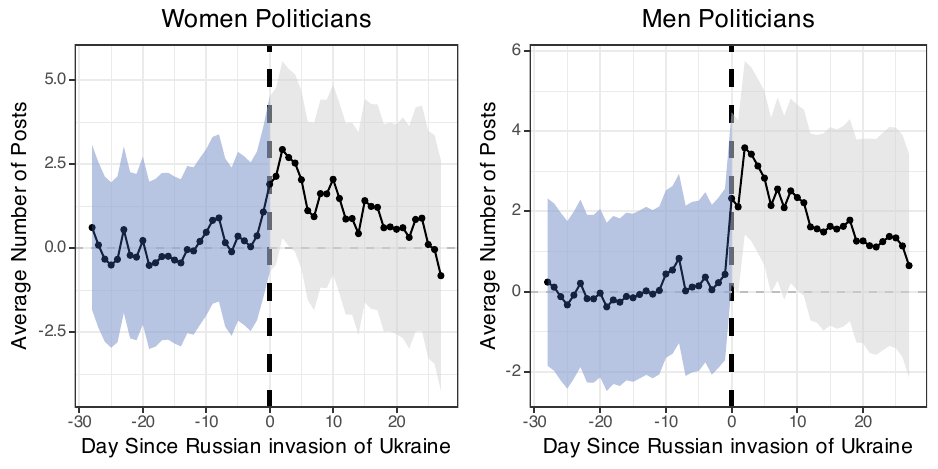}
    \caption{Daily Average Treatment Effects by Gender. \textit{Note: } Estimated daily treatment effects with 95\% prediction intervals. Left panel shows effects for women politicians; right panel for men. A one-month pre-treatment period serves as the placebo (blue). While both groups show increased engagement, men exhibit larger effects.}
    \label{fig:ukraine_effect}
\end{figure}

Our estimates align with the original findings: both groups increased posting after the invasion, men more so (Figure~\ref{fig:ukraine_effect}). Here the placebo check does the falsification work of Section~\ref{subsec:placebo_check}, surfacing a slight upward pre-trend for women that signals a possible violation of Mean Sufficiency. The pre-trend biases the women's counterfactual upward and so understates their own effect, which is therefore a lower bound. The gender gap is a separate matter: understating the women's effect widens the estimated men-minus-women difference, so the ``men more so'' comparison is best read as an upper bound rather than a conservative one.

\subsection{Identifying Assumption Plausibility in the Running Examples}\label{app:assumptions}
Section~\ref{sec:identification} of the main text states the identifying assumptions in general form. Here we read each one across our three applications, asking what each assumption requires, what would break it, and why we judge it credible (or not). The three settings stress the assumptions differently, and reading them together shows where the design is most and least comfortable. This is the substantive work the placebo checks (Section~\ref{subsec:placebo_check}) cannot do on their own.

\paragraph{Consistency.} This requires that each unit's treatment is a single well-defined intervention switching on at a known date. All three meet it, since the \textit{Heller} ruling, the AWPS openings, and the invasion each take force on one administratively or historically fixed date. The one subtlety is the AWPS standard stations, which are not treated directly but through the altered institutional environment around them, so their estimand is the total effect on a standard station under the new regime, inclusive of spillovers such as bureaucratic deflection. Consistency would break if ``treatment'' bundled several coincident changes, leaving the effect unattributable to a single intervention.

\paragraph{No Anticipation.} This requires that pre-treatment outcomes are undisturbed by the coming treatment, so $Y_{it} = Y_{it}(0)$ for $t < t_0$ and $g_i$ is learned from genuinely untreated data. Because the counterfactual is the continuation of the pre-treatment regime (Section~\ref{sec:identification}), expectations held at a stable level do not violate the assumption, and what breaks it is behaviour that shifts as the event approaches. \citet{bertoli2026analyzing} show that such anticipatory drift can threaten even seemingly unexpected events, when background public awareness or structural signals alter behaviour during the lead-up period, and that its severity depends on the counterfactual scenario the researcher chooses. At post-treatment periods the covariate equality $\mathbf{X}_{it^*}(1) = \mathbf{X}_{it^*}(0)$ is additionally required; it holds by construction when $\mathbf{X}_{it}$ consists of time and calendar indicators, and substantive covariates require the researcher to justify that they are not downstream of treatment.

For \textit{Heller}, the live case traces back to the March 2008 oral arguments, three months ahead of the decision, which could have moved purchasing in anticipation. We treat the ruling date itself as unanticipated because the 5--4 outcome was not a foregone conclusion, and the strictness of D.C.'s handgun ban structurally prevented preemptive legal purchasing, so the pre-treatment series remains uncontaminated by anticipatory behaviour. An analyst worried about partial anticipation elsewhere can set $t_0$ at the oral-argument date as a robustness check. For AWPS the abrupt opening gives station personnel no reason to redirect women beforehand, and the residual concern, advance media coverage, is mitigated because the eight-month pre-treatment window largely predates public awareness. For the invasion the preceding military build-up could have moved posting in anticipation, which is pre-treatment drift that the placebo checks probe directly and that we return to under Mean Sufficiency.

\paragraph{Mean Sufficiency (Assumption~\ref{assum:sufficiency}).} This is the assumption the design leans on most heavily. It requires that, conditional on observed covariates, no other time-varying factor shifts the mean untreated outcome, so $g_i$ carries across $t_0$. It breaks when a component of $\mathbf{U}$ takes values after $t_0$ that it never took before (Term~I, Section~\ref{sec:theory}). Measurement error that is mean-zero given $(\mathbf{X}_{it}, \mathbf{U}_{it})$ joins $\varepsilon_{it}$ and widens the intervals without biasing them; a stable recording convention joins the regime $g_i$ learns and a constant miscount differences out of $\tau_{it^*}$. A change in the measurement process near $t_0$ is a component of $\mathbf{U}_{it}$ taking values it never took before, the same failure. For \textit{Heller} the leading threat is political anxiety following Obama's November 2008 victory, removed by truncating the window at four months. For AWPS the parallels are concurrent shifts in recording practice or seasonal dynamics beyond the periodic kernel, whose scope the two-month window keeps narrow. The invasion is the hardest case; we restrict to one month and carry the detected women's pre-trend into bounded conclusions (Section~\ref{app:ukraine}). If treatment changed how observables predict the outcome, that is a treatment effect operating through behaviour, not a violation of an assumption restricting untreated potential outcomes.

\section{Additional Results for the Empirical Applications} 
\label{app:empirical_examples}

This section reports the placebo comparison against segmented regression for the \textit{Heller} and AWPS applications, and adds three further diagnostics for \textit{Heller}: cumulative effects, residual analysis, and a placebo outcome on long gun purchases.

\begin{figure}[hbt!]
    \centering
    \includegraphics[width=0.9\textwidth]{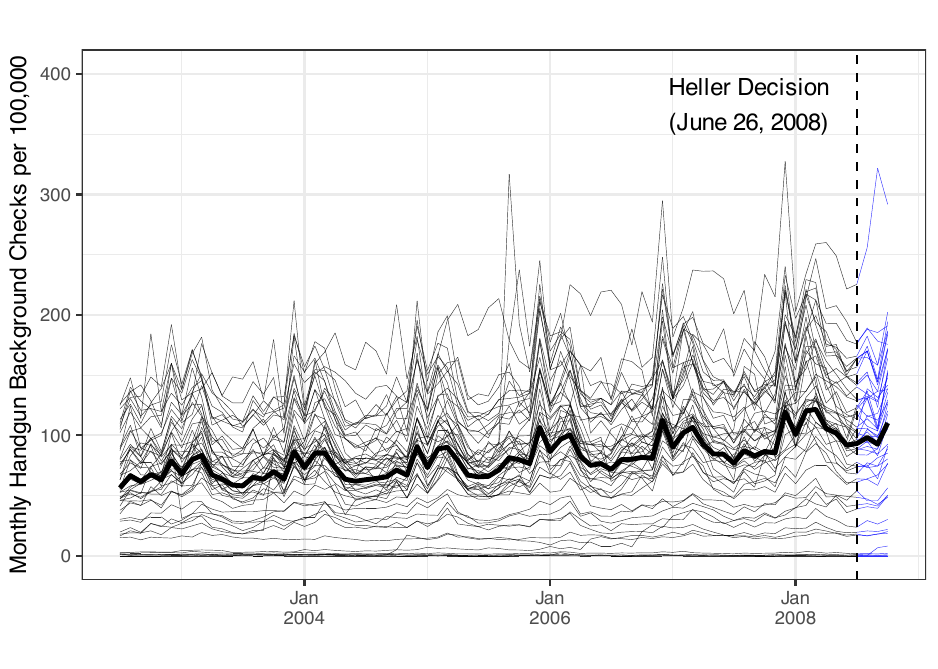}
    \caption{Raw Outcome Data for the \textit{Heller} Application. \textit{Note: } Monthly handgun background checks per 100{,}000 population from the FBI's National Instant Criminal Background Check System (NICS), for each state and the nationwide average. The vertical dashed line marks the \textit{Heller} decision (June 26, 2008).}
    \label{fig:rawdata_heller}
\end{figure}

\subsection{Out-of-Sample Predictive Performance with Placebo Checks: Comparison to Segmented Regression}
Segmented regression is the workhorse ITS model in applied works \citep{bernal2017interrupted}. To compare it with the GP on equal terms, we hold both to the one-step-ahead placebo of Section~\ref{subsec:placebo_check} on the aggregated national mean series, fitting on periods $1$ through $t$ and recording whether the $95\%$ interval covers $Y_{t+1}$. The segmented specification is
\begin{equation*}
    Y_{t} = \beta_0 + \beta_1 I\{t \geq t_0^m\} + \beta_2 (t - t_0^m) + \beta_3 I\{t \geq t_0^m\}(t - t_0^m) + f(t) + \varepsilon_{t},
\end{equation*}
where $t_0^m$ is a placebo treatment time, $I$ is the indicator taking the value one when $t \geq t_0^m$, and $f(t)$ carries polynomial terms up to cubic order together with month dummies. The coefficient $\beta_1$ is the placebo treatment effect, and we use HC1 robust standard errors.

\begin{figure}[hbt!]
    \centering
    \includegraphics[width=0.9\textwidth]{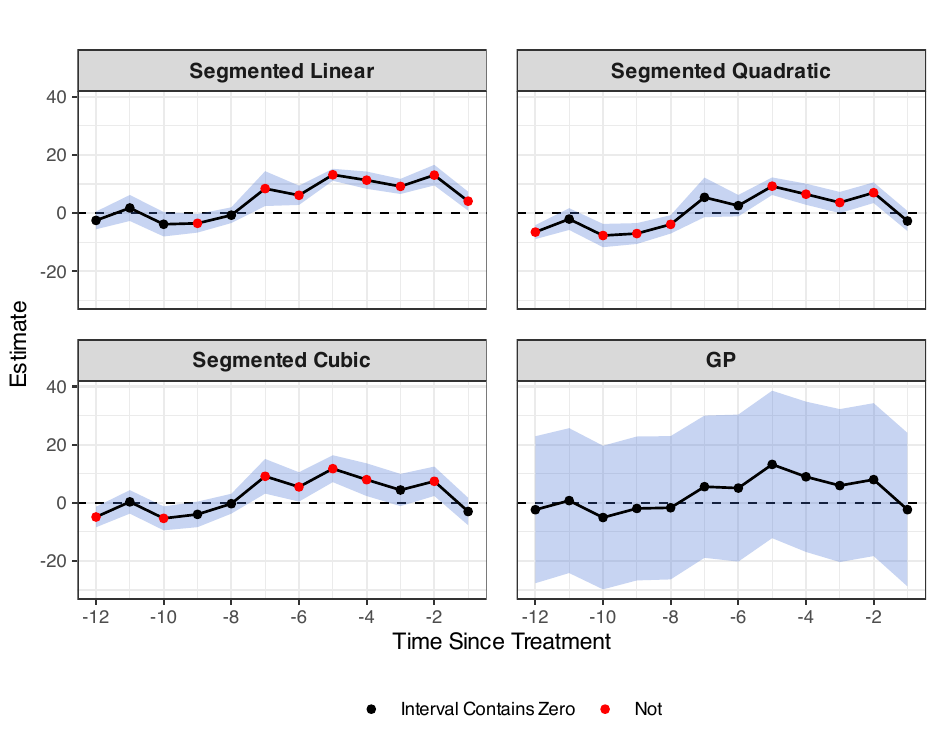}
    \caption{Placebo Check Performance: GP and Segmented Regression for \textit{Heller} example. \textit{Note: } The one-month-ahead placebo check results for 1 year, comparing GP against segmented regression with linear, quadratic, and cubic time trends. Points indicate placebo treatment effects with 95\% prediction intervals, coloured by whether the interval contains zero (black) or not (red).}
    \label{fig:placebo_comparison_heller}
\end{figure}

Figures~\ref{fig:placebo_comparison_heller} and~\ref{fig:placebo_comparison_jassal} report the comparison. For \textit{Heller} the GP intervals cover at or above the nominal $95\%$, while the segmented regressions cover $33$ to $42\%$ of placebo periods, rising slightly as the trend order increases; equivalently, they exclude zero at $58$ to $67\%$ of those periods, the figure quoted in Section~\ref{sec:application}. The AWPS application of Section~\ref{app:awps} shows the same pattern on a daily series, with the GP covering $90\%$ of placebo points against $35$ to $38\%$ for the segmented regressions.

\begin{figure}[hbt!]
    \centering
    \includegraphics[width=0.9\textwidth]{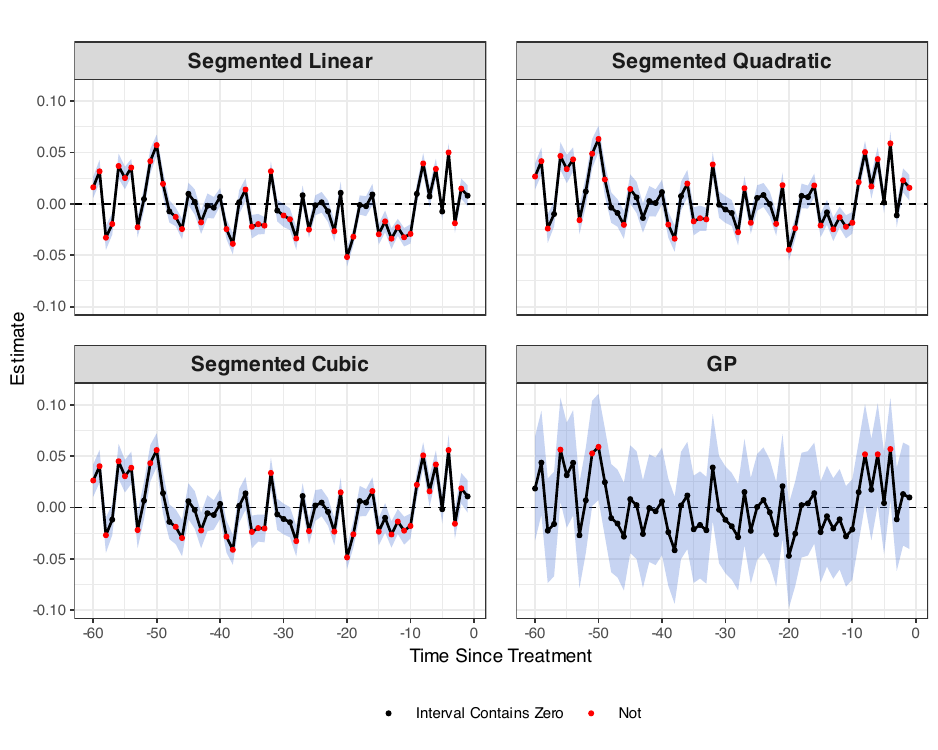}
    \caption{Placebo Check Performance: GP and Segmented Regression for \cite{jassal2020gender} example. \textit{Note: } The one-day-ahead placebo check results for 3 months, comparing GP against segmented regression with linear, quadratic, and cubic time trends. Points indicate placebo treatment effects with 95\% prediction intervals, coloured by whether the interval contains zero (black) or not (red).}
    \label{fig:placebo_comparison_jassal}
\end{figure}
The gap follows from the trend form rather than from any single tuning choice. A polynomial trend fixes the shape of the counterfactual in advance, so its interval reflects only the uncertainty in the fitted coefficients and not the uncertainty about whether that shape extrapolates. The raw series move in ways linear, quadratic, and cubic trends do not track (Figure~\ref{fig:rawdata_heller}), so out-of-sample the realised value often falls outside an interval built on the wrong shape, and a researcher reading these intervals would mistake ordinary fluctuation for a treatment effect. The GP's Gaussian, periodic, and linear kernel spans a wider set of admissible shapes and widens its interval as the horizon lengthens (Section~\ref{sec:theory}), which is what holds its placebo coverage near nominal.

\clearpage

\subsection{Cumulative National Average Effects}
The per-period national average effects of Section~\ref{sec:application} can be accumulated to summarise how the response builds over the post-treatment window. Figure~\ref{fig:heller_cumulative} plots the cumulative national average effect of the \textit{Heller} ruling on handgun background checks, summing the per-period effects and their uncertainty across months. The band widens as the horizon lengthens, since each additional month extrapolates the counterfactual one step further beyond the pre-treatment series, so the growth reflects accumulating extrapolation uncertainty rather than instability in the estimate.

\begin{figure}[hbt!]
    \centering
    \includegraphics[width=0.8\textwidth]{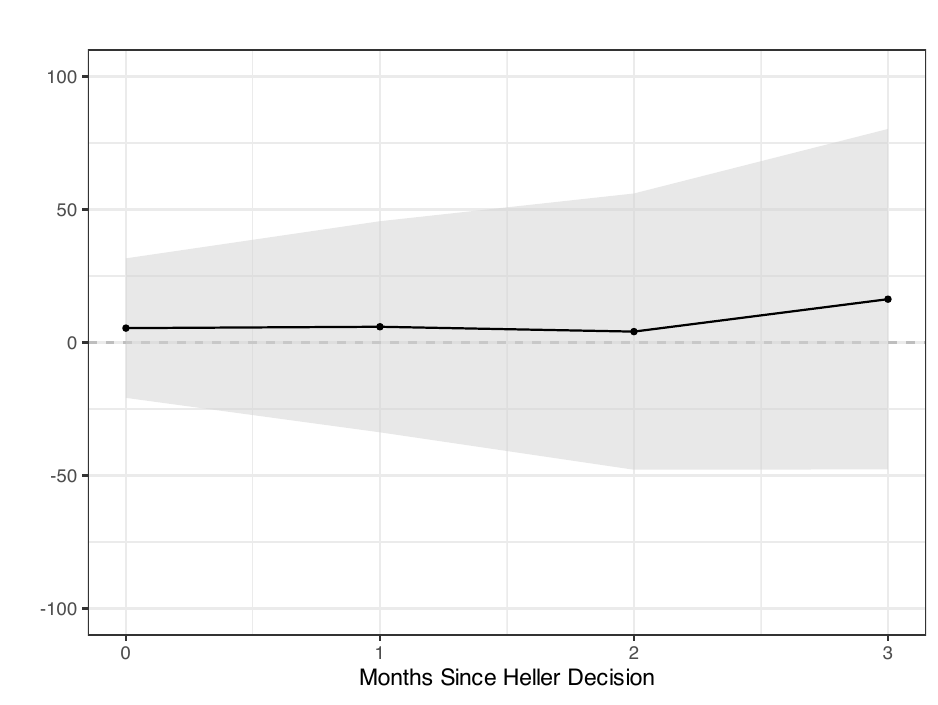}
    \caption{Cumulative National Average Effects. \textit{Note: } The plot shows the estimated national cumulative effects of the \textit{Heller} ruling on handgun background checks, discussed in Section~\ref{sec:application} of the main text.}
    \label{fig:heller_cumulative}
\end{figure}

\clearpage

\subsection{Residual Analysis}
A natural concern is that the estimated effects are an artefact of unusually noisy data around the 2008 treatment period rather than a response to the ruling. To check whether the pre-treatment fit deteriorates as the cutoff approaches, we fit a single GP to the entire series for each state and examine the residuals \citep{felton2023}. This diagnostic differs from the placebo checks, which assess out-of-sample prediction through iterative one-step-ahead forecasting; here we ask whether in-sample fit is stable across the study period.

Figure~\ref{fig:residuals_heller} shows the result. The residual variance stays stable through the pre-treatment period, with no sign of growing prediction error as June 2008 approaches. The one conspicuous exception, Louisiana in September 2005, reflects Hurricane Katrina, which made landfall in late August, rather than model failure; we read it as a substantive shock of the kind the placebo discussion flags rather than as a sign of instability. This stability supports reading the post-treatment deviations as treatment rather than noise.

\begin{figure}[hbt!]
    \centering
    \includegraphics[width=0.8\textwidth]{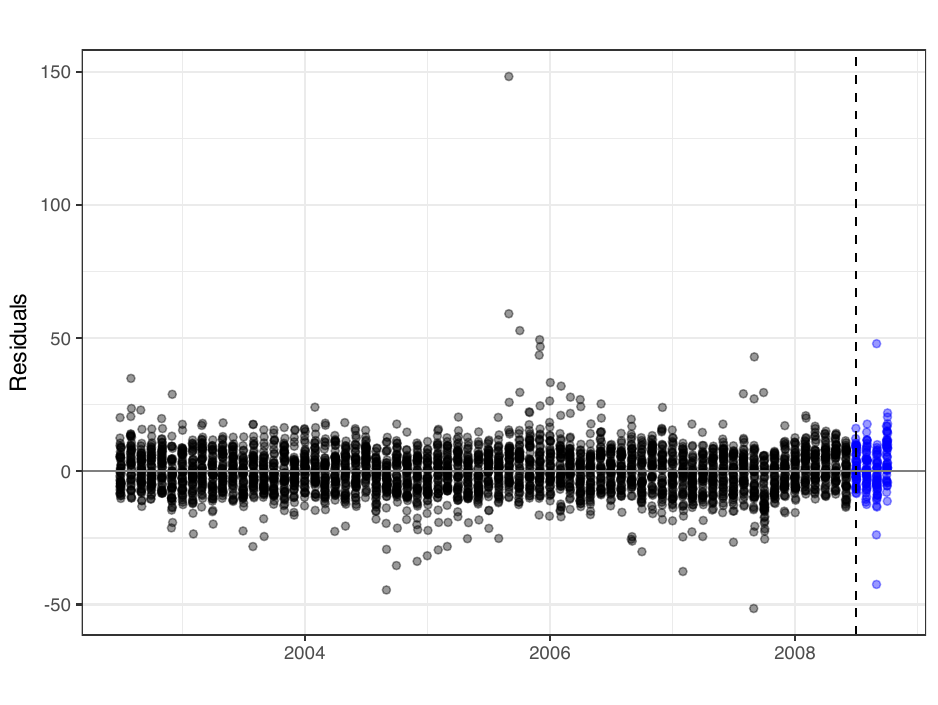}
    \caption{Residuals from Full-Period GP Fit. \textit{Note: } Residuals from a single GP model fitted to the entire time series for each state. Each point represents the ``residuals'', i.e., difference between observed and fitted values. The vertical dashed line indicates the treatment onset.}
    \label{fig:residuals_heller}
\end{figure}

\clearpage

\subsection{A Pseudo-Placebo Outcome: Long Gun Purchases}
\textit{Heller} struck down a ban on handguns, not long guns, so background checks for long guns are an outcome the ruling should not move. This makes them a placebo outcome, the cross-outcome counterpart to the temporal placebo \citep[cf.][]{felton2023}. Figure~\ref{fig:longgun} runs the same analysis for long gun checks that Figure~\ref{fig:heller_atts} runs for handguns. The state-specific effects cluster around zero with no clear geographic or political pattern, and the national effects fluctuate around zero without trend. A null response where the legal scope predicts none is what a clean placebo outcome should show, and it strengthens the causal reading of the handgun results.

\begin{figure}[hbt!]
    \centering
    \includegraphics[width=1\textwidth]{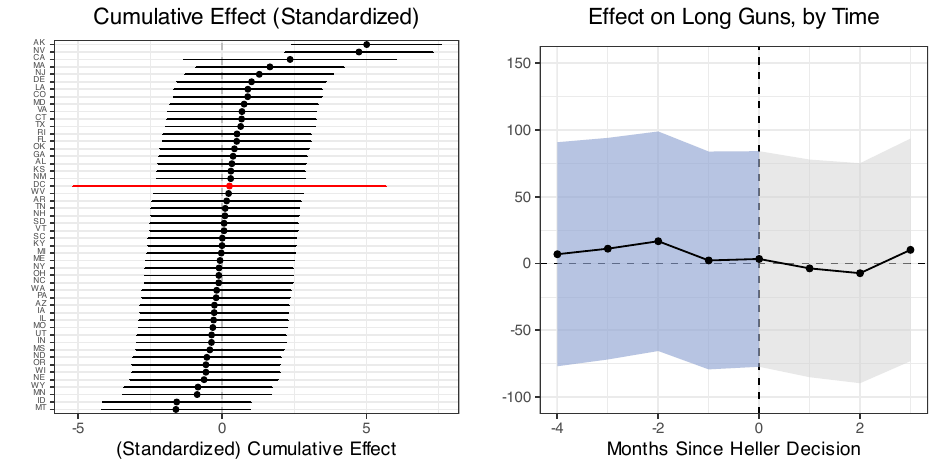}
    \caption{Effects of Heller on Long Gun Background Checks. \textit{Note: } State-specific (standardised) cumulative treatment effects (left panel) for the four-month period and national average treatment effects over time ($\hat{\bar{\tau}}_{t^*}$, right panel) for long gun background checks following \textit{Heller}. In contrast to handgun purchases, long gun checks show no systematic response to the decision, with effects centred near zero across both states and time periods.}
    \label{fig:longgun}
\end{figure}

\section{Details on Covariate Processing and Kernel}
\label{app:kernel_details}

\paragraph{Covariate scaling.} We centre the continuous covariates and the outcome and scale each to unit variance within every unit, so the kernel reads all continuous coordinates on a common scale and its behaviour does not depend on the units in which variables were recorded \citep{hainmueller2014kernel, doudchenko2016balancing}. Posterior means and intervals are then mapped back to the original outcome scale, so estimated effects and their uncertainty are reported in the natural units of the outcome. Categorical coordinates enter the same distance through a separate convention.

\paragraph{Categorical covariates.} Categorical covariates, such as the month and day-of-week indicators used in our applications, need different treatment from the continuous case \citep{johannemann2019sufficient}. Ordinary regression drops one level to avoid collinearity, but a kernel reads the covariates through distances, and dropping a level would make those distances asymmetric across categories, since the omitted level would sit at the origin and appear closer to every other level than the remaining levels are to one another. We instead one-hot encode without a reference level, writing $\mathbf{e}(\ell)$ for the indicator vector of level $\ell$; any two distinct levels are then equidistant, $\|\mathbf{e}(a) - \mathbf{e}(b)\|^2 = 2$, the symmetry appropriate for unordered categories with no natural baseline. We then rescale each indicator by $\sqrt{1/2}$, so that a mismatch on a categorical covariate contributes
\begin{equation}
    \bigl\| \sqrt{\tfrac{1}{2}}\,\mathbf{e}(a) - \sqrt{\tfrac{1}{2}}\,\mathbf{e}(b) \bigr\|^2
    = \tfrac{1}{2}\,(1 + 1) = 1,
    \qquad a \neq b,
    \label{eq:onehot_scale}
\end{equation}
the same amount as a one-standard-deviation difference in a standardised continuous covariate \citep{kpop}. Without this adjustment a categorical mismatch would weigh twice as much as such a difference, and a covariate with many levels would dominate the Euclidean distance, driving the Gaussian kernel towards zero for most pairs and distorting which observations the kernel treats as similar. The rescaling fixes the relative weight of categorical and continuous coordinates in the same spirit as the normalisation that balances the additive kernel components in Section~\ref{subsec:gp_estimation}, so that functions encoding category-specific shifts retain a moderate norm $\|g\|_{\mathcal{H}_k}$ (Appendix~\ref{app:rkhs_theory}) and the worst-case bound of Proposition~\ref{prop:worst_case} stays informative.

\paragraph{Length-scale hyperparameter $b$.} \citet{kpop} select the length-scale $b$ by maximising the variance of the off-diagonal entries of the kernel matrix. The length-scale governs how fast covariance decays with distance and so trades bias against variance, with small $b$ yielding localised kernels that track fine patterns but risk overfitting and large $b$ yielding smoother kernels that generalise but can mask local heterogeneity. We assess whether our findings depend on this rule by re-estimating all effects across a broad range of $b$ around the data-driven choice.

The findings are stable. The large positive D.C.\ effect and the flat national trend both persist across $b$ (Figure~\ref{fig:sense_b_combined}), and the state-specific distribution keeps its shape, a sharp peak at zero with D.C.\ as the outlier (Figure~\ref{fig:sens_b_state}). The localised D.C.\ effect and the national null are therefore not artefacts of hyperparameter selection.

\begin{figure}[hbt!]
    \centering
    \includegraphics[width=1\textwidth]{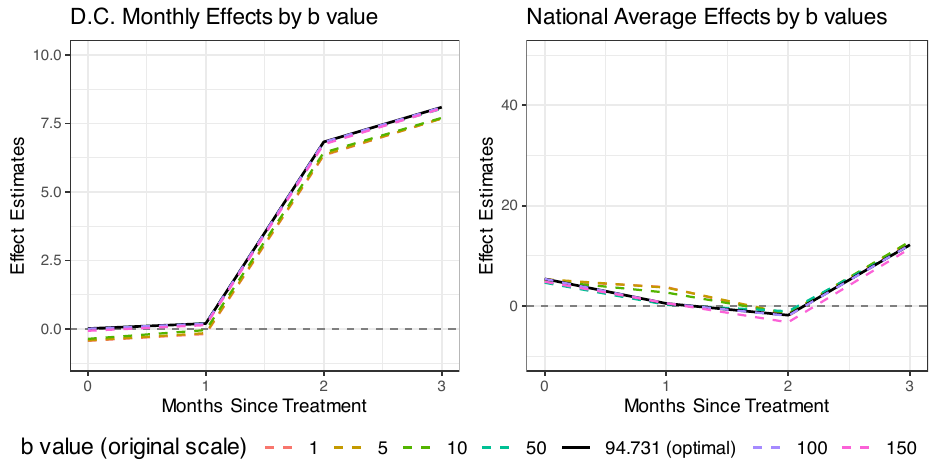}
    \caption{Sensitivity Analysis for Hyperparameter $b$. \textit{Note: } The left panel shows the estimated monthly treatment effect for Washington D.C. The right panel shows the estimated national average ($\hat{\bar{\tau}}_{t^*}$). For D.C., the large, positive effect is robust and stable across a broad range of $b$ values. For the national average, the flat, null trend is also robust. This demonstrates that our findings are not artefacts of hyperparameter selection. The solid black line in each plot indicates $\hat{b}$ chosen via variance-maximisation of the kernel matrix. Prediction intervals are omitted for visual clarity.}
    \label{fig:sense_b_combined}
\end{figure}

\begin{figure}[hbt!]
    \centering
    \includegraphics[width=0.8\textwidth]{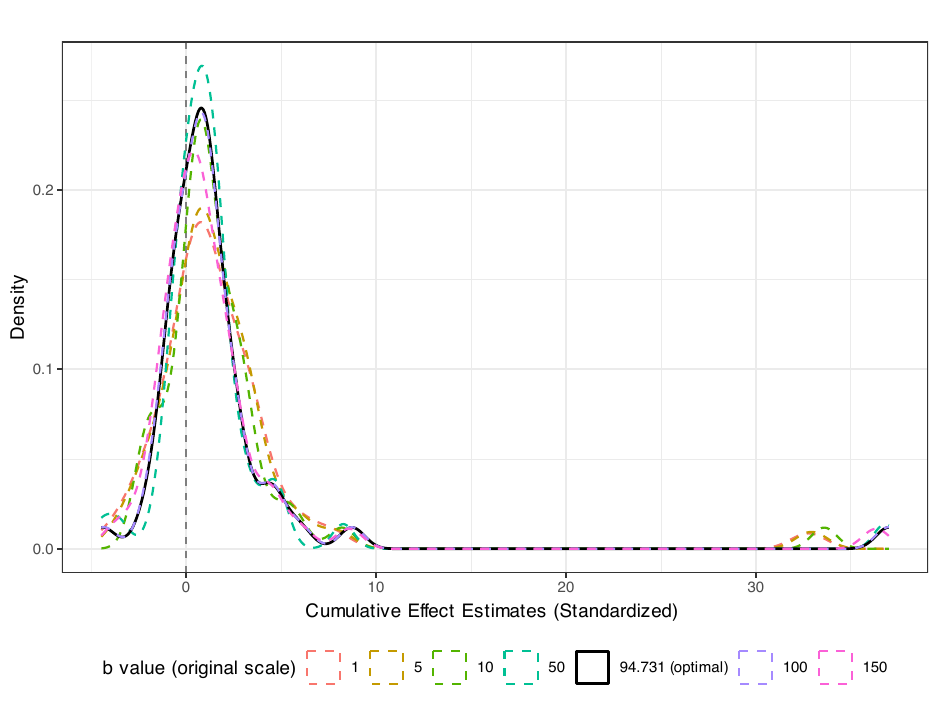}
    \caption{State-Specific Treatment Effect Estimates Across $b$ Values. \textit{Note: } Distribution of state-specific cumulative treatment effects estimated using different $b$ values given kernel choice. The optimal length-scale ($\hat{b} = 94.73$, solid black line) is compared against alternative $b$ selections (dashed coloured lines). Despite varying length-scale choices, the distributions demonstrate the stability of effect estimates across different $b$ parameter values. Intervals are dropped for improved visualisation.}
    \label{fig:sens_b_state}
\end{figure}

\paragraph{Period hyperparameter $p$.}
Periodicity corresponds to translation invariance $f(x) = f(x+p)$, so the periodic kernel reaches maximum correlation when $|x - x'|$ is a multiple of $p$, telling the model to expect similar values every $p$ units. The parameter should match the natural cycle in the data's time units, for example $p = 12$ for monthly data with annual seasonality, $p = 52$ for weekly data, and $p = 7$ for daily data with weekly cycles. When the period is unknown, it can be selected automatically by detrending the pre-treatment series with local regression, computing the periodogram by FFT, taking the frequency of maximum power, and converting it to $p$ \citep{bloomfield2004fourier}.

\clearpage

\begin{singlespacing}
\putbib[bib]
\end{singlespacing}

\end{bibunit}

\end{document}